\PassOptionsToPackage{twoside=false}{geometry}
\documentclass[manuscript, nonacm]{acmart}

\usepackage{amsthm,amsfonts}
\usepackage{mathtools}

\usepackage[linesnumbered, ruled,vlined]{algorithm2e}

\usepackage{graphicx}
\usepackage{booktabs}
\usepackage{tabularx}
\usepackage{makecell}
\usepackage{subcaption}
\usepackage{enumitem}

\usepackage{microtype}

\usepackage{url}
\usepackage{cleveref}
\theoremstyle{plain}
\newtheorem{theorem}{Theorem}
\newtheorem{lemma}{Lemma}
\newtheorem{sublemma}{Sublemma}
\newtheorem{corollary}{Corollary}

\theoremstyle{definition}
\newtheorem{definition}{Definition}

\theoremstyle{remark}
\newtheorem{remark}{Remark}

\newcommand{\In}{\textbf{In}}
\newcommand{\Out}{\textbf{Out}}
\newcommand{\InOut}{\textbf{In/Out}}

\newcommand{\Field}{{\textbf{field} }}

\SetKwProg{Class}{class}{}{end class}
\SetKwComment{tcp}{$\triangleright$ }{}
\SetKwComment{tcc}{$\triangleright$ }{}
\SetKwInOut{Input}{Input}
\SetKwInOut{Output}{Output}
\SetKwInOut{AlgDescription}{Description}
\SetKwProg{Function}{Function}{}{end}
\SetKwProg{Procedure}{Procedure}{}{end}
\SetKw{Continue}{continue}
\SetKw{True}{true}
\SetKw{False}{false}

\newcommand{\StateC}{}
\newcommand{\State}[1]{#1 \;}
\newcommand{\Comment}{\tcp*[r]}
\newcommand{\WhileC}[1]{\While(\tcc*[f]{#1})}

\newcommand{\ForAllC}[1]{\ForAll(\tcc*[f]{#1})}
\newcommand{\IfC}[1]{\If(\tcc*[f]{#1})}

\newcommand{\vdown}[1]{{#1}_{\downarrow}}
\newcommand{\vup}[1]{{#1}_{\uparrow}}
\newcommand{\vbot}[1]{{#1}_{\bot}} 
\newcommand{\bigO}{\mathcal{O}}
\newcommand{\Lfixed}{L_{\mathrm{fixed}}}

\newcommand{\qacc}{q_{\text{acc}}} 
\newcommand{\qrej}{q_{\text{rej}}}

\newcommand{\IPrec}{\mathrm{IPrec}}
\newcommand{\ISucc}{\mathrm{ISucc}}
\newcommand{\Prev}{\mathrm{Prev}}
\newcommand{\Next}{\mathrm{Next}}
\newcommand{\nextOf}{\mathrm{next}}

\newcommand{\ipred}{\mathrm{ipred}}
\newcommand{\isucc}{\mathrm{isucc}}
\renewcommand{\state}{\mathrm{state}} 
\newcommand{\tier}{\mathrm{tier}}
\newcommand{\nextIndex}{\mathrm{next\_index}}
\newcommand{\nextState}{\mathrm{next\_state}}
\newcommand{\lastState}{\mathrm{last\_state}}
\newcommand{\lastSymbol}{\mathrm{last\_symbol}}
\newcommand{\dir}{\mathrm{dir}}
\newcommand{\NIL}{\texttt{NIL}}
\newcommand{\timeOf}{\mathrm{time}}
\newcommand{\indexOf}{\mathrm{index}}
\renewcommand{\output}{\mathrm{output}}

\renewcommand{\symbol}{\mathrm{symbol}}

\begin{document}

\begin{abstract}
The $\mathsf{P}$ versus $\mathsf{NP}$ problem asks whether every language verifiable in polynomial time can also be decided in deterministic polynomial time. 
In this paper, we present a constructive proof that $\mathsf{P} = \mathsf{NP}$ by introducing a universal, 
graph-based deterministic framework applicable to all $\mathsf{NP}$ problems without requiring reduction to an $\mathsf{NP}$-complete problem. 
We model computational transitions as edges within a unified graph structure, where edges correspond to the steps of a deterministic verifier Turing machine for all possible certificates. 
Due to the overlap of edges among computation paths, the total cardinality of the edge set remains polynomially bounded. 
 Furthermore, by employing a certificate-oblivious verifier Turing machine—whose head movement is independent of the certificate contents—we force all computation paths to align their edge transitions.
A key feature of our approach is that each extension step enforces global consistency via a local infeasibility trimming tool. 
This mechanism systematically preserves valid $\mathsf{NP}$ paths that lead to the target edge under polynomial verification, 
ensuring the graph remains globally feasible at every stage without explicit enumeration. 
This represents a paradigm shift from searching over exponential certificates to the incremental extension of verified edges. 
Since our construction decides all $\mathsf{NP}$ problems in deterministic polynomial time, it provides a direct resolution to the $\mathsf{P}$ versus $\mathsf{NP}$ question 
and demonstrates that every $\mathsf{NP}$ problem can be solved in deterministic polynomial time via incremental graph edge extension.
\end{abstract}



\title{Graph-Based Deterministic Polynomial Framework for NP Problems}

\author{Changryeol Lee}
\authornote{The core framework was initiated at Yonsei University; further development and completion were conducted by the author as an Independent Researcher.}
\affiliation{%
  \institution{Independent Researcher}
  \country{Republic of Korea}
}
\email{changryeol@kaist.ac.kr}

\settopmatter{printacmref=false} 
\setcopyright{none}
\renewcommand\footnotetextcopyrightpermission[1]{} 
\pagestyle{plain} 

\acmConference[]{}{}{}
\acmDOI{}
\acmISBN{}
\acmBooktitle{}

\maketitle

\tableofcontents

\clearpage

\section{Introduction}
The question of whether $\mathrm{P}$ equals $\mathrm{NP}$ remains one of the most fundamental open problems in theoretical computer science.
It asks whether every problem in $\mathsf{NP}$, whose membership can be verified in polynomial time given a suitable certificate, can also be decided by a deterministic Turing machine in polynomial time.
A resolution of this question would have far-reaching consequences for cryptography, optimization, artificial intelligence, and formal verification.
Despite extensive research over several decades, no polynomial-time algorithm has been found for any NP-complete problem, nor has it been proven that such algorithms cannot exist.

To address these challenges, our framework departs from restricted methodologies. Unlike relativizing or algebrizing approaches that rely on the oracle-independent analysis of NTMs, we utilize a constructive algorithm involving a verifier DTM that exploits the specific transition structures of deterministic computation. Furthermore, we do not assume a priori super-polynomial bound identified by the natural proofs barrier. Instead, we focus on a constructive polynomial-time decidability. These aspects are discussed in detail in \cref{sec:related_works}.

In this paper, we pursue such a perspective by introducing a \textbf{graph-based deterministic polynomial framework} designed to resolve the existential certificate structure inherent in $\mathsf{NP}$ decision problems. 
Rather than explicitly enumerating and checking exponentially many certificates, our approach represents the entire certificate space within a unified, polynomially bounded computation graph. 
This is achieved by employing a verifier TM directly—a component essential to any $\mathsf{NP}$ problem—thereby bypassing the need for reduction to other $\mathsf{NP}$-complete problems. 
Furthermore, to guarantee global consistency, we use a grid-aligned footmark structure derived from a certificate-oblivious Turing machine, which can be mechanically constructed from any general verifier with only polynomial overhead.

A key innovation in our model lies in the augmented structure of the graph nodes. Unlike intuitive models where a node represents only the tape position, symbol, and current state, each node in our framework is defined as a multifaceted 6-tuple. 
It incorporates the transition count (tier), the state and symbol, and crucially, the former state and former tape symbol for each specific cell. 
By embedding these historical parameters directly into the node structure, the graph inherently maintains execution consistency, 
allowing the deterministic simulation to trace causal dependencies across the entire computation space without loss of historical context within the polynomially bounded graph.

The central idea is to incrementally construct a computation graph that integrates transitions corresponding to all possible certificates. Because these paths share a common origin and extensively overlap, the total number of edges in this unified structure remains strictly polynomial. This approach triggers a paradigm shift: instead of the existential verification of exponential certificates, we perform deterministic structural analysis on a polynomial-sized graph manifold. By incrementally extending valid edges through local consistency checks, our framework resolves $\mathsf{NP}$ problems in deterministic polynomial time, effectively replacing global exponential search with a local, structure-preserving extension process

We formalize this idea through the concept of a \emph{Feasible Graph}, which maintains consistent paths from the start node to the edge under verification while trimming inconsistent or unnecessary edges. 
A key feature of our approach is that \textbf{each extension step enforces global consistency via a local inconsistency trimming tool}. 
To guarantee global consistency, we use a grid-aligned footmark structure derived from a certificate-oblivious Turing machine, which can be mechanically constructed from any general verifier with only polynomial overhead.
This structure organizes computation paths into a spatial grid, where each coordinates of the vertices are synchronized across all possible certificate branches. 
By aligning these coordinates, we effectively project the nondeterministic branching of certificates onto a deterministic, lattice-like graph, ensuring that local transition consistency translates directly into global path validity. 
This structure underpins our \emph{total collapse} property: if a certain essential edge—one belonging to all valid computation paths leading to the target—is missing, the feasible graph immediately collapses to an empty set

This mechanism ensures that the graph remains globally consistent at every stage by systematically preserving consistent $\mathsf{NP}$ paths while eliminating structural unnecessary edges.
The method for detecting these redundant edges involves the sequential elimination of implausible edges until the feasible graph collapses. This collapse identifies the \emph{essential edges} necessary for computation; 
by identifying these critical components, the algorithm can locate and prune such edges that do not contribute to a valid path reaching the target edge.

By reducing nondeterministic certificate verification to deterministic structural analysis and trimming, the framework yields a deterministic polynomial-time algorithm for $\mathsf{NP}$ decision problems. This provides a direct and constructive resolution of the $\mathsf{P} = \mathsf{NP}$ question.
To achieve this, we organize our proof into a hierarchical framework. \Cref{sec:roadmap_highlevel_proof} provides a high-level strategy of our simulation, detailing how we transition from existential search to structural analysis through the lens of Deterministic Polynomial Complexity and Verification Targets. The subsequent sections then provide related work, the formal definitions and rigorous proofs for each layer of this framework.


\section{Roadmap and High-Level Proof} \label{sec:roadmap_highlevel_proof}

This section provides a structural overview and a high-level synthesis of the entire proof architecture. We first establish the global roadmap that governs our derivation, followed by the core foundations of our computational model. 
To ensure a clear conceptual grasp of the framework’s architecture, we present the core logic in a top-down manner, focusing on structural mechanisms and deterministic convergence. 
The proofs in this section are intended as intuitive and constructive demonstrations that elucidate the functional interplay of the system. 
For the rigorous, bottom-up formalization—including the exhaustive derivation of each lemma and the inductive verification of the algorithm’s correctness—the reader is referred to the detailed proofs in the corresponding sections.

\subsection{Proof Strategy and Roadmap}
\label{subsec:roadmap}

The fundamental challenge in proving $P=NP$ lies in deterministically deciding the existence of an accepting certificate without an exponential brute-force search. Our approach achieves this by transforming the search problem into a structural analysis of a polynomially-bounded graph.

\subsubsection{The Core Philosophy: From Search to Structure}
\label{subsec:philosophy}

The central paradigm shift of this proof is the transition from the \textbf{existential verification} of exponential certificates to the \textbf{deterministic analysis} of a polynomial-sized graph manifold. 

Traditionally, deciding a language $\mathcal{L} \in \textsf{NP}$ requires verifying a witness from an immense search space of $\Sigma^{p(n)}$. Our strategy bypasses this exponential barrier by constructing a \textbf{Universal Computation Graph} $G_U$. Instead of "searching" for a single valid certificate, we embed the "footmarks" of \textit{all} potential computation walks into a unified structure. By doing so, the problem of finding a certificate is reduced to identifying a globally consistent "signal" (a feasible walk) amidst a background of "structural noise" (obsolete or orphaned edges). 

This reduction ensures that the non-deterministic search is replaced by a deterministic pruning process, where the structural integrity of the graph serves as the decider for the existence of a valid computation.

\subsubsection{The Foundation: 6-Tuple Model and Polynomial Universe}
The bedrock of our strategy is the definition of a finite, polynomially-bounded configuration space that encompasses all possible execution traces.
\begin{itemize}
    \item \textbf{Computation Model:} By encoding each node as 6-tuples $(index, tier, state, symbol, last\_state, last\_symbol)$, we ensure that the total number of nodes representing all potential execution paths is strictly bounded by a polynomial complexity $p(n)$ (\Cref{subsec:comp_graph}).
    \item \textbf{Universal Footmarks:} As will be proven in Lemma~\ref{lem:poly-bounded-graph}, the union of all computation walks across the entire certificate space—the \textbf{Footmarks}—is bounded by $\bigO(p(n)^3)$. This provides the mathematical basis for embedding the nondeterministic search space into a polynomial-sized manifold.
\end{itemize}

\subsubsection{Layer 1: The Utility Layer (The Feasible Graph)}
Building on the Foundation, this layer instantiates the operational graph and enforces structural integrity.
\begin{itemize}
    \item \textbf{Feasible Graph Construction:} Utilizing the Dynamic Computation Graph (\Cref{subsec:dynamic_graph}), we construct a \textbf{Feasible Graph} $G_f$ by removing all edges that violate local consistency rules.
    \item \textbf{Feasibility Maintenance:} This layer ensures the graph remains a universal container for valid certificates. Per Lemma~\ref{lem:preservation_of_feasible_walks}, it preserves every valid \textbf{computation walk} to the target extension candidate edge while providing the first level of structural refinement necessary for global analysis.
	\item \textbf{Total Collapse:} This layer also ensures that the graph collapses into an empty structure whenever a certain essential edge is eliminated. This mechanism operates on the grid-aligned structure derived from a certificate-oblivious Turing machine, where head movement is independent of the certificate content once the problem instance is fixed.
\end{itemize}

\subsubsection{Layer 2: The Strategic Layer (Verification of Global Consistency)}
This layer (Section~\ref{sec:walk_verification}) bridges the gap between local consistency and global validity. To ensure that the target extension candidate edge $e_t$ is not merely a locally valid transition but a component of a globally consistent path, we employ a \textbf{Cascading Collapse Mechanism} to filter or certify the global viability of edges:

\begin{enumerate}
    \item \textbf{Selection and Candidate Inclusion:} From the current feasible graph $G_{f}$, we identify a potential computation walk $W_p$. If $W_p$ contains the \textbf{verification target edge} $e_t$, this walk serves as a witness for the candidate edge's inclusion in the final graph.
    
    \item \textbf{Trial Elimination of Implausible Edges:} To verify the necessity and stability of $e_t$, the algorithm simulates the removal of "implausible" edges—the first splitting edge appear in compuation walk cannot reach to the target edge. We then observe whether the graph undergoes a \textit{Total Collapse} of paths dependent on those edges.
    
    \item \textbf{Criticality Identification:} By analyzing the structural response to these simulated removals, the algorithm distinguishes between \textbf{computing-targeted walks} (globally valid trajectories) and \textbf{computing-futile noise} (locally consistent segments that cannot form a complete, halting execution).
    
    \item \textbf{Iterative Redundancy Removal:} Through the systematic elimination of edges that do not critically contribute to any valid computation walk containing $e_t$, the algorithm certifies whether the target edge truly belongs to the \textbf{Universal Footmarks} of an accepting computation.
\end{enumerate}

\subsubsection{Layer 3: The Application Layer (Deterministic Decision Procedure)}
The final layer (Section~\ref{sec:np_is_p}) executes the deterministic simulation and final decision.
\begin{itemize}
    \item \textbf{Boundary-Guided Extension:} The algorithm incrementally expands a verified set $H$ across the universal footmark graph by testing each candidate edge's ability to support a valid walk. This deterministic traversal replaces the traditional non-deterministic search.
    \item \textbf{Edge Promotion:} An edge is "promoted" to the verified history $H$ only after its global necessity is certified by the Strategic Layer's validity tests.
\end{itemize}

\subsubsection{Convergence and Complexity}
As established, the total number of edges in the universal graph is bounded by $|E| = \bigO(p(n)^3)$. Since each iteration of the pruning cycle identifies and removes at least one structural redundancy or confirms an edge's validity, the algorithm is guaranteed to converge in polynomial time.

The process concludes with a definitive deterministic decision:
\begin{itemize}
    \item \textbf{Accept:} If any promoted edge reaches a node in the accept state $q_{acc}$, certifying the existence of at least one valid computation walk.
    \item \textbf{Reject:} If the graph stabilizes such that no further extensions are possible and no verified path to $q_{acc}$ exists.
\end{itemize}

By transforming the exponential existential search for a certificate into a deterministic structural analysis of a polynomially-bounded graph, we demonstrate that any language $\mathcal{L} \in \textsf{NP}$ is decidable in deterministic polynomial time, establishing $P = NP$.

\paragraph{\textbf{Summary of Proof Flow:}}
$\text{Foundation}$ (Computitaion Graph) $\xrightarrow{\text{Layer 1}}$ $G_f$ (Local Feasibility) $\xrightarrow{\text{Layer 2}}$ Verified $e_t$ (Global Consistency) $\xrightarrow{\text{Layer 3}}$ $H$ (Deterministic Footmarks) $\implies \mathsf{P} = \mathsf{NP}$.

\subsubsection{Deterministic Complexity Hierarchy and Polynomial Convergence}
\label{subsec:complexity_hierarchy}

The computational efficiency of the framework is rooted in its nested deterministic structure. Unlike nondeterministic branching, our algorithm operates through four distinct layers of polynomial iterations. The total time complexity is bounded by the product of these layers, where $|E| = \bigO(p(n)^3)$ denotes the total number of edges in the universal computation graph.

\begin{itemize}
    \item \textbf{Layer 1: Footmarks Extension Loop}
    \begin{enumerate}
        \item \textit{Extension Loop ($\times |E|$):} Drives the iterative growth of the verified set $H$.
        \item \textit{Candidate Validation Loop ($\times |E|$):} Identifies potential candidates and validates them for promotion to $H$.
    \end{enumerate}

    \item \textbf{Layer 2: Global Consistency and Strategic Pruning Loop}
    \begin{enumerate}
        \item \textit{Elimination of Futile edges ($\times |E|$):} Removal of confirmed redundant/futile edges.
        \item \textit{Trial Detection ($\times |E|$):} Secondary iteration for trial eliminationsto to conduct the validity test for the target edge $e_t$.
    \end{enumerate}

    \item \textbf{Layer 3: Feasible Graph Construction Loop}
    \begin{enumerate}
        \item \textit{Convergence Loop ($\times |E|$):} This loop repeats the sweep process until the state of the feasible graph remains unchanged, ensuring the structural stability of the entire manifold and the finality of the pruned results.
        \item \textit{Bidirectional Sweep ($\times |E|$):} The systematic construction of the graph by collecting locally consistent edges. This consists of \textit{Horizontal Sweeps} (to ensure reachability and connectivity across tiers) and \textit{Vertical Sweeps} (to ensure history consistency and causal integrity across configurations).
    \end{enumerate}

    \item \textbf{Layer 4: Data Structure and Primitive Operations ($\times \text{t}(n) < |E|$):} 
    Handles low-level manipulation of the 6-tuple space and adjacency lists.
\end{itemize}

By expressing the total execution time as a nested product of these hierarchical layers, the complexity $T(n)$ can be formally upper-bounded. Given that each layer introduces nested iterations over $|E|$, the cumulative complexity is:
\[
T(n) = \left(|E|_{\text{ext}} \times |E|_{cand} \right) \times \left( |E|_{\text{elim}} \times |E|_{\text{detect}} \right) \times \left( |E|_{\text{conv}} \times |E|_{\text{sweep}} \right) \times \text{t}(n) = \bigO(|E|^6 \cdot \text{t}(n)).
\]
Substituting $|E| = \bigO(p(n)^3)$, we obtain $T(n) = \bigO(p(n)^{18} \cdot \text{t}(n)) \approx \bigO(p(n)^{20})$. This remains strictly within deterministic polynomial bounds, establishing $\mathsf{P} = \mathsf{NP}$.

\subsection{Polynomial Bounded Computation Model} \label{sec:polynomial_bounded_model}
We project the dynamic execution of verifier $M$ onto a static graph $G=(V, E)$. To ensure that the entire certificate space $\Sigma^{q(n)}$ can be handled deterministically, we define the fundamental unit of our computation graph as a Computation Node. The formal definitions and complete structural proofs regarding this section are provided in Appendix \cref{subsec:comp_graph}.

\begin{definition}[Computation Node]
A node is a 6-tuple $v = (i, q, \sigma, \vdown{q}, \vdown{\sigma}, t)$, where $(q, \sigma)$ denotes the current state and symbol at tape index $i$ and tier $t$, and $(\vdown{q}, \vdown{\sigma})$ represents the preceding state and symbol at the same cell. For $t=0$, we define $\vdown{q} = \bot$ and $\vdown{\sigma} = \bot$, where $\bot$ means none.
\end{definition}

The computation graph forms a lattice bounded by its height (maximum tier) and width (index span). A \textbf{computation walk} $W = (v_0, v_1, \dots, v_n)$ is a sequence of computation nodes representing the step-by-step execution of a Turing machine $M$. Formally, for each step $k$, the node $v_k$ represents the 6-tuple configuration of the tape cell pointed to by the head after $k$ transitions. Note that a computation walk is a simple path without cycles in the graph due to the monotonic increment of the node tier at each cell.
The power of this representation lies in its structural economy. While the number of certificates is exponential, the union of all possible computation walks—which we term the Footmarks Graph $F(\mathcal{W})$—remains strictly within polynomial in size.

\begin{lemma}[Structural Bounds of Footmarks]
For a verifier $M$, the total number of vertices and edges in the footmarks $F(\mathcal{W})$ are bounded by $\bigO(p(n)^k)$, where $\mathcal{W}$ denotes the set of any possible computation walks for $M$ and $p(n)$ is a polynomial in the input length $n$.
\end{lemma}
\begin{proof}
By definition, the verifier $M$ terminates within polynomial $p(n)$ steps. Each node is constrained by a fixed coordinate $(i, \text{tier})$, where $i$ is the cell index and $\text{tier}$ represents the transition count at that cell. Since the head visits at most $p(n)$ cells and each cell can undergo at most $p(n)$ transitions, the width  and height of the graph are both $\bigO(p(n))$. Given that $|Q|$ and $|\Gamma|$ are constants, the total configuration space per coordinate remains $\bigO(1)$. Thus, $|V(G)|=\bigO(p(n)^2), |E(G)|=\bigO(|V(G)|^2)$.
\end{proof}
This polynomial bound on the graph size prevents exponential explosion and establishes the feasibility of deterministic trimming. Besides these polynomial bounds, the fundamental structural constraint of our model is \textbf{spatially localized continuity}. We use the prefix \textit{index-} to denote relations that occur within the same spatial coordinate (the same Pole) across the temporal axis (tiers). To formalize the spatial segments between Poles, we introduce the concept of an edge slice.
An \textbf{Edge Slice} $E_i$ is defined as the set of all edges incident to nodes at indices $i$ and $i+1$, representing the spatial interval between Pole $i$ and Pole $i+1$. All edges $e \in E_i$ are associated with the same spatial interval but are distributed across various tiers $k \in \{0, \dots, p(n)\}$.

For a node $u \in V_{i, t}$ located on Pole $i$ (the $i$-th cell at tier $t$), we distinguish between path-specific and structural relations. Within any specific computation walk $W$, the \textbf{index-predecessor} of $u$ is the node at the same index in the immediate lower tier ($t-1$), while its \textbf{index-successor} is the corresponding node in the immediate higher tier ($t+1$). 

Extending this to the computation graph structure, an \textbf{index-precedent} of $u$ is any node $\vdown{u}$ that can serve as an index-predecessor in a computation walk; specifically, in addition to the index and tier conditions, the last state and the last symbol of $u$ must match the state and the symbol of $\vdown{u}$, respectively. Conversely, an \textbf{index-succedent} of $u$ is any node $\vup{u}$ that can function as an index-successor where, in addition to $u$ belonging to the index-precedents of $\vup{u}$, the transition output at $u$ matches the symbol at $\vup{u}$.

Similarly, for an edge $e$ residing in Edge Slice $E_i$, we define its temporal and structural relations as follows. Within any specific computation walk $W$, the \textbf{index-predecessor} of $e$ is the former edge in the same edge slice, while its \textbf{index-successor} is the latter edge. Extending this to the computation graph structure, an \textbf{index-precedent} of $e$ is any edge $\vdown{e}$ that can serve as an index-predecessor in a computation walk; due to the direction change (folding) in the pole containing the tail of $e$, it suffices that the tail of $e$ belongs to the index-precedent of the head of $\vdown{e}$ and the head of $\vdown{e}$ forms a folding index-precedent chain starting from tail of $e$. Conversely, an \textbf{index-succedent} of $e$ is any edge $\vup{e}$ that can serve as an index-successor in a computation walk; due to the direction change in the pole containing the head of $e$, it suffices that the head of $e$ belongs to the index-succedent of the tail of $\vup{e}$ and the tail of $\vup{e}$ belongs to a folding index-succedent chain from the head of $e$.

Since any computation walk corresponds to a valid transition sequence, it must satisfy both \textbf{Vertical Edge Integrity}---requiring the existence of an index-precedent ($\IPrec(e)$) and an index-succedent ($\ISucc(e)$) where applicable---and \textbf{Horizontal Integrity}, which ensures adjacency to neighboring edges unless the edge is at the initial or final boundary. These structural constraints, which form the mechanical basis for our Feasible Graph construction, are further refined in \cref{sec:trimming_tool_feasible_graph}.

\subsection{Edge Extension for All Computation Paths (The Top Layer)} \label{sec:edge_extension_for_footmarks}

The top layer of our algorithm is the \textbf{Incremental Extension} of the polynomial-sized footmarks of all computation walks across every possible certificate. This stage serves as the global driver for the step-by-step construction of the entire machine's transition trajectory, ensuring that every newly added edge is globally verified through our local inconsistency trimming tool.

The extension algorithm operates on an $\mathsf{NP}$ problem $X$, a certificate-oblivious verifier $M$, and a certificate of polynomial length $m$. Since all $\mathsf{NP}$ problems are defined by a verifier and polynomial-length certificates and the verifier can be converted to oblivious TM (see appendix \cref{sec:computation_theory}), this approach remains generally applicable. 

To maintain polynomial efficiency, the algorithm replaces exhaustive certificate space searching with an incremental extension of the footmarks graph $G_{\text{footmark}}$. We identify a set of extension \textbf{Candidate Edges} ($E_{\text{cand}}$) based on two structural criteria:

(i) \textbf{Boundary Edge Condition}: the edge $e$ must be incident to a node already in $G_{\text{footmark}}$ but not yet integrated into the verified component, while either possessing a valid \textit{index-precedent edge} in $G_{\text{footmark}}$ \textbf{or} acting as a \textbf{floor edge}---a \textbf{frontier edge} of its respective edge slice---to ensure vertical consistency; and 
(ii) \textbf{Certificate Area Enumeration}: within the designated tape area for certificate strings (representing branching computation walks), we generate candidates for \textbf{all possible symbols} $\sigma \in \Sigma$ at the tier $0$ node of each pole (each cell) to serve as new floor edges. \\
This mechanism effectively transforms the exponential enumeration of certificates into a deterministic, parallel structural edge extension.

\begin{algorithm}[!ht] 
    \SetKwFunction{VerificationWalkToTarget}{VerificationWalkToTarget}
    \DontPrintSemicolon
    
    \caption{Deterministic Edge Extension and Halting} \label{higalg:incremental_extension_of_footmarks}
    
    \Input{Verifier TM $M$, Problem instance $X$, Length of certificate string $m$.}
    \Output{Accept (Yes) or Reject (No).}
    
    \BlankLine
    Initialize $G_{\text{footmark}}$ with the verifier Turing machine $M$. \;
    Extend $G_{\text{footmark}}$ the initial edge with the initial state of $M$ and the first symbol of $X$ \;
    
    \While{True}{
        Identify Candidate Boundary Edges($E_{cand}$)\;
        Let  $existsNewEdge \gets \False$ \;

        \ForEach{edge $e \in E_{cand}$}{
            \If{\VerificationWalkToTarget($G_{\text{footmark}}+e, e$) = \True} {
                $G_{\text{footmark}} \gets G_{\text{footmark}} \cup \{e\}$; Set $existsNewEdge \gets \True$ \;
            }
        }
        
        \If{$\exists e \in G_{\text{footmark}}$ such that $e$ reaches $q_{acc}$}{
            \Return{\textsf{Accept (Yes)}}\;
        }
        \ElseIf(\tcc*[f]{No new edges were added to $G_{\text{footmark}}$}){$existsNewEdge=\False$}{
            \Return{\textsf{Reject (No)}}\;
        }
    }
\end{algorithm}

While this iterative extension captures all potential computation paths, its validity hinges on a robust filtering process capable of determining the existence of a valid computation walk that contains the target candidate edge. The formal mechanism is explained in Appendix \cref{sec:np_is_p}.

\begin{lemma}[Correctness of Incremental Footmarks Construction] \label{highlem:footmark_construction_correctness}
Assuming that the procedure \VerificationWalkToTarget{} correctly determines the existence of a computation walk to any target edge $e$ in polynomial time, the \cref{higalg:incremental_extension_of_footmarks} correctly determines the existence of an accepting certificate. Specifically, the algorithm halts and returns \textsf{'Accept (Yes)'} if and only if there exists a verified edge reaching an accepting state $q_{\text{acc}}$, and it halts and returns \textsf{'Reject (No)'} if and only if no further edge extensions to the footmark graph $G_{\text{footmark}}$ are possible, implying the absence of any computation walk reaching an accepting state. The entire process terminates within polynomial time.
\end{lemma}

\begin{proof}
The proof proceeds by analyzing the structural completeness of the candidate set and the monotonic convergence of the iterative extension process.
\begin{itemize}
\item \textbf{Completeness of Candidate Selection} 
The set of candidate edges $E_{cand}$ is constructed to be exhaustive within the spacetime constraints. For any existing node in the footmark graph, $E_{cand}$ includes all adjacent edges that are floor edges or possess a valid \textit{index-precedent} (former transition) within the same \textit{edge slice}. Specifically, at the \textit{floor edges} of the \textit{certificate area}, the algorithm includes transitions for \textbf{all possible tape symbols} $\sigma \in \Sigma$. This ensures that $E_{cand}$ necessarily contains the edge incident to the tier $0$ node.

\item \textbf{Correctness of Verification and Soundness} 
The \VerificationWalkToTarget{} procedure correctly identifies the existence of a globally consistent computation walk to any given candidate edge in polynomial time, by the premise.  Consequently, if a valid computation path to an accepting state $q_{acc}$ exists, the algorithm will eventually incorporate the corresponding edge into $G_{footmark}$ and return \textsf{Yes}. Conversely, if no such path exists—meaning all potential trajectories terminate in rejection and no further extension is possible—no edge reaching $q_{acc}$ will ever pass the verification sieve, and the algorithm will correctly conclude with \textsf{No}.

\item \textbf{Polynomial Time Complexity and Convergence} 
The total number of potential edges in the spacetime lattice is bounded by $|E| = \bigO(p(n)^k)$. Since each iteration of the extension adds at least one verified edge or terminates, the number of boundary edges is at most $|E|$. Given that the size of the footmarks is polynomial and \VerificationWalkToTarget{} operates in polynomial time, the entire algorithm terminates within \textbf{deterministic polynomial time}. 
\end{itemize}
When no further edges can be added to $G_{footmark}$ and $q_{acc}$ has not been reached, the algorithm correctly returns \textsf{'No'}, signifying that no valid certificate exists for the input.
\end{proof}

\subsection{Verification of Computation Path to the Target Edge (The Filter)} \label{sec:verification_path_to_candidate_edge}

To maintain the integrity of the extension, every candidate edge must pass through a verification mechanism. This sieve induces a structural collapse of the graph, isolating the unique, deterministic computation walk to the \textit{verification target edge} $e_t$, which serves as the extension candidate. For this process, we utilize the \textit{trimming tool} described in the subsequent section, which eliminates locally inconsistent edges while preserving any computation walk to the target; conversely, the graph collapses to empty if some essential edges are missing. The formal proof and detailed explanation are provided in Appendix \cref{sec:walk_verification}.

To determine the existence of a computation walk to the verification target edge $e_t$, we categorize computation walks and edges based on their structural contribution. A \textbf{computing-targeted walk} is a valid computation path that successfully reaches $e_t$, whereas a \textbf{computing-futile walk} is a maximal comutation walk that does not contain $e_t$. Correspondingly, we define edges within these structures: an edge is \textbf{computing-effective} if it belongs to at least one computing-targeted walk. In contrast, a \textbf{computing-futile edge} is one that belongs exclusively to computing-futile walks, thus contributing nothing to the reachability of the target. Furthermore, a computing-effective edge is considered a \textbf{computing-redundant edge} if its removal does not eliminate all computing-targeted walks—that is, alternative computation paths to $e_t$ remain available within the graph.

The verification proceeds by systematically identifying and removing edges that do not contribute to a valid computation walk. This verification process is structured into two functional layers: the \textbf{lower layer} identifies \textit{computing-redundant} or \textit{computing-futile} edges, while the \textbf{higher layer} executes the iterative removal of these detected edges. If the mechanism successfully confirms a computing-targeted walk to $e_t$, the algorithm returns \textbf{True}; if no such path exists or the graph is reduced to an empty set, it returns \textbf{False}.

\begin{algorithm}[!ht]
    \DontPrintSemicolon
    \SetKwFunction{FeasibleGraphTrim}{FeasibleGraphTrim}
    \SetKwFunction{DetectRedundantFutileTargetEdge}{DetectRedundantFutileTargetEdge}    

    \caption{Verifying the existence of valid computation walks to the target edge}
    
\Input{The target edge $e_t$ and  a computation graph $G$(footmarks augmented by $e_t$)}
\Output{\True if exists, otherwise \False}
\BlankLine
\Function(\tcc*[f]{Detection (Lower Layer)}){\DetectRedundantFutileTargetEdge{$G', e_t$}}{
\While{$G'$ is not empty} {
    Select an arbitrary computation walk $W_{p} \subseteq G'$\;
    \If {$W_{p}$ is a computing-targeted walk (contains $e_t$)}{
        \Return $e_t$ \tcp*{Target confirmed}
    }
    Let $G'_{pre} \gets G'$ \tcp*{\textbf{Record} current state as $G'_{pre}$}
    Identify the first splitting edge $e_s \in W_{p}$ (or the final edge if no splitting exists)\;
    Let $G' \gets \FeasibleGraphTrim(G'-e_s, \{e_t\})$ \tcp*{Pruning walk $W_{p}$}
}
Extend $G'$ with the next edges of the final edges $E_o$ of computing-futile walks\;
Let $G' \gets \FeasibleGraphTrim(G'_{pre}-e_s, E_o)$ \tcp*{Preserve computing-futile walks}
\lIf(\tcc*[f]{No computing-futile walks exist}){$G'$ is empty} {
    \Return \textsc{Nil} 
} \Else{
    Select the first edge $e$ on a computing-futile walk in $G'_{pre}$ such that $e \notin W_{p}$ \;
    \Return $e$ \tcp*[f]{Identified as redundant or futile}
}
}
\BlankLine

\Function(\tcc*[f]{Main Loop (Higher Layer)}){\VerificationWalkToTarget{$G, e_t$}} {
\While{$G$ is not empty} {
    Let $e \gets$ \DetectRedundantFutileTargetEdge($G, e_t$)\;
    \lIf{$e = e_t$}{\Return \True}
    \lElseIf(\tcc*[f]{No verification-target walks}){$e = \textsc{Nil}$}{ \textbf{break} }
    Let $G \gets$ \FeasibleGraphTrim{$G-e, \{e_t\}$} \tcp*{Permanently remove $e$ from $G$} 
}
\Return \False \;
}
\end{algorithm}

The detection mechanism (lower layer) operates through the following steps:
\textbf{(i) Selection and Candidate Inclusion:} The algorithm selects a computation walk $W_p$ in the current graph. If $W_p$ contains target edge $e_t$, this confirms the existence of a computing-targeted walk.
\textbf{(ii) Trial Elimination of Implausible Edges:} To detect unnecessary edges, the algorithm simulates the removal of an ``implausible'' edge---the first splitting edge in the computing-futile walk. Then it observes whether the graph undergoes a \textit{Total Collapse}--- the emptiness of the feasible graph due to the absence of any computing-targeted walks.
\textbf{(iii) Criticality Identification:} If the trial removal causes a collapse of the feasible graph, it indicates that $W_p$ contains essential \textbf{computing-effective edges}. Consequently, the algorithm re-computes the feasible graph excluding the essential edge while preserving the computing-futile walks, and identifies the first edge in the re-computed graph where another computation walk deviates from $W_p$.

\begin{lemma}[Correctness of Detecting Computing-Futile/Redundant Edges] \label{highlem:detect_futile_redundant_edge}
Assume that the procedure \FeasibleGraphTrim{} correctly removes, in polynomial time, all edges lacking vertical or horizontal consistency while preserving any computation walks to the designated final edges but collapses to empty if the essential edge before the second splitting edge is missing. Then, the algorithm \DetectRedundantFutileTargetEdge{} correctly identifies either a computing-futile or a computing-redundant edge, returns the target edge $e_t$ (if a valid walk exists), or returns \textsc{Nil} (if no such walk is possible)---all within deterministic polynomial time.
\end{lemma}

\begin{proof}
The proof relies on the structural response of the trimmed graph $G'$ to the removal of the first splitting edge $e_s$.
\begin{itemize}
\item \textbf{Identification of Computing-Targeted Walks and Structural Collapse:}
If the selected $W_p$ contains the verification target edge $e_t$, the algorithm returns $e_t$, confirming the identification of a computing-targeted walk. Conversely, if the trimmed graph $G'$ becomes empty after the temporary removal of $e_s$, it implies that $e_s$ was an \textbf{essential edge}---a necessary component for every possible \textit{computing-targeted walk} in the current $G'$---by the premise. This collapse demonstrates that $W_p$ is a computing-targeted walk up to the removed edge $e_s$. Since the trimming process preserves computing-futile walks while the essential edge is removed, any remaining structure in $G'_{pre}$ indicates the existence of only computing-futile walks.

\item \textbf{Identification of Computing-Redundant or Computing-Futile Edges:}
Since the selected $W_p$ coincides with a computing-targeted walk up to the removed edge $e_s$, the first edge $e$ that diverges from $W_p$ within a computing-futile walk serves as a definitive structural witness of either a transition leading to a dead-end (futile) or an unnecessary bypass (redundant). By selecting such an edge $e \notin W_p$, the algorithm effectively isolates an edge that does not contribute to all computing-targeted walks. If no such edge exists, the algorithm correctly returns \textsc{Nil}, signifying that no computing-targeted walk exists in the graph.

\item \textbf{Polynomial Time Convergence:}
The \texttt{while} loop executes at most $|E|$ iterations, as each step either confirms the target or temporarily removes at least one edge. Since \FeasibleGraphTrim{} is performed in polynomial time $\mathcal{O}(\text{poly}(n))$ by the premises, and each internal operation—such as the recording of $G'_{pre}$ and edge selection—is performed in $\mathcal{O}(\text{poly}(n))$, the overall complexity is $\mathcal{O}(|E| \cdot \text{poly}(n))$. Since $|E|$ is polynomially bounded by the input size, the algorithm converges within $\mathcal{O}(\text{poly}(n))$.
\end{itemize}
\end{proof}

\begin{corollary}[Correctness and Polynomial Bound of Verification] \label{highcor:verification_walk_correctness}
If the \FeasibleGraphTrim{} procedure correctly removes all inconsistent edges in polynomial time with preserving all computing-targeted walks and total collapse property, then the \VerificationWalkToTarget{} procedure correctly identifies the existence of a valid computation walk to the target edge in polynomial time. 
\end{corollary}

\begin{proof}
The proof follows from the polynomial size of the edge set $E$ and the correctness of the sub-procedures. First, by \cref{highlem:detect_futile_redundant_edge}, \DetectRedundantFutileTargetEdge{} correctly detects computing-futile or redundant edges and preserves at least one computation walk to the target edge if it exists, or identifies the target if a computing-targeted walk to $e_t$ is confirmed. Next, each iteration of the \texttt{while} loop involves the \textit{permanent removal} of at least one edge $e$ from the graph $G$. Since $|E|$ is of polynomial size and the detection mechanism is bounded as per \cref{highlem:detect_futile_redundant_edge}, the total complexity is $\bigO(|E| \cdot \text{poly}(n)) = \bigO(\text{poly}(n))$.
\end{proof}

\subsection{Feasible Graph as a Trimming Tool (The Utility)} \label{sec:trimming_tool_feasible_graph}
The efficiency of global verification is sustained by our trimming utility. By enforcing local consistency via bi-directional sweeps across edge slices, this tool removes `dangling' or inconsistent edges that cannot form a part of any computation walk or a valid sequence of transitions.
The formal proof and detailed explanation are provided in Appendix \cref{sec:feasible_graph}.

The stability of the \textit{Feasible Graph} is maintained through a unified consistency check. For any edge $e \in E_i$ to survive the trimming operator, it must satisfy both vertical and horizontal integrity constraints, with specific exceptions for the temporal marginal anchors of a computation walk. We define an edge that is horizontally or vertically adjacent to another as a \textit{step-adjacent edge}. To accommodate the temporal limits of a finite execution, we define the following boundary edge sets:
\begin{itemize}
    \item \textbf{Floor Edge:} The first edge to appear in an edge slice $E_i$ within a valid computation walk. Edges whose head nodes are at tier $0$ belong to this category.
    \item \textbf{Cover Edge:} The set of all structural potential edges that can serve as a \textit{Ceiling Edge}---the last edge to appear in an edge slice $E_i$ within a valid computation walk---on any computation walk to the designated final edges within the entire graph.
\end{itemize}
To maintain the integrity of the computation lattice, any edge failing to satisfy local consistency is defined as a \textbf{Step-Pendant Edge}. These are classified into two categories:
\begin{itemize}
    \item \textbf{Horizontal Step-Pendant:} An edge $e \in E_i$ that is neither an \textit{Initial} nor a designated \textit{Final} edge, yet lacks at least one adjacent edge.
    \item \textbf{Vertical Step-Pendant:}  An edge $e \in E_i$ that is not a \textit{Cover Edge} but lacks an \textbf{Index-Succedent} in $E_i$, or is not a \textit{Floor Edge} but lacks an \textbf{Index-Precedent} in $E_i$.
\end{itemize}
An edge that is not a horizontally step-pendant is considered \textbf{Index-Adjacent} to its neighboring edge slices, ensuring a consistent transition across the spacetime manifold. 

\begin{algorithm}[!ht]
    \DontPrintSemicolon
    \SetKwFunction{SweepEdges}{SweepEdges}
    \SetKwFunction{StepUp}{StepUpEdges}
    \SetKwFunction{StepDown}{StepDownEdges}
    
    \caption{Feasible Graph Construction via Structural Trimming}
    
    \Input{Initial graph $G$, designated final edges $E_f$.}
    \Output{Feasible graph $H$ with respect to $E_f$.}
    \BlankLine
    \Function{\FeasibleGraphTrim{$G, E_f$}}{
        $H \gets G$, $C \gets$ Compute Cover Edges with respect to $E_f$ in $G$\;
        \While{$H$ is changing and not empty}{
            $H \gets$ \SweepEdges{$H, C, E_f, \text{min\_idx}, +1$} \tcp*{min\_idx:minimum edge index of H}
            $H \gets$ \SweepEdges{$H, C, E_f, \text{max\_idx}, -1$}\tcp*{max\_idx:maximum edge index of H}
        }
        \Return{$H$}\; 
    } 

    \Function{\SweepEdges{$G, C, E_f, i, d$}}{
        $H_{new} \gets \emptyset$\;
        \While{edge slice $E_i$ exists}{

            $I \gets$ \StepUp{$G, H_{new}, E_i, E_{i-d}$} \tcp*{Step 1: Filtering Upward}
            $H_{new} \gets H_{new} \cup$ \StepDown{$G, I, C \cap I$} \tcp*{Step 2: Filtering Downward}
            $i \gets i + d$\;
        }
        \Return{$H_{new}$}\;
    }

    \Function{\StepUp{$G, H, E_{curr}, E_{prev}$}}{
        Find floor edges $E_b$ in $E_{curr}$\;
        Let $I \gets \{ \text{Index-Succedent upward chains from } E_b \mid \text{index-adjacent to } E_{prev} \}$\;
        \Return{$I$}\;
    }

    \Function{\StepDown{$G, I, C_{cover}$}}{
        Let $I' \gets \{ \text{Index-Precedent downward chains from } C_{cover} \}$\;
        \Return{$I'$ }\;
    }
\end{algorithm}

\begin{lemma}[Correctness of Feasible Graph Trimming]\label{highlem:feasible_graph_trimming_correctness}
The  \FeasibleGraphTrim{} procedure extracts a consistent subgraph $H$  from the computation graph $G$ within polynomial time such that:
\begin{enumerate}
    \item $H$ contains no \textbf{Step-Pendant Edges} (i.e., all local inconsistencies are trimmed).
    \item All \textbf{Feasible Walks} (computation walks to the designated final edges) are strictly preserved.
	\item Removing an essential edge before the second-splitting edge leads to a total collapse of $H$.
\end{enumerate}
\end{lemma}

\begin{proof}
\begin{itemize}
\item \textbf{Elimination of Step-Pendant Inconsistency:}
Assume for contradiction that a \textsf{Step-Pendant} edge $e$ remains in the converged graph $H$. 
If $e$ is a \textbf{Horizontal Step-Pendant}, it must lack an index-adjacent edge in its neighboring edge slice. However, the \textit{Step-Up} phase only incorporates edges that are \textbf{index-adjacent}; any edge failing this condition is  never added during iterative sweeps. 
If $e$ is a \textbf{Vertical Step-Pendant}, it must lack an \textit{index-precedent} (if $e \notin \text{Floor}$) or an \textit{index-succedent} (if $e \notin \text{Cover}$). In the \textit{Step-Up} phase, an edge is retained only if it possesses a valid index-precedent or is a \textsf{Floor} edge. In the \textit{Step-Down} phase, it is retained only if it has an \textit{index-succedent} or is a \textsf{Cover} edge. 
Thus, the existence of any \textsf{Step-Pendant} contradicts the iterative sweep mechanism of the algorithm.

\item \textbf{Preservation of Feasible Walks:}
Assume a feasible walk $W \subseteq G$ (leading to a designated final edge) is not preserved in $H$. Let $e \in W$ be the \textbf{first edge} of $W$ to be removed during the trimming process.
Since $W$ is a valid computation path, all its constituent edges are by definition \textbf{index-adjacent} to their respective neighboring edge slices. Thus, $e$ cannot be a Horizontal Step-Pendant.
Regarding vertical consistency, since $W$ starts from a \textit{Floor edge} and terminates at a \textit{Cover edge}, every intermediate edge $e \in W$ necessarily has an index-precedent and an index-succedent within $W$. 
Since $e$ is the \textit{first} to be removed, its neighbors in $W$ must still exist in the graph at that moment, satisfying its adjacency requirements.
Therefore, $e$ cannot be identified as a Step-Pendant, and no such ``first removed edge" can exist. It follows that $W \subseteq H$ is preserved.

\item \textbf{Total Collapse:}
Let the time value be the number of transitions required to reach the node, which is well-defined due to the certificate-oblivious property. We proceed by induction on $n = |E(G-E_r)|$, where $E_r$ is a set of removed feasible edges (edges on feasible walks) containing an essential edge located before the second-splitting edge.
First, there must exist at least one \textit{Step-Pendant} edge $e'$ in $G' \setminus E_r$ (where $G' = G \cup F$). 
This holds because the feasible graph of $G \setminus E_r$ is empty by the inductive hypothesis, and any edge $f \in F$—the set of edges incident to nodes with maximum time—cannot share a \textit{step-adjacent} edge with others, nor can it convert an existing \textit{step-pendant} edge into a non-\textit{step-pendant} one.
Second, if the \textit{step-pendant} edge $e'$ is a feasible edge, we apply the inductive hypothesis to $G' \setminus E'_r$ where $E'_r = E_r \cup \{e'\}$. 
Otherwise, we apply it to $G'' \setminus E_r$ where $G'' = G' \setminus \{e'\}$.

\item \textbf{Complexity:}
The total number of edges $|E|$ is polynomial with respect to the input size. Each iterative sweep constructs a graph of size at most $|E|$ using a polynomial number of edge operations. Since each sweep removes at least one edge and never reinstates it, the process terminates in at most $|E|$ sweeps. Consequently, the total complexity remains \textbf{polynomial}. 
\end{itemize}
\end{proof}

\subsection{P=NP:Complexity and Convergence.}

\begin{theorem}[Main Result: $\mathsf{P}=\mathsf{NP}$]
For every language $\mathcal{L}$ in the complexity class $\mathsf{NP}$, there exists a deterministic algorithm that decides $\mathcal{L}$ in polynomial time. Consequently, $\mathsf{P} = \mathsf{NP}$.
\end{theorem}

\begin{proof}
The proof is established through the integration of the results from the preceding sections. 
First, according to \cref{highlem:feasible_graph_trimming_correctness}, any computation graph can be reduced to a feasible graph in polynomial time; this process eliminates all locally inconsistent edges while strictly preserving all computation walks to the designated final edges but results in a total collapse (an empty graph) if the essential edge before the second splitting edge is missing.
Building upon this result, the \VerificationWalkToTarget{} procedure, as established in \cref{highcor:verification_walk_correctness}, accurately \textbf{decides} the existence of \textit{computing-targeted walks} in a given computation graph within polynomial time. 
Finally, as \textbf{proven} by \cref{highlem:footmark_construction_correctness}, the incremental extension of $G_{footmark}$ via the exhaustive candidate edges in \cref{higalg:incremental_extension_of_footmarks} utilizes the existence of computing-targeted walks to each edge to decide the existence of an accepting certificate for any language $\mathcal{L} \in \mathsf{NP}$ in polynomial time. 
This sequence of deterministic polynomial-time operations ensures that $\mathsf{NP} \subseteq \mathsf{P}$, which, given the trivial inclusion $\mathsf{P} \subseteq \mathsf{NP}$, leads to the conclusion that $\mathsf{P} = \mathsf{NP}$.
\end{proof}

Crucially, every stage of our framework is constructed upon unit-edge operations—specifically, the incremental addition of edges, step-by-step removal of edges, or graph traversals designed to avoid revisiting edges---yielding an overall  time complexity of $\bigO(|E|^k)$. A formal correctness proof and comprehensive analysis of the time complexity is provided in \cref{subsec:reduction_from_np_to_p} of the Appendix.

\section{Related Work} \label{sec:related_works}
The foundations of the $\mathsf{P}$ vs $\mathsf{NP}$ problem were established through the development of $\mathsf{NP}$-completeness theory. The Cook–Levin theorem \cite{cook1971complexity} introduced the notion of $\mathsf{NP}$-completeness by showing that the Boolean satisfiability problem is complete for $\mathsf{NP}$ under polynomial-time reductions. Subsequently, Karp demonstrated that a wide range of fundamental combinatorial problems are $\mathsf{NP}$-complete \cite{karp1972reducibility}, establishing these problems as central objects in complexity theory.

Despite decades of sustained effort, traditional approaches have encountered fundamental barriers. Relativizing techniques \cite{baker1975relativizing}, Natural Proofs \cite{razborov1993natural}, and Algebrizing techniques \cite{aaronson2008algebrizing} have each identified distinct structural limitations that any valid proof must transcend. 

Specifically, the \textbf{Natural Proofs} barrier demonstrates the inability to establish a super-polynomial lower bound. Our framework suggests that this difficulty arises from the non-existence of the lower bound itself. In this sense, the ``hardness'' identified by Natural Proofs is not an obstacle to our construction; rather, it is a structural reflection of the inherent polynomial-time solvability of $\mathsf{NP}$ problems within our unified graph manifold.

The Relativization Barrier \cite{baker1975relativizing} for $\mathsf{P}=\mathsf{NP}$ proofs arises from a fundamental constraint in the functional simulation of Nondeterministic Turing Machines (NTMs) such as an $NP$ decider. It demonstrates that if a proof mechanism remains invariant even when augmented with a black-box oracle---implying the proof must hold across all oracle worlds---it cannot resolve the $\mathsf{P}=\mathsf{NP}$ question. This is because there exists a specific oracle $B$ such that $\mathsf{P}^B \neq \mathsf{NP}^B$, creating a logical contradiction for any oracle-independent proof attempt. This barrier implies that any successful proof must transcend such oracle-independent methods. 

Our framework, however, is inherently immune to this barrier because it bypasses the simulative paradigm of NTM entirely. Instead of ``running'' or ``tracing'' NTM branches, our approach provides a \textbf{constructive derivation} of the computation manifold within a static spacetime lattice. Crucially, the object of our structural analysis is the \textbf{Deterministic Verifier (DTM)}, not the NTM where an oracle $B$ would exert asymmetrical power. By shifting the focus to the fixed $\delta$-logic of the Verifier and utilizing the \textbf{topological information} of edges and paths---specifically the local and global connectivity invariants of the transition graph---we operate in a domain where the query-complexity gap of oracle-augmented NTMs is a logical non-sequitur. Furthermore, since the introduction of an oracle would destroy the very geometric manifold and fixed alphabet constraints we analyze, this is not an oracle-independent proof; thus, the relativization thesis is categorically inapplicable to this white-box structural model.

Similarly, our framework is exempt from the Algebrization Barrier \cite{aaronson2008algebrizing}, which targets proof techniques relying on arithmetization—the mapping of Boolean computations to low-degree polynomials. As with the relativization barrier, our methodology avoids the arithmetization of NTMs. We represent the verifier DTM transitions directly as edges in a graph without any intermediate polynomial mapping or functional extension; thus, our framework remains outside the mathematical domain where an ``extended algebraic oracle'' can be applied. Consequently, the proposed methodology bypasses both the NTM-based and algebraic constraints that have historically limited $\mathsf{P}$ vs $\mathsf{NP}$ research.

The long-standing complexity barriers—relativization, algebrization, and natural proofs—are not inherent obstacles to the $\mathsf{P} = \mathsf{NP}$ identity itself, but rather limitations of specific proof techniques. By shifting the perspective from black-box NTM-based lower bounds to an explicit, graph-based simulation of deterministic transitions, we provide a constructive methodology that operates beyond these traditional constraints.

To operationalize this structural shift, we ground our framework in the formal foundation of Certificate-Oblivious Turing Machines (OTMs), defined by head movements that are independent of the certificate content for a given input. 
The concept of obliviousness, originally formalized by Pippenger and Fischer~\cite{pippenger1979relations}, establishes that any Turing machine can be simulated by an oblivious counterpart with only polynomial overhead. 
By leveraging this transformation, we standardize the computational structure, ensuring that the input length—rather than the specific certificate values—governs the progression of the verification process.

\section{Preliminaries}

In this section, we introduce key definitions and concepts relevant to our study, providing a rigorous mathematical framework to support our analysis. This includes essential \textit{computation-theoretic} and \textit{graph-theoretic} concepts, such as formal definitions of NP, verifiers, and certificates, which lay the foundation for our approach.

Throughout this paper, we use \textit{function parameters} exclusively for input values, while \textit{procedure parameters} are annotated with directionality indicators:  
$\In$ denotes an input parameter (read-only),  
$\Out$ denotes an output parameter (write-only), and  
$\InOut$ denotes a parameter that is both read and modified during execution.

\subsection{Computation Theory} \label{sec:computation_theory}

A Turing machine consists of a tape and a tape head.
The tape is an infinite sequence of cells indexed by the integers, each containing a symbol from a finite alphabet.
The tape head can read the symbol in the current cell, write a symbol to the same cell, and move by one cell to the left or to the right.

Formally, a Turing machine is a tuple
\[
M = (Q, \Sigma, \Gamma, \delta, q_0, F),
\]
where:
\begin{itemize}
\item $Q$ is a finite non-empty set of states,
\item $\Gamma$ is a finite non-empty tape alphabet containing the blank symbol $\epsilon$,
\item $\Sigma \subseteq \Gamma \setminus \{\epsilon\}$ is the input alphabet,
\item $q_0 \in Q$ is the initial state,
\item $F \subseteq Q$ is the set of final (halting) states, and
\item $\delta$ is a transition relation
\[
\delta \subseteq ((Q \setminus F) \times \Gamma) \times (\Gamma \times \{-1, +1\} \times Q).
\]
\end{itemize}

A transition $((q, \sigma), (\sigma', d, q')) \in \delta$ indicates that when the machine is in state $q$ and reads symbol $\sigma$, it may write symbol $\sigma'$, move the tape head in direction $d$, and enter state $q'$.
The input is written on the tape starting at cell $0$, and the tape head initially points to cell $0$.
This convention is fixed throughout the paper and is adopted for notational convenience.

With this definition, each instruction can be represented as a 5-tuple $(q, \sigma, \sigma', d, q')$ where:
$q \in Q$ is the current state, 
$\sigma \in \Gamma$ is the current tape symbol, 
$\sigma' \in \Gamma$ is the symbol to write, 
$d \in \{-1,+1\}$ is the movement direction, and 
$q' \in Q$ is the next state.

An execution of a Turing machine program is a sequence of such instructions, denoted as $I_1, I_2, \cdots, I_n$, where $n$ is the number of instructions executed.

\begin{definition}
If a Turing machine has at least two transitions defined for the same state and input symbol, then it is called \textbf{nondeterministic}; otherwise, it is called \textbf{deterministic}~\cite{hein1996theory}.
We denote a deterministic Turing machine as DTM and a nondeterministic Turing machine as NTM.
\end{definition}

In a deterministic Turing machine, the transition relation $\delta$ becomes a partial function:
\[
\delta : (Q \setminus F) \times \Gamma \to \Gamma \times \{-1, +1\} \times Q.
\]

\begin{definition}
The \textbf{cell index} is defined as the offset from the starting cell, such that the cells to its right are assigned positive indices and those to its left are assigned negative indices.
\end{definition}

A \textbf{decision problem} is a computational problem where the output is either ``yes'' or ``no''. Such problems are typically addressed using a specific kind of Turing machine known as an \textbf{acceptor}.

\begin{definition}\label{def:acceptor}
An \textbf{acceptor Turing machine} is a Turing machine that always halts in either an accepting state or a rejecting state. This is functionally equivalent to a recognizer, but the term ``acceptor'' emphasizes its role in deciding acceptance. The final states of an acceptor consist only of a designated accepting state $\qacc$ and rejecting state $\qrej$.
\end{definition}
Since the set of final states for an acceptor is fixed as $F = \{\qacc, \qrej\}$, we can more specifically denote an acceptor Turing machine as a tuple:$M_{A} = (Q, \Sigma, \Gamma, \delta, q_0, \qacc, \qrej)$,
where $\qacc, \qrej \in Q$ are the distinct accepting and rejecting states, respectively. In this case, the transition relation $\delta$ is defined over $(Q \setminus \{\qacc, \qrej\}) \times \Gamma$.

Accordingly, a decision problem is said to be \textbf{decidable} if there exists an acceptor Turing machine that halts on every input and correctly decides whether to accept or reject.

\begin{definition}[NP via Nondeterministic Turing Machine]
A decision problem is in the class \textbf{NP} if it can be solved by a nondeterministic Turing machine within polynomial time. More formally, a language $\mathcal{L} \subseteq \Sigma^*$ belongs to \textsf{NP} if and only if there exists a nondeterministic Turing machine $M$ and a polynomial $p(\cdot)$ such that for every $X \in \Sigma^*$:
\begin{itemize}
    \item Every computation path of $M$ on input $X$ halts within $p(|X|)$ steps.
    \item $X \in \mathcal{L} \iff$ there exists at least one computation path of $M$ on input $X$ that reaches an \texttt{accept} state.
\end{itemize}
\end{definition}

An equivalent characterization of \textsf{NP} uses the notion of an \textbf{efficient certifier}, capturing the verifier-based perspective of nondeterministic computation.

\begin{definition}[NP via Verifier and Certificate]~\cite[p.464]{Kleinberg2005-qp} \label{def:certifier-based_definition}
A decision problem (or language) $\mathcal{L}$ belongs to the class \textsf{NP} if there exists a deterministic polynomial-time algorithm $B$, called an \textbf{efficient certifier} for $\mathcal{L}$, such that:
\begin{enumerate}
  \item $B$ is a deterministic algorithm that runs in polynomial time;
  \item $B$ takes two inputs: an instance $X$ and a certificate $Y$, and outputs ``yes'' or ``no'';
  \item There exists a polynomial function $q$ such that for every string $X$, we have $X \in \mathcal{L}$ if and only if there exists a certificate $Y$ with $|Y| \le q(|X|)$ such that $B(X, Y) = \text{``yes''}$.
\end{enumerate}
\end{definition}

While the above provides a high-level algorithmic view, a more rigorous formulation is required for analyzing the mechanical transitions of computation. This leads to a reformulation using acceptor Turing machines, where the physical layout of the input tape is explicitly considered.

\begin{definition}[NP via Acceptor Verification]
A language $\mathcal{L} \subseteq \Sigma^*$ belongs to \textsf{NP} if there exists a deterministic acceptor Turing machine $M$ (the verifier) and polynomials $p(\cdot)$ and $q(\cdot)$ such that for every $X \in \Sigma^*$ and  $Y \in \Sigma^*$ with $|Y| \le q(|X|)$:
\begin{itemize}
    \item $M$ halts on the input string $X\#Y$ within $p(|X| + |Y| + 1)$ steps, where the $+1$ accounts for the delimiter in the concatenated input.
    \item $X \in \mathcal{L} \iff \exists Y$ such that $M$ reaches an \texttt{accept} state on input $X\#Y$.
\end{itemize}
\end{definition}
\begin{remark}
It is important to emphasize that not every string $Y \in \Sigma^*$ within the length bound $q(|X|)$ is necessarily a valid or well-formatted certificate for $X$. These candidates may consist of arbitrary symbol sequences that do not constitute a correct proof; the verifier $M$ \textbf{is designed to} reject such invalid candidates by the definition. The essential property of \textsf{NP} is the existence of \textit{at least one} string $Y$ (among the exponential space of all possible candidates) that causes $M$ to reach an \texttt{accept} state when $X \in \mathcal{L}$. Throughout this paper, we use the term \textbf{problem instance} to refer to the fixed sequence $X\#$, and we use \textbf{verifier} to refer to the deterministic acceptor $M$ to maintain consistency with the standard computational model.
\end{remark}

This definition implies that every $NP$ problem has a deterministic verifier machine $M$ where $Y$ consists of all the symbols of $\Sigma$.

Beyond ordinary verifier Turing machines, there exists a class of structured verifiers where head trajectories remain uniform across different certificates. To formally analyze computation paths that are as independent as possible from the specific symbols processed—thereby ensuring a consistent computational layout—we introduce the concept of an oblivious Turing machine.
\begin{definition}[Oblivious Turing Machine] \label{def:oblivious_tm}
A deterministic Turing machine $M = (Q, \Sigma, \Gamma, \delta, q_0, F)$ is said to be \textbf{oblivious} if the position of its tape head at any computation step $t$ is uniquely determined by the length of the input string $|X|$ and the step index $t$, completely independent of the content of $X$. 
More formally, let $H_M(X, t)$ denote the cell index pointed to by the tape head of $M$ at step $t$ given an input $X \in \Sigma^*$. Then $M$ is oblivious if for any two inputs $X, X' \in \Sigma^*$ satisfying $|X| = |X'|$, it holds that:
\[
H_M(X, t) = H_M(X', t) \quad \text{for all } t \ge 0.
\]
\end{definition}

\begin{theorem}[Pippenger-Fischer Theorem] \label{thm:pippenger_fischer} 
Let $M$ be a standard deterministic multi-tape Turing machine that decides a language $\mathcal{L}$ within time $T(n)$ for an input of length $n$. Then, there exist functionally equivalent oblivious Turing machines that decide $\mathcal{L}$ such that:
\begin{enumerate}
    \item On any input $X \in \Sigma^*$, the final halting state ($\qacc$ or $\qrej$) and the output written on the tape are identical to those of $M$;
    \item There exists a \textbf{two-tape} oblivious Turing machine $M_{\text{obl}}^{(2)}$ whose total number of computation steps is bounded by $\bigO(T(n) \log T(n))$;
    \item There exists a \textbf{one-tape} oblivious Turing machine $M_{\text{obl}}^{(1)}$ whose total number of computation steps is bounded by $\bigO(T(n)^2)$.
\end{enumerate}
\end{theorem}

\begin{remark} \label{rem:verifier_otm}
In the original paper by Pippenger and Fischer~\cite{pippenger1979relations}, Theorem~3 formally highlights the two-tape result to achieve the $\bigO(T(n) \log T(n))$ upper bound. However, the one-tape $\bigO(T(n)^2)$ simulation mechanism (analyzed as `Version I` in their proof) is a well-established foundational result. Since our proposed architecture is inherently based on a single-tape framework, we utilize the one-tape variant presented in Item 3. This eliminates the need for complex multi-tape synchronization and aligns perfectly with our structural constraints.

To achieve obliviousness within a single tape, the simulator systematically sweeps across the entire used tape boundaries during every step, maintaining a dedicated marker to keep track of the virtual head's position.
This choice is fundamentally driven by architectural compatibility and implementation efficiency. Since this head marker is maintained within an enlarged alphabet space rather than altering the raw input stream, the input format remains strictly identical to that of the original Verifier Turing Machine. 
Consequently, this one-tape approach ensures that the input format of the constructed oblivious verifier Turing machine can be regarded as completely identical to that of the original verifier Turing machine, thereby fundamentally avoiding any requirement for input format translations or algorithmic modifications.
\end{remark}

We describe a simplified version of the transformation adapted to a single-tape Turing machine without modifying the original input format. 
The alphabet is enlarged to include a head marker $\sigma^H$ for each symbol $\sigma$ and a boundary marker $\$$.
 Furthermore, the states are extended to embed the direction and operation mode, while the original halting states are preserved.

\begin{remark}[Mechanical transform to 1-Tape Oblivious Verifier Turing Machine]
To achieve an oblivious simulation of a verifier Turing machine on a 1-tape model with quadratic-factor overhead, we define the following deterministic mechanical protocol:

\begin{enumerate}
    \item \textbf{Tape Initialization (Pre-processing):}
    Let $L = s_0 s_1 \dots s_n$ be the input string. The tape $L'$ is preprocessed before the main execution step as:
    \[ L' = \$ \cdot s_0^H \cdot s_1 \cdot \dots \cdot s_n \cdot \$ \]
    where $\$$ denotes the unique boundary markers at indices $-1$ and $n$, and $s_0^H$ (with the head marker at index $0$) represents the first symbol coupled with the head marker. The machine is initialized in state $(Build, Right)$ at index $0$, which transitions to $q_0$ with an embedded $(Scan, Right)$ mode upon completion, where $q_0$ is the initial state of the original Turing machine.

    \item \textbf{Execution Cycles (Step):}
    The simulation proceeds through three mutually exclusive operational modes:
    \begin{itemize}
        \item \textbf{Scan Mode $(Scan, Dir)$:} The machine traverses the tape in direction $Dir \in \{Left, Right\}$. The transition logic is prioritized as follows:
        \begin{itemize}
            \item \textbf{Move Trigger:} Upon encountering a symbol marked with $H$ (e.g., $s_i^H$), the machine executes the original transition rule and shifts to \textit{Move Mode}.
            \begin{itemize}
                \item \textbf{Match ($Dir_{TM} = Dir$):} Execute the original transition by changing the state, writing the output symbol $s'_i$, and transitioning to \textit{Move Mode}.
                \item \textbf{Mismatch ($Dir_{TM} \neq Dir$):} Maintain \textit{Scan Mode} and continue scanning in $Dir$, effectively ignoring the head marker to preserve scan continuity.
            \end{itemize}
            \textit{Example:} Given an original transition $(q, A) \rightarrow (p, B, Right)$:
            \begin{itemize}
                \item If $Dir = Right$, the machine triggers: $$(q\_Scan\_Right, A^H) \rightarrow (p\_Move\_Right, B, Right)$$.
                \item If $Dir = Left$, the machine bypasses the symbol: $$(q\_Scan\_Left, A^H) \rightarrow (q\_Scan\_Left, A^H, Left)$$.
            \end{itemize}
            \item \textit{Boundary Trigger:} Upon encountering the boundary marker $\$$, the machine transitions to \textit{Expand Mode}.
        \end{itemize}
        \item \textbf{Move Mode $(Move, Dir)$:} The machine updates the virtual head position using the head marker and returns to \textit{Scan Mode} while retaining the direction.
            \textit{Example:} $(p\_Move\_Right, B) \rightarrow (p\_Scan\_Right, B^H, Right)$.
        \item \textbf{Expand Mode $(Expand, Dir)$:} A deterministic procedure to accommodate the infinite tape requirement:
        \begin{enumerate}
            \item Replace the boundary marker $\$$ with the empty symbol $\epsilon$.
            \item Shift the boundary by generating a new $\$$ at the adjacent cell in direction $Dir$.
            \item Update the scan direction $Dir' \leftarrow \neg Dir$ and resume $(Scan, Dir')$.
        \end{enumerate}
    \end{itemize}
	Each round-trip cycle executes exactly one transition of the original Turing machine.

    \item \textbf{Formal Properties:}
    \begin{itemize}
        \item \textit{Decoupling:} Expanding the data domain $\Gamma$ using the control symbols $(\$, \sigma^H \text{for all} \sigma \in \Gamma)$ ensures data integrity for each symbol $\sigma \in \Gamma$.
        \item \textit{Deterministic Completeness:} The protocol maintains boundary markers dynamically, ensuring sufficient tape space at each step.
        \item \textit{Initialization Safety:} The $(Build, Right)$ state is a non-revisitable pre-processing phase that guarantees the initial head marker is at index $0$ and boundary markers are correctly positioned at indices $-1$ and $n$, ensuring the simulation begins in a valid initial state $q_0$ with the embedded $(Scan, Right)$ mode without premature expansion.
    \end{itemize}
\end{enumerate}
\end{remark}

\subsection{Graph Theory}\label{sec:graph_theory}
As the behavior of a Turing machine will later be modeled using graph structures,
we next introduce basic notions from graph theory that will be used to describe and analyze
such representations.

A directed graph is an ordered pair $G=(V,E)$ where $V$ is the set of vertices and $E \subseteq V \times V$ is the set of directed edges.

\newcommand{\init}{\mathrm{init}}
\newcommand{\term}{\mathrm{term}}
\newcommand{\Incoming}{\mathrm{In}}
\newcommand{\Outgoing}{\mathrm{Out}}
For an edge $e=(u,v) \in E$, defined as an ordered pair of vertices, $u$ is the \textbf{initial vertex} (or \textbf{tail}) of $e$, and $v$ is the \textbf{terminal vertex} (or \textbf{head}) of $e$. We denote these by $\init(e)$ and $\term(e)$, respectively.
Thus, an edge $e$ points from its tail to its head.

In a directed graph, an edge $e = (u, v)$ is said to be an \textbf{incoming edge} to vertex $v$ and an \textbf{outgoing edge} from vertex $u$.
For any vertex $w \in V(G)$, the sets of all incoming and outgoing edges incident to $w$ in graph $G$ are denoted by $\Incoming_G(w)$ and $\Outgoing_G(w)$, respectively. If the context is clear, the subscript $G$ may be omitted.
For convenience, the notation $e \in G$ means that $e$ is an edge of $G$, i.e., $e \in E(G)$.

The order of a graph $G$, denoted by $|V(G)|$, is the number of vertices.
We define the \textbf{size of a graph} $G$ to be the number of its edges, denoted by $|E(G)|$.

A \textbf{degree} of a node $v$, denoted by $\textrm{deg}(v)$, is the total number of edges incident to it. Specifically, in a directed graph, the degree is the sum of its \textit{in-degree} (number of incoming edges) and \textit{out-degree} (number of outgoing edges). A \textbf{pendant edge} is then defined as an edge incident to a node with a degree of $1$.

A \textbf{walk} of length $k$ in a directed graph $G$ is traditionally defined as an alternating sequence $v_0 e_0 v_1 e_1 \dots e_{k-1} v_k$ of vertices and edges such that $e_i = (v_i, v_{i+1})$ for all $0 \le i < k$. For notational convenience, a walk can be uniquely represented as an \textbf{edge sequence} $W = e_0 e_1 \dots e_{k-1}$ when $k > 0$. Alternatively, it can be denoted as a \textbf{vertex sequence} $W = (v_0, v_1, \dots, v_n)$, which is particularly useful when analyzing the internal attributes of each vertex along the walk. A \textbf{subwalk} is any contiguous subsequence of $W$ that itself constitutes a valid walk.

A path is a walk in which all vertices are distinct (i.e., no vertex is repeated). A path in a directed graph $G$ is called maximal if it cannot be extended in $G$ while maintaining the path property.
For an edge $e = (u, v)$ in $G$, both vertices $u$ and $v$ are said to be \textbf{incident to} the edge $e$.
For a nonempty set of edges $E$, we say that the subgraph $\langle E \rangle$ (or $G[E]$) is the subgraph \textbf{induced by} $E$ if it consists of all edges in $E$ and all vertices incident to at least one edge in $E$.

\begin{definition}[Graph Extension and Removal]
Let $G = (V, E)$ be a graph and $H \subseteq G$ be a subgraph.
For a set of edges $F \subseteq E(G)$, let
\[
V(F) := \{ u \in V(G) \mid \exists v \text{ such that } (u,v)\in F \text{ or } (v,u)\in F \}.
\]
We define:
\begin{itemize}
    \item The extension of $H$ by $F$, denoted $H + F$, as the graph
    \[
        H + F := (V(H) \cup V(F),\, E(H) \cup F).
    \]
    \item The removal of $F$ from $H$, denoted $H - F$, as the graph
    \[
        H - F := (V(H),\, E(H) \setminus F).
    \]
\end{itemize}
For a single edge $e$, we write $H + e$ and $H - e$ instead of
$H + \{e\}$ and $H - \{e\}$, respectively.
\end{definition}

\subsection{Summary of Notations}
The following table summarizes the key symbols and notations used throughout the paper. 
For clarity, only the most important or frequently referenced symbols are included. 
A more comprehensive list of technical terms and definitions is provided in \cref{appendix:terminology}.

\clearpage

\begin{table}[htbp]
\centering 
\caption{Summary of Key Symbols}
\begin{tabularx}{\textwidth}{XlXlX}
\toprule
\textbf{Symbol} & \textbf{Description} & \textbf{Reference} \\
\midrule
$G = (V, E)$ & Computation graph with vertex set $V$ and edge set $E$ & \cref{sec:graph_theory} \\
$Q$ & Set of Turing machine states & \cref{sec:computation_theory} \\
$\Gamma$ & Tape alphabet & \cref{sec:computation_theory} \\
$s'$ & Symbol to write on tape as per transition function (element of $\Gamma$) &	\cref{subsec:comp_graph}\\
$s$ & Current symbol read from tape (element of $\Gamma$) & \cref{subsec:comp_graph}\\
$\Sigma$ & Input alphabet ($\Sigma \subseteq \Gamma \setminus \{\epsilon\}$) & \cref{sec:computation_theory} \\
$\epsilon$ & Blank symbol on tape & \cref{sec:computation_theory} \\
$\delta$ & Transition function $\delta(q, \sigma) = (q', \sigma', d)$ & \cref{sec:computation_theory} \\
$p(n)$ & Polynomial time bound for verifier & \cref{sec:computation_theory} \\
$q_{\mathrm{acc}}$, $q_{\mathrm{rej}}$ & Accept / Reject states & \cref{sec:computation_theory} \\
$F$ & Set of final states  & \cref{sec:computation_theory} \\
$V_{j,t}^{q,\sigma}$ & Transition case at cell $j$, tier $t$, state $q$, symbol $\sigma$ & \cref{subsec:comp_graph} \\
$W$ & Computation walk (sequence of edges) & \cref{subsec:comp_graph} \\
$e_f$ & Final edge of a computation walk & \cref{subsec:feasible_graph_concept} \\
$e_t$ & Verification target edge for the footmarks extension & \cref{sec:walk_verification} \\
$E_f$ & Set of designated final edges & \cref{subsec:feasible_graph_concept} \\
$E_i$ & Edge slice  in $G$ with index $i$, i.e., $E_i = \{e \in E(G) \mid \mathsf{index}(e) = i\}$. & \cref{subsec:comp_graph}  \\
$C^E{(k)}$ & Step-extended component of depth $k$, constructed from $E_f$. &  \cref{subsec:feasible_graph_concept} \\
$\widehat{C}, C$ & Set of cover edges reachable via ceiling-adjacency from $E_f$. & \cref{subsec:feasible_graph_concept}\\
$\mathsf{MSEC}_G(E_f)$ & Maximal step-extended component of $E_f$ in $G$.& \cref{subsec:feasible_graph_concept} \\
$\indexOf(v)$ & Cell index (tape position) of node $v$ & \cref{subsec:comp_graph} \\
$\tier(v)$ & Tier (number of transitions at the cell of $index(v)$) & \cref{subsec:comp_graph} \\
$\state(v)$ & Current state of computation node $v$ & \cref{subsec:comp_graph} \\
$\symbol(v)$ & Current tape symbol at node $v$ & \cref{subsec:comp_graph} \\
$\lastState(v)$ & Former state at the cell of $index(v)$ before current node & \cref{subsec:comp_graph} \\
$\lastSymbol(v)$ & Former symbol at the cell of $index(v)$ before current node& \cref{subsec:comp_graph} \\
$\nextState(v)$ & State after transition from $v$ & \cref{subsec:comp_graph} \\
$\output(v)$ & Symbol written by transition from $v$ & \cref{subsec:comp_graph} \\
$\nextIndex(v)$ & Cell index moved to after transition from $v$ & \cref{subsec:comp_graph} \\
$\IPrec_G(v)$ & Index-precedent transition case of node $v$ & \cref{subsec:comp_graph} \\
$\ISucc_G(v)$ & Index-succedent nodes of $v$ in graph $G$ & \cref{subsec:comp_graph} \\
$\indexOf(e)$ & Edge index, i.e., $\min(\indexOf(u), \indexOf(v))$ for edge $e = (u,v)$ & \cref{subsec:comp_graph} \\
$\mathrm{init}(e)$ & Initial vertex of edge $e = (u, v)$, i.e., $u$ & \cref{subsec:comp_graph} \\
$\term(e)$ & Terminal vertex of edge $e = (u, v)$, i.e., $v$ & \cref{subsec:comp_graph} \\
$F(\mathcal{W})$ & Footmarks graph induced by walk set $\mathcal{W}$ & \cref{subsec:comp_graph} \\
$G^{(i)}$ & The $i$-th pruned graph in the pruning sequence. & \cref{subsec:redundant_futile_edge_detection} \\
$W_i$ & The $i$-th attempted walk selected from $G^{(i)}$, which is not feasible. & \cref{subsec:redundant_futile_edge_detection} \\
$R$&  \makecell[l]{A critical attempted walk (or graph) that removes all feasible walks \\during the pruning process.} & \cref{subsec:redundant_futile_edge_detection} \\
\bottomrule
\end{tabularx}
\label{tab:key_symbols}
\end{table}
\clearpage
\section{Turing Machine Computation Model}\label{sec:comp_model}

This section formalizes a geometric framework that maps Deterministic Turing Machine (DTM) operations onto a discrete graph structure. The objective is to transform temporal transition sequences into a spatial manifold, enabling a polynomial-time analysis of the entire certificate space of an $\mathsf{NP}$ verifier. The logical progression of this model is as follows:

\begin{enumerate}
    \item \textbf{Atomic Formalization:} We define \textbf{Computation Nodes} and \textbf{Transition Cases} to encode DTM configurations and local histories (tiers) into a 6-tuple topological unit.
    \item \textbf{Structural Confinement:} We prove that the \textbf{Footmarks Graph}---the union of all valid computation walks across all possible certificates---is strictly contained within a \textbf{Polynomial Bound} of $\bigO(p(n)^3)$. 
    \item \textbf{Proof Roadmap:} Leveraging this bound, we provide a strategic roadmap demonstrating how the $\mathsf{P}$ versus $\mathsf{NP}$ question is resolved through deterministic pruning and structural verification within this constrained space.
\end{enumerate}

By establishing these bounds upfront, we provide the structural guarantee that our subsequent verification procedures remain computationally tractable, shifting the problem from exhaustive search to deterministic geometric analysis.
\subsection{Concept of Computation Model}\label{subsec:comp_graph}

For each cell, we can consider the prior state $\vdown{q}$ and the prior tape symbol $\vdown{\sigma}$ before the current symbol is written by the last instruction executed (or transition occurred) at the cell in a DTM  including the current head position, current state, current tape symbol.
We also consider the number of times transitions occurred at each cell of the Turing machine.

Now we can define a computation node (or vertex) as a 6-tuple, where a directed edge between them denotes a transition or an instruction of the Turing machine.

\begin{definition}[Computation Node]\label{def:computation_node}
Given a deterministic Turing machine $M = (Q, \Sigma, \Gamma, \delta, q_0, F)$, a \textbf{computation node} is defined as a 6-tuple $v = (i, t, q, \sigma, \vdown{q}, \vdown{\sigma})$, representing a local configuration of a tape cell, where:
\begin{itemize}
    \item $i \in \mathbb{Z}$ representing the tape cell index, denoted by $\indexOf(v)$.
    \item $t \in \mathbb{N}_0$ representing the \textbf{tier}, the number of transitions that have occurred at cell $i$ prior to this configuration, denoted by $\tier(v)$.
    \item $q \in Q$ representing the state after the $t$-th transition at cell index $i$, denoted by $\state(v)$.
    \item $\sigma \in \Gamma$ representing the tape symbol after the $t$-th transition at cell index $i$, denoted by $\symbol(v)$.
    \item $\vdown{q} \in Q \cup \{\bot\}$ representing the former state right before the $t$-th transition at cell $i$, denoted by $\lastState(v)$. If $t=0$, $\vdown{q} = \bot$.
    \item $\vdown{\sigma} \in \Gamma \cup \{\bot\}$ representing the former tape symbol right before the $t$-th transition at cell $i$, denoted by $\lastSymbol(v)$. If $t=0$, $\vdown{\sigma} = \bot$.
\end{itemize} 
A node $v$ is called an \textbf{initial node} if $i=0, t=0, q=q_0, \vdown{q} = \bot$, and $\vdown{\sigma} = \bot$. The set of initial nodes $V_0$ denotes the starting configuration of $M$.

\textbf{Deterministic Projections:} 
Since $M$ is deterministic, the transition function $\delta$ is uniquely defined for the pair $(\state(v), \symbol(v))$. For each node $v$, we denote the resulting state, written symbol, and head direction as $\nextState(v)$, $\output(v)$, and $\dir(v)$, respectively, satisfying:
\[ \delta(\state(v), \symbol(v)) = (\nextState(v), \output(v), \dir(v)) \]
\end{definition}

\begin{remark}[Domain of Computation Nodes]
The set of all computation nodes $V$ is designed to encompass every theoretically possible logical configuration of a tape cell. Consequently, $V$ includes nodes representing configurations that cannot occur in any valid execution of $M$. By defining such a comprehensive domain, we can treat the identification of a valid computation as a filtering process (triming) over this potential configuration space.
\end{remark}

\begin{remark}[Clarification of Symbol Roles]
In a computation node $v = (i, t, q, \sigma, \vdown{q}, \vdown{\sigma})$, it is important to distinguish the roles of $\sigma$ and $\vdown{\sigma}$. The component $\sigma$ represents the \emph{current tape symbol}, which is the output written by the $t$-th transition at cell $i$. By contrast, $\vdown{\sigma}$ denotes the \textbf{pre-transition symbol} that was present in cell $i$ before the most recent transition occurred (i.e., before the transition that caused the tape head to leave the current cell).
In particular, $\vdown{\sigma}$ is \emph{not} the output of the most recent transition; rather, $\sigma$ is the last output symbol, while $\vdown{\sigma}$ records the former content of the cell.
\end{remark}

\begin{definition}[Transition Case] \label{def:transition_case}
A \textbf{transition case} $T$ (specifically $V_{i,t}^{q,\sigma}$) is defined as the set of all computation nodes $v \in V$ such that $\indexOf(v)=i$, $\tier(v)=t$, $\state(v)=q$, and $\symbol(v)=\sigma$.

\textbf{Deterministic Projections for Cases:} 
Since all nodes $v \in T$ share the same state $q$ and symbol $\sigma$, the deterministic projections are uniquely determined by the case $T$. We denote them as $\nextState(T)$, $\output(T)$, and $\dir(T)$, satisfying:
\[ (\nextState(T), \output(T), \dir(T)) = \delta(q, \sigma) \]
\end{definition}

\begin{remark}[Bounded Variability and Future Determinism]
The transition case $T$ functions as a structural unit that groups nodes based on their current configuration at $(i, t)$. This categorization highlights a key property of DTMs in our graph:
\begin{itemize}
    \item \textbf{History-Independent Future:} While nodes within a single case $T$ may represent different histories (i.e., different $\vdown{q}, \vdown{\sigma}$), their subsequent behavior—the state handed to the next cell, the symbol written to the tape, and the movement direction—is identical.
    \item \textbf{Internal Cardinality:} The number of nodes in each transition case is $|Q| \times |\Gamma|$ for $t > 0$, and exactly $1$ for $t = 0$.
    \item \textbf{Global Complexity:} For any coordinate $(i, t)$, there are exactly $|Q| \times |\Gamma|$ transition cases, ensuring that the local complexity of the graph is independent of the input length $n$. These constraints ensure that even though the computation graph explores a large configuration space, each coordinate $(i, t)$ is represented by a finite and constant number of states (bounded by $|Q|^2 \times |\Gamma|^2$), facilitating the pruning of invalid computation paths.
\end{itemize}
\end{remark}

\begin{definition}[Computation Graph]
Given a deterministic Turing machine $M$, the \textbf{computation graph} $G = (V, E)$ is a directed graph where $V$ is the set of all possible computation nodes as defined in Definition \ref{def:computation_node}. A directed edge exists from node $u$ to node $v$, denoted by $(u, v) \in E$, only if:
\[ |\,\indexOf(v) - \indexOf(u)\,| = 1 \]
When the context is clear, we refer to such a graph simply as a \emph{computation graph} $G$. When specified, we refer to $G$ as a \textbf{computation graph with a set of initial nodes $V_0$} to denote a computation graph whose initial nodes are only those in $V_0$.
\end{definition}

\begin{remark}[Finiteness and Complexity]
In general, a computation graph may be infinite. However, since we restrict our attention to halting Turing machines with time complexity $T(n)$ and space complexity $S(n)$, the resulting graph $G$ is established as a finite directed graph. The number of vertices is bounded by $S(n) \times T(n) \times |Q|^2 \times |\Gamma|^2$, ensuring the graph remains within a polynomial size relative to the input length and the number of transitions.
\end{remark}

\begin{definition}\label{def:width-height-definition}
Given a computation graph $G$, the \textbf{width} of $G$ (denoted by $w$) is defined as the difference between the maximum and minimum indices of the vertices in $G$:
\[ w = \max_{v \in V(G)} \indexOf(v) - \min_{v \in V(G)} \indexOf(v). \]
The \textbf{height} of $G$ (denoted by $h$) is defined as the maximum tier among all vertices in $G$:
\[ h = \max_{v \in V(G)} \tier(v). \]
\end{definition}

\begin{proposition}\label{prop:graph_size_bound}
Since edges exist only between nodes $u, v$ such that $|\indexOf(u) - \indexOf(v)| = 1$, the total number of edges $|E(G)|$ is bounded by:
\[ |E(G)| \le w \cdot h^2 \cdot |Q|^2 \cdot |\Gamma|^2, \]
where $Q$ is the set of states and $\Gamma$ is the tape alphabet of the underlying Turing machine.
\end{proposition}

\begin{definition}[Edge Index and Direction]\label{def:edge_index_dir} 
For an edge $e = (u, v)$ in a computation graph $G$, the \textbf{index of an edge} is defined as the minimum index of its endpoints: $\indexOf(e) = \min(\indexOf(u), \indexOf(v))$. The \textbf{direction of an edge} represents the head's displacement: $\dir(e) = \indexOf(v) - \indexOf(u)$.
\end{definition}

\begin{definition}[Edge Slice]\label{def:edge_slice}
For each integer $i \in \mathbb{Z}$, the \textbf{edge slice} $E_i$ is the set of all edges in $G$ with index $i$:
\[ E_i = \{ e \in E(G) \mid \indexOf(e) = i \} \]
An edge slice is treated as a formally indexed set $(i, E_i)$. In particular, $E_i$ remains a distinguished, index-specific slice even if $E_i = \emptyset$. 
\end{definition}

\begin{remark}[Physical Interpretation of Slices]
Geometrically, an edge slice $E_i$ represents the set of all possible transitions that cross the boundary between tape cell $i$ and cell $i+1$. Specifically, an edge $e=(u,v)$ belongs to $E_i$ if it corresponds to a move from $i$ to $i+1$ (where $\dir(e)=1$) or from $i+1$ to $i$ (where $\dir(e)=-1$). This construction allows us to analyze the computation as a flow of information across discrete spatial boundaries.
\end{remark}

\begin{definition}[Index-Predecessor] \label{def:index-predecessor}
Let $W$ be a walk in $G$. The \textbf{index-predecessor} of a node $v$ on $W$, denoted by $\ipred_W(v)$, is the last node $p$ appearing before $v$ on $W$ such that $\indexOf(p) = \indexOf(v)$. If $v$ is the first node on $W$ with that index, we define $\ipred_W(v) = \bot$.

Similarly, the \textbf{index-predecessor} of an edge $e$ on $W$, denoted by $\ipred_W(e)$, is the last edge $e_p$ appearing before $e$ on $W$ such that $\indexOf(e_p) = \indexOf(e)$. If no such edge exists, $\ipred_W(e) = \bot$.
\end{definition}

\begin{definition}[Index-Successor] \label{def:index-successor}
Let $W$ be a walk in $G$. The \textbf{index-successor} of a node $v$ on $W$, denoted by $\isucc_W(v)$, is the first node $s$ appearing after $v$ on $W$ such that $\indexOf(s) = \indexOf(v)$. If $v$ is the last node on $W$ with that index, we define $\isucc_W(v) = \bot$.

Similarly, the \textbf{index-successor} of an edge $e$ on $W$, denoted by $\isucc_W(e)$, is the first edge $e_s$ appearing after $e$ on $W$ such that $\indexOf(e_s) = \indexOf(e)$. If no such edge exists, $\isucc_W(e) = \bot$.
\end{definition}

\begin{definition}[Computation Walk] \label{def:computation_walk}
A \textbf{computation walk} $W = (v_0, v_1, \dots, v_n)$ is a sequence of computation nodes representing the step-by-step execution of a Turing machine $M$. Formally, for each step $k$, the node $v_k$ represents the 6-tuple configuration of the tape cell pointed to by the head after $k$ transitions. Equivalently, a walk $W$ is a computation walk if and only if for every node $v_k$, the following consistency conditions are satisfied:

\begin{enumerate}
    \item \textbf{Initial Vertex Condition:} $v_0$ is an \textbf{initial node} (i.e., $\indexOf(v_0) = 0$, $\tier(v_0) = 0$, and $\state(v_0) = q_0$, where $q_0$ is the initial state of the corresponding Turing machine $M$).
    \item \textbf{Tier Consistency:} The tier of $v_k$ reflects its sequential visit count to cell $\indexOf(v_k)$:
    \[ \tier(v_k) = \begin{cases} \tier(\ipred_{W}(v_k)) + 1 & \text{if } \ipred_{W}(v_k) \neq \bot \\ 0 & \text{if } \ipred_{W}(v_k) = \bot \end{cases} \]

    \item \textbf{State/Symbol History (Index-Predecessor Flow):} The node $v_k$ inherits the state and symbol recorded during the cell's previous visit:
    \[ (\lastState(v_k), \lastSymbol(v_k)) = \begin{cases} (\state(\ipred_{W}(v_k)), \symbol(\ipred_{W}(v_k))) & \text{if } \ipred_{W}(v_k) \neq \bot \\ (\bot, \bot) & \text{if } \ipred_{W}(v_k) = \bot \end{cases} \]
    \item \textbf{Transition Consistency:} The transition $\delta(\state(v_k), \symbol(v_k)) = (q', \sigma', d)$ at node $v_k$ uniquely determines the configuration of subsequent nodes in the walk:
    \begin{itemize}
        \item \textbf{Head Flow:} The resulting state $q' = \nextState(v_k)$ must match the state of the next node, $\nextState(v_k) = \state(v_{k+1})$ for $k < n$.
        \item \textbf{Cell Flow:} The written symbol $\sigma' = \output(v_k)$ must match the symbol read by the cell's next visitor, $\output(v_k) = \symbol(\isucc_{W}(v_k))$ if $\isucc_{W}(v_k) \neq \bot$.
        \item \textbf{Displacement:} The direction $d = \dir(v_k)$ must satisfy the spatial constraint $\dir(v_k) = \indexOf(v_{k+1}) - \indexOf(v_k)$.
    \end{itemize}
\end{enumerate}
\end{definition}

\begin{remark}[Edges as Transitions]
In this framework, each directed edge $e_k = (v_k, v_{k+1})$ in a computation walk represents a single transition step of $M$. The node $v_k$ provides a local snapshot of the cell currently scanned by the head. Specifically, for an instruction $I = (q, \sigma, \sigma', d, q')$ where $\delta(q, \sigma)=(q', \sigma', d)$; let $E_I$ be the set of edges compatible with $I$. An edge $(u, v)$ belongs to $E_I$ only if $u \in V_{i,t}^{q,\sigma}$ and $v \in \bigcup_{s \in \Gamma, t' \in \mathbb{N}_0} V_{i+\Delta, t'}^{q', s}$ for some $i \in \mathbb{Z}$ and $t \in \mathbb{N}_0$, where $\Delta$ is the spatial displacement (1 for $R$, $-1$ for $L$). This formally maps the state-transition logic of the Turing machine onto the structural connectivity of the computation graph.
\end{remark}

\begin{remark}[Computation Walk-Path Equivalence and Terminology] \label{rem:equivalence_of_computation_walk_path}
From the consistency conditions, it follows that for any vertex $v$ on a computation walk $W$, $\tier(\isucc_W(v)) = \tier(v) + 1$ and $\tier(v) = \tier(\ipred_W(v)) + 1$. This strict monotonicity of the tier attribute ensures that no vertex can be repeated in $W$. Consequently, every computation walk is structurally a \textbf{simple path}. In particular, a \textbf{maximal computation walk} is a maximal path that is a computation walk. This implies that a computation walk $W$ in a computation graph $G$ is maximal if no $W+e$ computation walk exists for any $e \in E(G) \setminus E(W)$. 

To avoid confusion with general connectivity paths in graph theory, we strictly adhere to the following terminology:
\begin{itemize}
    \item We use the term \textbf{computation walk} exclusively to denote a sequence that satisfies all consistency conditions in Definition~\ref{def:computation_walk}. 
    \item A walk in $G$ that violates any of these conditions is referred to simply as a \textbf{graph walk} or a sequence of edges, and is not considered a computation walk. 
    \item While the term \textbf{valid computation walk} may be used for emphasis, the adjective is technically redundant: by definition, an "invalid computation walk" is a contradiction in terms and does not exist within our framework.
\end{itemize}
Furthermore, since no vertex or edge is repeated, a computation walk $W$ can be treated as an induced subgraph of $G$ consisting of the vertex and edge sets of $W$. This allows for the direct application of graph-theoretic properties to a computation walk representing an individual execution.
\end{remark}

\begin{remark}[Determinism and Path Bijectivity]
Since $M$ is a deterministic Turing machine, a fixed input (and certificate) string uniquely determines the entire execution sequence. In the context of our graph, this implies that for a given starting configuration, there exists at most one valid computation walk of any particular length. Consequently, each maximal computation walk that terminates in a halting configuration corresponds bijectively to a valid deterministic execution of the machine for the provided input.
\end{remark}

\begin{definition}[Surface of a Computation Walk]\label{def:surface}
For a computation walk $W$, the \textbf{surface} $\mathcal{S}(W)$ is defined as a sequence of transition cases indexed by cell positions $i \in \mathbb{Z}$:
\[ \mathcal{S}(W) = \{ T_i \mid i \in \mathbb{Z} \} \]
where each $T_i$ represents the status of cell $i$ after the execution of $W$, satisfying:
\begin{itemize}
    \item $T_i = T$ if $T$ is the transition case of the \textbf{last node} $v \in W$ such that $\indexOf(v) = i$.
    \item $T_i = \bot$ if no node $v \in W$ satisfies $\indexOf(v) = i$ (i.e., cell $i$ has not been visited).
\end{itemize}

\end{definition}
\begin{remark}[Functional Role and Composition of the Surface]
The primary purpose of the surface $\mathcal{S}(W)$ is to track the evolving state of the tape and the transition history of each cell. It serves as a dynamic interface that bridges the completed computation with its potential future steps by maintaining the following information for each cell $i$:
\begin{itemize}
    \item \textbf{Transition History:} For visited cells, $\mathcal{S}(W)$ stores the transition case $T$ of the most recent visit. This provides the \textbf{current tape symbol} ($\sigma = \output(T)$) for future reads and the \textbf{history indicators} ($\state(T), \symbol(T)$) required to define the $\lastState$ and $\lastSymbol$ of the next node at that index.
    \item \textbf{Computational Frontier:} The surface explicitly distinguishes between visited and unvisited regions. A cell $i$ is unvisited if $T_i = \bot$, denoting regions that remain in their initial state. When a cell is visited for the first time, it enters the surface at \textbf{tier 0} (where $\tier(v)=0$ and $\lastState(v)=\lastSymbol(v)=\bot$).
\end{itemize}
Consequently, the surface corresponding to the final node of a walk $W$ contains precisely the necessary and sufficient information to determine the \textbf{immediate next transition}, acting as a compressed snapshot of the machine's operational status.
\end{remark}

\begin{figure}[t]
	\centering
	\includegraphics[width=1.0\columnwidth]{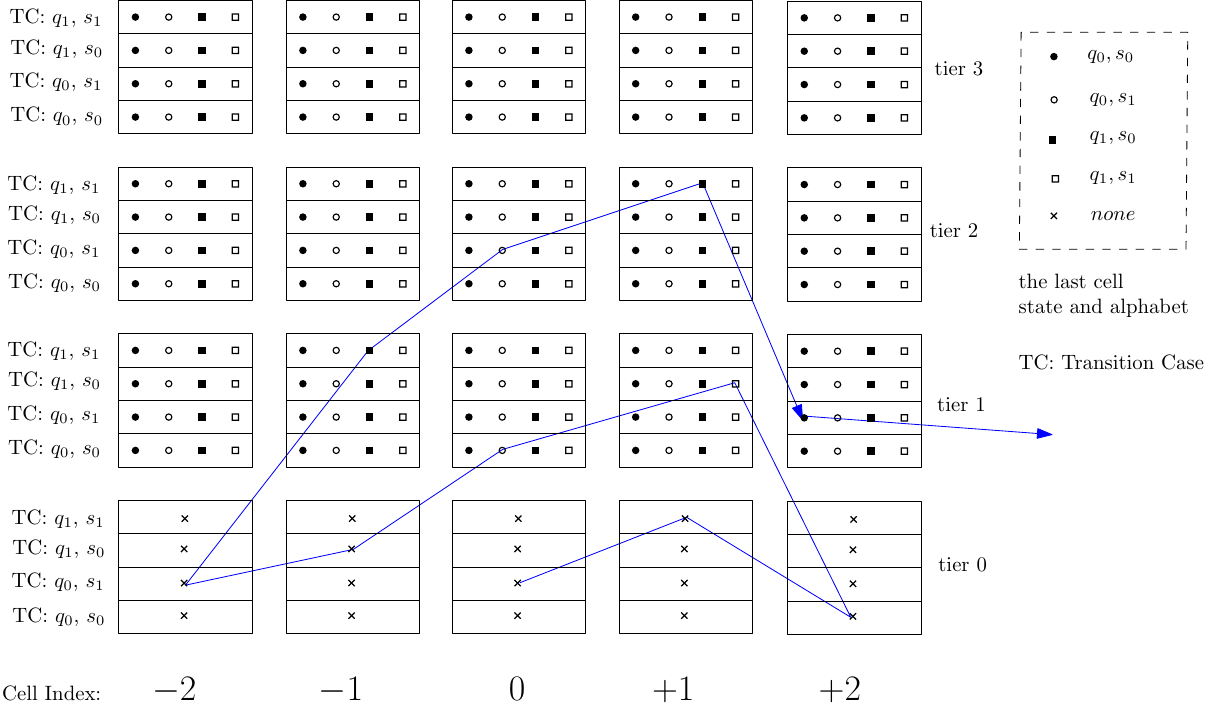}
	\caption{Turing machine Computation Model}
	\label{fig:computation_model}
	\Description{A two-dimensional representation of a computation graph organized as a grid where each cell corresponds to a specific index and tier.}
\end{figure}

Since a computation node encapsulates comprehensive information beyond just the machine state and tape symbol—specifically by incorporating temporal (tier) and historical (last state/symbol) data—it provides a sufficient basis for a complete simulation of a Turing machine. 
A formal demonstration of how computation walks are constructed upon these surfaces is provided in Appendix \ref{subsec:construct_computation_walk}.

\begin{definition}\label{def:footmarks_graph}
Let $\mathcal{W}$ be a set of computation walks. The \textbf{footmark graph} $F(\mathcal{W})$ is the graph defined by the union of all vertices and edges appearing in the walks of $\mathcal{W}$. Formally:
\[ 
V(F(\mathcal{W})) = \bigcup_{W \in \mathcal{W}} V(W), \quad E(F(\mathcal{W})) = \bigcup_{W \in \mathcal{W}} E(W). 
\]
When no ambiguity arises, we may \textit{simply} refer to $F(\mathcal{W})$ as the \emph{footmarks}. For an edge $e$ (or a set of edges $E$), we refer to $F(\mathcal{W}) + e$ (or $F(\mathcal{W}) + E$) as the \textbf{$e$-augmented} (or \textbf{$E$-augmented}) \textbf{footmarks}.
\end{definition}

\begin{definition}
A computation node $v$ is called a \textbf{folding node} (or \textbf{folding vertex}) if there exist two edges $e$ and $f$ incident to $v$, where $e$ is an incoming edge, $f$ is an outgoing edge, and $\indexOf(e) = \indexOf(f)$.
In this case, both $e$ and $f$ are called \textbf{folding edges} of $v$. \textit{If needed, we refer to $e$ as the incoming folding edge and $f$ as the outgoing folding edge of $v$.}
\end{definition}

\begin{remark}
The existence of folding nodes implies that for an edge $e = (u,v)$, its index-predecessor edge $\vdown{e} = (v', u')$ does not necessarily consist of the index-predecessor nodes of $v$ and $u$ in a straightforward manner. This discrepancy occurs because folding nodes allow the computation walk to traverse edges with the same index through a sequence of \textbf{"folded'' (direction-changing)} transitions, decoupling the immediate predecessor relationship from the global index-predecessor node relationship.
\end{remark}

The structural consistency of a computation walk implies that for any node $v$ with $\tier(v) > 0$, the state and symbol it inherits must originate from a unique prior configuration. While $\ipred_W(v)$ denotes the specific node in a walk, we can generalize this to the transition case from which $v$ must have descended.
\begin{definition}[Index-Precedent] \label{def:index-precedent_nodes}
The \textbf{index-precedent} of a computation node $v \in V(G)$ with $\tier(v) > 0$ is the unique transition case $P$ satisfying:
\[ (\indexOf(P), \tier(P), \state(P), \symbol(P)) = (\indexOf(v), \tier(v) - 1, \lastState(v), \lastSymbol(v)). \]
This unique case is formally denoted by $\IPrec_G(v)$, or simply $\IPrec(v)$ when the graph context is clear. Since $P$ is a transition case of a deterministic Turing machine, all nodes $u \in P$ share a unique transition rule, and consequently, all outgoing edges from $P$ share the same direction $\dir(P)$ determined by $\delta(\state(P), \symbol(P))$.
\end{definition}

\begin{remark}[Deterministic Structural Integrity]
The determinism of the transition function implies that for any node $v$ in the footmark graph $G = F(\mathcal{W})$, the symbol $\symbol(v)$ is not arbitrary; it must be the "written" result of its index-precedent. 
Specifically, $\symbol(v) = \output(P)$, where $P = \IPrec_G(v)$. This ensures that every node in the footmarks is logically consistent with the cell's history.
\end{remark}

\begin{lemma}[Properties of Index-Predecessor Edges] \label{lem:property_of_index_predecessor}
Let $e = (u,v)$ be an edge in a computation walk $W$ in a computation graph $G$.
Then the index-predecessor of $e$ on $W$ satisfies the following properties:
\begin{enumerate}
    \item The index-predecessor edge $\vdown{e}$ with $\indexOf(\vdown{e}) = \indexOf(e)$ must originate from a node in $\IPrec_G(v)$.
    \item The index-predecessor edge $\vdown{e}$ must have the opposite direction to $e$.
\end{enumerate}
\begin{proof}
\textbf{(1) Inclusion in $\IPrec_G(v)$:}  
Suppose, for the sake of contradiction, that there exists an index-predecessor $\vdown{e}=(v',u')$ of $e=(u,v)$ on $W$ such that $\indexOf(\vdown{e}) = \indexOf(e)$ but $v' \notin \IPrec_G(v)$.  
By the definitions of an index-predecessor node and a computation walk, let $w = \ipred_W(v) \in \IPrec_G(v)$ be the node immediately preceding $v$ in the sequence of nodes on the walk $W$.  
Since $W$ is a connected path, there must exist a subwalk $W_1$ from $v'$ to $w$ and a subsequent subwalk $W_2$ from $w$ to $u$.  

Because the indices of adjacent nodes in a computation walk differ by exactly $\pm 1$, the walk $W_2$ starting from $w$ (where $\indexOf(w) = \indexOf(v)$) must traverse an edge $e'$ with $\indexOf(e') = \indexOf(e)$ to reach $u$.  
However, as $W$ is a simple path, there is no overlap of edges between $W_1$ and $W_2$, implying $e' \neq \vdown{e}$. 
The existence of $e'$ after $\vdown{e}$ contradicts the assumption that $\vdown{e}$ is the \textit{index-predecessor} (the last such edge before $e$).  
Hence, all such edges $\vdown{e}$ must originate from a node in $\IPrec_G(v)$.

\textbf{(2) Direction of index-predecessor:}  
Let $e = (u,v)$ be an edge and its index-predecessor on $W$ be $\vdown{e} = \ipred_W(e) = (\vdown{v}, w)$, where $\vdown{v} = \ipred_W(v)$.  
Suppose, for contradiction, that $e$ and $\vdown{e}$ have the same direction. Let $i = \indexOf(u)$.  
Then $|\indexOf(\vdown{v}) - i| > 0$ while $\indexOf(u) - i = 0$, implying that a subwalk exists from $\vdown{v}$ toward $u$ along $W$.  

To reach the index $i$ of node $u$ from the position of $\vdown{v}$, the walk must traverse at least one edge $e'$ with $\indexOf(e') = \indexOf(e)$. 
However, the existence of such an edge $e'$ between $\vdown{e}$ and $e$ contradicts the assumption that $\vdown{e}$ is the \textit{index-predecessor} (the last such edge before $e$) on $W$.  
Therefore, $\ipred_W(e)$ must have the opposite direction to $e$.
\end{proof}
\end{lemma}
\begin{lemma}\label{lem:folding_node_property}
Given the footmark graph $G$ of a set of computation walks, if $v$ is a folding node, then all edges incident to $v$ share the same index.
\end{lemma}

\begin{proof}
First, observe that all outgoing edges from any node must have the same index. This is because each edge represents a transition at a specific tape cell in a deterministic Turing machine (DTM), and the direction of movement from a given configuration (state and symbol) is deterministic.
Next, we consider the incoming edges to node $v$.
 Let $e=(u,v)$ be a incoming edge of the same index with the outgoing folding edges. 

Suppose, for contradiction, that there exists another incoming edge $f = (w, v)$, the different index, from opposite direction to $e$ with:

\textbf{Case 1: $\tier(v) = 0$.} \\
A node $v$ with $\tier(v) = 0$ represents the first time a tape cell with $\indexOf(v)$ is visited in any computation walk within the footmarks $G$. By the structure of the graph, tier $0$ nodes cannot have incoming edges from nodes with a larger absolute index, as edge indices change by at most $\pm 1$ along a walk. Thus, all incoming edges of $v$ have the same index. Then, both $e$ and $f$ should income from smaller absolute index, which is contradiction to the assumption that $dir(e) \ne dir(f)$.

\textbf{Case 2: $\tier(v) > 0$.} \\ 
Let $W$ be a computation walk containing edge $e = (u, v)$. By the definition of a walk, there exists an index-predecessor edge $\vdown{e} = ({\vdown{v}}_1, u')$ such that $\indexOf(\vdown{e}) = \indexOf(e)$.
Then, there must exist an index-predecessor edge $\vdown{f} = ({\vdown{v}}_2, w')$ such that $\indexOf(\vdown{f}) = \indexOf(f)$ in some computation walk $W'$. Note that $\dir(e) \ne \dir(\vdown{e})$ and $\dir(f) \ne \dir(\vdown{f})$ by \cref{lem:property_of_index_predecessor}.

Since both ${\vdown{v}}_1$ and ${\vdown{v}}_2$ belong to $\IPrec_G(v)$. Since these nodes share the same state and symbol, their outgoing transitions must have the same direction and index due to the determinism of the transition function at the corresponding configuration. 
However, this implies that $\vdown{e}$ and $\vdown{f}$ must have the same index, which in turn implies $\indexOf(e) = \indexOf(f)$. This contradicts the assumption that $e$ and $f$ have different indices. This contradiction completes the proof.

\begin{figure}
	\centering
	\includegraphics{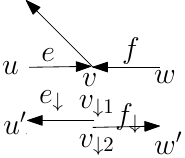}
	\caption{Incoming Folding Edge Direction with tier >0}
	\label{fig:incoming_folding_edge_direction}
\Description{
    A diagram illustrating the contradiction in Case 2 of the proof. 
    A node $v$ is shown with two incoming edges, $e$ and $f$, originating from different directions, and a single outgoing folding edge. 
    The diagram traces $e$ and $f$ back to their respective index-predecessors, $\vdown{e}$ and $\vdown{f}$, which originate from nodes ${\vdown{v}}_1$ and ${\vdown{v}}_2$ in the $\IPrec_G(v)$ set. 
    Because ${\vdown{v}}_1$ and ${\vdown{v}}_2$ represent the same DTM configuration, their outgoing transitions are forced to have the same index and direction. 
    The visual highlights that if $e$ and $f$ were to have different indices, it would necessitate their index-predecessors to also have different indices, which directly contradicts the determinism of the DTM at the index-predecessor level.
}
\end{figure}

\end{proof}

\begin{definition}\label{def:index-succedent_nodes}
The \textbf{index-succedent} of a computation node $v$ in a computation graph $G$ is defined as the set
\[
\ISucc_G(v) = \{ s \in V(G) \mid \symbol(s)=\output(v), v \in \IPrec_G(s)  \}.
\]
Equivalently, for each $s \in \ISucc_G(v)$, we have $v \in  \IPrec_G(s)$ and $\symbol(s)=\output(v)$.
When the context is clear, we may simply write $\ISucc(v)$ instead of $\ISucc_G(v)$.
 \end{definition}

\begin{definition}
Let $e$ be an edge of the computation graph, and let $\indexOf(e)$ denote its index as defined in \cref{def:edge_index_dir}.

\begin{itemize}
    \item A \emph{left-adjacent edge} of $e$ is an edge adjacent to $e$ whose index is $\indexOf(e) - 1$.
    \item A \emph{right-adjacent edge} of $e$ is an edge adjacent to $e$ whose index is $\indexOf(e) + 1$.
    \item A \emph{aligned-adjacent edge} of $e$ is an edge adjacent to $e$ whose index is exactly $\indexOf(e)$.
\end{itemize}
\end{definition}

\begin{definition}[Previous and Next Edges]
Let $G = (V, E)$ be a directed computation graph, and let $W = (e_1, e_2, \dots, e_k)$ be a computation walk in $G$. For any edge $e \in E$, we define the following:
\begin{itemize}
    \item \textbf{Walk-based previous and next edge:} \\
    If $e = e_i$ for some $1 < i < k$, then we define the previous and next edge within the walk $W$ as:
    \[
    \mathrm{prev}_W(e_i) = e_{i-1}, \quad \nextOf_W(e_i) = e_{i+1}.
    \]
    For boundary cases, $\mathrm{prev}_W(e_1)$ and $\nextOf_W(e_k)$ are undefined, which we denote as $\mathrm{prev}_W(e_1) = \bot$ and $\mathrm{next}_W(e_k) = \bot$. \\
    These denote individual edges or the null symbol $\bot$, and are written in lowercase.

    \item \textbf{Graph-based previous and next edges:} \\
    Regardless of any walk, we define the set of graph-adjacent edges:
    \[
    \Prev_G(e) = \{ e' \in E \mid \term(e') = \init(e) \}, \quad
    \Next_G(e) = \{ e' \in E \mid \init(e') = \term(e) \}.
    \]
    These denote the sets of incoming and outgoing edges adjacent to $e$ in the graph $G$, and are written in capitalized form to reflect their set-valued nature.
If the context is clear, $_G$ can be omitted.
\end{itemize}
\end{definition}

\begin{definition}\label{def:index-precedent_index-succedent_edges}
Let $(u,v)$ be an edge in a computation graph. Then:

\begin{enumerate}
	\item \textbf{Index-precedent edges:}  
	The \emph{index-precedent} of $e=(u,v)$, denoted by $\vdown{e}$, is the set of edges $\vdown{e} = (v', u')$ such that:
	\begin{itemize}
	    \item $v' \in \IPrec(v)$, and
	    \item there exists an \emph{IPrec-folding node chain} $(u_0, \dots, u_m)$ ($m \ge 0$) such that $u_0 = u$, $u_m = u'$, $u_{i+1} \in \IPrec(u_i)$ for all $0 \le i < m$, and each $u_i$ is a folding node for $0 < i < m$.
	\end{itemize}
	An edge $\vdown{e}$ is referred to as a \emph{direct index-precedent} if $m \le 1$ (equivalently, $u = u'$ or $u' \in \IPrec(u)$); otherwise, $\vdown{e}$ is referred to as an \emph{indirect index-precedent}.

    \item \textbf{Index-succedent edges:}  
	The \emph{index-succedent} of $e=(u,v)$, denoted by $\vup{e}$, is the set of edges $\vup{e} = (v', u')$ such that: 
	\begin{itemize}
		\item $u' \in \ISucc(u)$, and 
		\item there exists an \emph{ISucc-folding node chain} $(v_0, v_1, \dots, v_n)$ ($n \ge 0$) such that $v_0 = v, v_n = v'$, $v_{i+1} \in \ISucc(v_i)$ for all $0 \le i < n$, and each $v_i$ is a folding node for $0 < i < n$.
    \end{itemize}
	An edge $\vup{e}$ is referred to as an \emph{direct index-succedent} if $n \le 1$ (equivalently, $v = v'$ or $v' \in \ISucc(v)$); otherwise $\vup{e}$ is referred to as an \emph{indirect index-succedent}
\end{enumerate}

\begin{figure}
    \centering
    \includegraphics[width=0.8\textwidth]{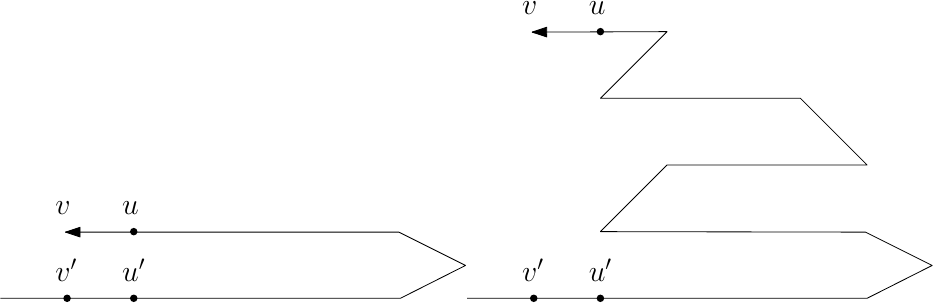}
    \caption{The Index-Precedent of an Edge}
    \label{fig:index-precedent_of_edge}
    \Description{A visualization of index-precedent edges. The left side illustrates a direct index-precedent where the edge $(u, v)$ is immediately preceded by $\vdown{e}$ in a single walk. The right side illustrates an indirect index-precedent through a folding node, demonstrating how a walk's history at the same tape index is preserved and linked across the footmark graph.}
\end{figure}

\begin{figure}
    \centering
    \includegraphics[width=0.8\textwidth]{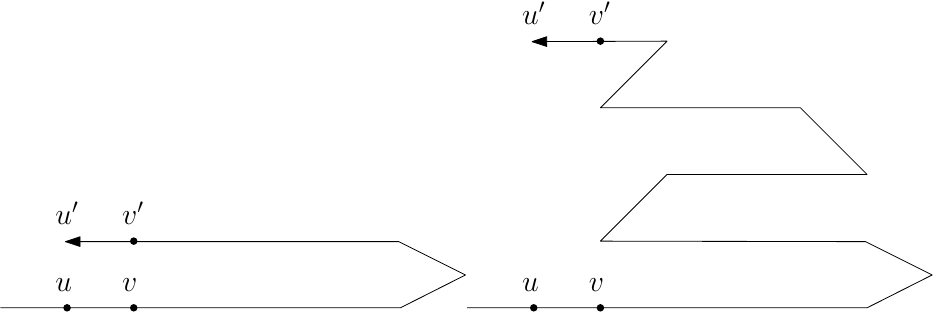}
    \caption{The Index-Succedent of an Edge}
    \label{fig:index-succedent_of_edge}
    \Description{A visualization of index-succedent edges. The left side shows a direct index-succedent where the next visit to the same tape index occurs immediately following the current edge's transition. The right side shows an indirect index-succedent via a folding node, highlighting the existence of paths where the next visit to the index is reached through shared nodes in the consolidated graph structure.}
\end{figure}
 For an edge $e$, we write $\IPrec_G(e)$ and $\ISucc_G(e)$ to denote
    its index-precedent and index-succedent sets.
    When the underlying graph $G$ is clear from context, the subscript $G$
    may be omitted.
\end{definition}
\begin{remark}The operators $\IPrec(\cdot)$ and $\ISucc(\cdot)$ are used for both
vertices and edges.
The intended meaning is determined by the argument type, and no ambiguity
arises in context.
\end{remark}

\subsection{Dynamic Computation Graph and Polynomial Footmarks Size} \label{subsec:dynamic_graph}
In this section, we establish the formal properties of the Dynamic Computation Graph and prove that the total size of its Footmarks remains polynomially bounded, despite the exponential number of possible certificates.

To decide a language $\mathcal{L} \in \mathsf{NP}$, the simulation must account for the union of all valid computation walks corresponding to every possible certificate $Y \in \Sigma^m$. We refer to this union of all these computation walks as the \textbf{Footmarks} of the verifier, as defined in \cref{def:footmarks_graph}. However, the precise spatial and temporal bounds of these computation walks—specifically the maximum tape head excursion and the total number of transitions—cannot be determined a priori without executing the simulation itself.
Consequently, rather than initializing a static, worst-case structure that may be computationally infeasible, we employ a graph that grows adaptively. This \textbf{Dynamic Computation Graph} ensures that we only allocate memory and processing time for configurations that are reachable under at least one valid certificate, allowing the simulation to scale naturally with the complexity of the verifier's behavior.

\begin{definition}[Dynamic Computation Graph]
Let $M$ be a deterministic Turing machine on input of length~$n$.  
The \textbf{dynamic computation graph} $G$ is constructed incrementally by the simulation algorithm.  
Each node represents a configuration of $M$, and each edge corresponds to a transition.  
An edge is added to $G$ only when it is visited or verified during the simulation.  
This approach avoids explicitly constructing the full computation graph in advance,  
which may be infeasible when its width and height are unknown.
\end{definition}

When implemented using dynamic array data structures, the expansion cost of the graph is amortized constant time, though it may reach linear time relative to its current size in the worst case.
Formally, the graph $G$ is organized as a two-level dynamic array structure:
\begin{itemize}
\item \textbf{First-level Array:} Indexed by tape cell positions (cell indices), this array expands as the tape head explores previously unvisited regions.
\item \textbf{Second-level Arrays:} Each entry in the first-level array points to a second-level dynamic array indexed by tiers, representing the successive computation layers for that specific cell.
\end{itemize}
This two-level architecture guarantees amortized constant-time access and insertion for nodes and edges within a fixed cell index.
 However, when a new tape cell index is accessed for the first time, the first-level array may require resizing. 
 In the worst case, this reallocates and copies all previously stored nodes and edges—whose total size is bounded by $\bigO(wh^2)$—resulting in a worst-case temporal overhead of $\bigO(wh^2)$.

Each node and transition case in the graph is also designed to encapsulate key computational properties:

\begin{itemize}
    \item For any transition case $T$ stored on the surface, the values $\nextState(T)$, $\output(T)$, and $\nextIndex(T)$ can be computed in \textbf{constant time}, without referring back to the transition function $\delta$.
    \item This is possible because $T$ already includes both the current state and the scanned symbol, and the transition rule for each $(q, s)$ pair is unique in a deterministic Turing machine.
    \item By encoding the results of $\delta(q,s)$ directly into each transition case object, or by storing the transition function $\delta$ itself within the node's local context, the simulation avoids repeated lookups and redundant recomputations during graph traversal.
\end{itemize}

To manipulate edges efficiently within the computation graph, each \textbf{edge slice} $E_i$ (as defined in \cref{def:edge_slice}) is implemented using a dynamic array or a hash map structure. This ensures that each slice remains a distinct, addressable unit, even if no edges currently exist at index $i$. Specifically, the implementation maps each node (at index $i$ or $i+1$) to an \texttt{AdjacencyList} object, which partitions the node's incident edges into four specialized categories:
\begin{itemize}
\item \texttt{left\_incoming}, \texttt{left\_outgoing}, \texttt{right\_incoming}, and \texttt{right\_outgoing}.
\end{itemize}
This representation enables near-constant-time identification of folding edges and ensures seamless connectivity between adjacent layers ($E_{i-1}, E_{i+1}$). Operations such as retrieving incident edges or verifying the existence of an edge between specific nodes are performed in time proportional to the graph's height $h$ (effectively $\bigO(p(n))$), thereby maintaining the overall polynomial efficiency of the simulation.

This structure supports the flexible, efficient, and scalable construction of the computation graph without requiring prior knowledge of the total tape length or computation depth.
By doing so, it preserves the theoretical guarantees of the verification procedure while remaining highly suitable for practical simulation or model-checking tasks.
Detailed implementation specifications for these data structures are provided in \cref{sec:appendix_computation_graph}.

With the dynamic computation graph and the data structures described above, it is evident that simulating the verifier for a particular certificate can be performed in polynomial time. 
The length of each individual computation walk is at most $\bigO(p(n))$, and more generally, it is bounded by the product of the width $w$ and the height $h$ of the computation.
The union of all computation walks for every possible certificate can be constructed in exponential time using a brute-force approach, as detailed in \cref{subsec:simulate_verifier_exp}.
However, our primary objective is to obtain the Footmarks of all computation walks efficiently by avoiding redundant simulations of overlapping edges. 
Since all computation walks originate from the same initial node and share a fixed tape region (as illustrated in \cref{fig:verifier_tape_area}) within the computation graph $G_M$ of the verifier $M$, such overlap is not only inevitable but serves as the fundamental basis for our polynomial-time optimization.
\begin{figure}
	\centering
	\includegraphics[width=0.8\textwidth]{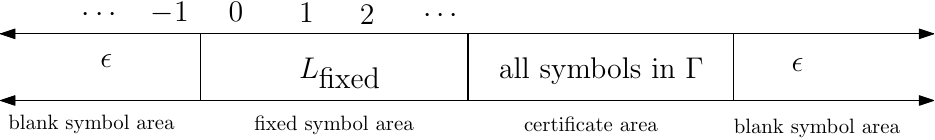}
	\caption{Turing Machine Tape Area For Verifier}
	\label{fig:verifier_tape_area}
	\Description{A linear Turing machine tape layout. From left to right: a blank region, a fixed tape area ($L_{\textrm{fixed}}) starting at index 0 representing the problem instance, a certificate area of length m for variable input, and a final blank region to the right.}
\end{figure}

\begin{lemma}[Polynomial Bounds of Footmarks of Computation Walks of All NP Certificates]
\label{lem:poly-bounded-graph}
Let $M = (Q, \Gamma, \delta, q_0, \qacc, \qrej)$ be a verifier Turing machine for an NP problem,
and suppose that for any input $L$ and any certificate of length at most $m$,
$M$ halts within $p(n)$ steps, where $n = |L| + m$ and $p$ is a polynomial.

Let $\mathcal{W}$ be the set of all computation walks of $M$ on all such certificates,
and let $F(\mathcal{W})$ be the footmark graph of $\mathcal{W}$.

Then the footmark graph $F(\mathcal{W})$ satisfies:
\begin{enumerate}
    \item its height $h$ is $\bigO(p(n))$,
    \item its width $w$ is $\bigO(p(n))$,
    \item the total number of vertices is $\bigO(p(n)^2)$,
    \item the total number of edges is $\bigO(p(n)^3)$.
\end{enumerate}
\begin{proof}
By \cref{def:footmarks_graph}, the footmark graph $F(\mathcal{W})$ is the graph defined by the union of all vertex sets and the union of all edge sets appearing in any computation walk in $\mathcal{W}$.
Hence, every computation walk of $M$ on any certificate is fully contained in
$F(\mathcal{W})$.

By the definitions of width and height given in \cref{def:width-height-definition}, we can bound the structural size of $F(\mathcal{W})$ as follows:

\begin{enumerate}[label=\textbf{\arabic*.}, wide=0pt]
    \item \textbf{Height.} 
    The tier of a computation node counts the number of transitions
    that have occurred at the corresponding tape cell prior to the configuration represented by the node.
    Since the verifier halts within $p(n)$ transitions on every computation walk,
    the tier value of any node appearing in any walk is at most $p(n)$.
    Therefore, the maximum tier over all nodes in $F(\mathcal{W})$ is bounded by
    $\bigO(p(n))$, and hence the height of $F(\mathcal{W})$ is $\bigO(p(n))$.

    \item \textbf{Width.}
    The width of the computation graph is defined as the range of tape indices
    visited during the computation.
    Since the verifier runs in time $p(n)$, the head can visit at most $p(n)$ distinct
    tape cells along any computation walk.
    Consequently, the width of $F(\mathcal{W})$, which is the maximum tape index range
    used by any walk in $\mathcal{W}$, is $\bigO(p(n))$.

    \item \textbf{Vertex bound.} 
	A node in the computation graph is identified by its tape index, tier, and the corresponding transition cases.  
	There are $|Q|\cdot|\Gamma|$ possible transition cases.  
	For tier-0 nodes, each transition case corresponds to exactly one node.  
	For higher tiers ($t \ge 1$), each transition case may give rise to at most $|Q|\cdot|\Gamma|$ nodes,  
	corresponding to all possible combinations of prior state and prior tape symbol.  
	Hence, at any fixed tier and tape index, the total number of nodes is bounded by
	\[
	(|Q|\cdot|\Gamma|) \times (|Q|\cdot|\Gamma|) = (|Q|\cdot|\Gamma|)^2 = \bigO(1),
	\]
	since $|Q|$ and $|\Gamma|$ are constants.  
	
	Therefore, the total number of vertices in $F(\mathcal{W})$ is bounded by
	\[
	\bigO(h \times w \times (|Q|\cdot|\Gamma|)^2) = \bigO(p(n)^2).
	\]
       \item \textbf{Edge bound.} 
       For any node $u$ in a computation graph, its outgoing edges must satisfy the spatial constraint $|\indexOf(u) - \indexOf(v)| = 1$. Since the target index has at most $h$ instantiated tiers, the out-degree of each vertex is bounded by $h$ . Therefore, for $G=F(\mathcal{W})$, the total number of edges is:
\[ |E(G)| \le |V(G)| \cdot h = \bigO(p(n)^2) \cdot \bigO(p(n)) = \bigO(p(n)^3). \]
\end{enumerate}

This completes the proof.
\end{proof}

\end{lemma}

\begin{remark}[Structural Density, Asynchronicity, and Polynomial Collapse]
The polynomial bound established in \cref{lem:poly-bounded-graph} represents a fundamental structural collapse of the exponential NP-verification space into a polynomially-sized geometric manifold. 

While the number of distinct certificates $Y \in \Sigma^m$ is exponential ($2^m$), the total number of available configurations in the computation graph is strictly limited. By the \textbf{Pigeonhole Principle}, a vast majority of these exponential computation walks must necessarily traverse overlapping vertices and edges. 

This structural density is further characterized by the \textbf{asynchronicity of tiers}: for an edge $(u, v)$, the tier of the terminal node $v$ is not necessarily related to the tier of the source node $u$. For instance, a tape head moving from a frequently visited cell (high tier) to a previously unvisited one will connect to a tier-0 node. Although global computation time always increases, these local tier counts capture "disparate local histories," allowing the graph to encode complex folding behaviors while remaining within the $\bigO(p(n)^3)$ bound. This ensure that any consistency-checking procedure over the verifier's entire certificate space remains tractable, ultimately bridging the gap between non-deterministic search and deterministic verification.
\end{remark}

To exploit the structural density and the resulting overlap of edges described above, we propose an incremental construction of the footmarks by extending edges one by one, thereby avoiding redundant simulations of overlapping trajectories. This approach necessitates a robust membership verification mechanism to determine whether a given edge is already contained within the existing footmarks. The following two sections are dedicated to the formal specification and analysis of this mechanism.

\section{Feasible Graph} \label{sec:feasible_graph}

To transition from exponential-time exhaustive search to a deterministic polynomial-time decision procedure, we introduce the \textbf{Feasible Graph}. This structure serves as the primary analytical tool for the pruning mechanism established in \cref{sec:roadmap_highlevel_proof}. The fundamental role of the feasible graph is to isolate and preserve all computation walks that terminate at a specific \textit{designated final edge}---the candidate edge currently under evaluation. Simultaneously, it provides a mechanism to systematically identify and eliminate edges that are mutually incompatible with any valid computation walk ending at said edge.

By refining the global computation graph into this feasible subgraph, we shift our analytical focus from the mere existence of paths to the \textbf{structural consistency} of the footmarks. To facilitate this refined analysis, we develop the following topological constructs:
\begin{itemize}
    \item \textbf{Ceiling- and Floor-edges:} These define the vertical boundaries of a computation walk within the tiered structure of the graph.
    \item \textbf{Step-pendant and Step-extended Components:} These formalize the structural building blocks required to define and construct the feasible graph.
    \item \textbf{Feasible vs. Infeasible Edges/Walks:} These categorize edges and walks based on their compatibility with the designated final edge, distinguishing valid execution traces from structurally inconsistent ones.
\end{itemize}

These concepts capture localized propagation and interaction patterns within the computation graph, providing the theoretical foundation for the polynomial-time simulation algorithm detailed in the subsequent sections.

\subsection{Feasible Graph on Computation Graph} \label{subsec:feasible_graph_concept}

In this subsection, we introduce key concepts and terminology related to the feasible graph. We begin by defining essential terms and the notion of step-extended components. Using these, we then formally define the feasible graph. 


\begin{definition} \label{def:floor_ceiling_edges}
Let $W$ be a computation walk in a computation graph.

An edge $e \in W$ is called a \textbf{ceiling edge} if it has no index-successor edge in the computation walk $W$.

An edge $e \in W$ is called a \textbf{floor edge} if it has no index-predecessor edge in the computation walk $W$.
\end{definition}

A floor edge corresponds to the first edge at a given edge index, and a ceiling edge corresponds to the last edge at a given edge index within a walk.  
Thus, we refer to floor or ceiling edges by their edge index, and the \emph{$i$-th ceiling edge} of $W$ refers to the unique ceiling edge with edge index $i$ in the walk $W$.

\begin{lemma}\label{lem:floor_edge_condition}
An edge $e = (u, v)$ is a floor edge if and only if $\tier(v) = 0$ for any computation walk.
\end{lemma}
\begin{proof}[Proof Sketch]
\begin{itemize}
\item \textbf{(If direction.)}
Assume that $\tier(v)=0$ but $e=(u,v)$ is not a floor edge.
Then there exists another edge whose terminal node is an index-predecessor of $v$.
However, any node incident to such an edge must lie on the same tape index as $v$, and therefore must have a tier smaller than that of $v$, implying a negative tier.
This contradicts the non-negativity of the tier index.

\item \textbf{(Only-if direction.)}
Assume that $e=(u,v)$ is a floor edge but $\tier(v) > 0$.
Then there exists a node of tier $0$ at the same tape index.
Along the computation walk from that tier-$0$ node to $v$,
an index-predecessor of $v$ must appear.
This contradicts the assumption that $e$ is a floor edge.
Therefore $\tier(v)=0$.

The complete formal proof is provided in Appendix
\cref{proof:lem:floor_edge_condition}.
\end{itemize} 
\end{proof}
\begin{remark}[Global Character of Floor Edges]
Although the notion of a \emph{floor edge} is defined locally with respect to a computation walk, 
\cref{lem:floor_edge_condition} shows that an edge $(u,v)$ is a floor edge in \emph{any} computation walk if and only if $\tier(v) = 0$.  
Therefore, the classification of floor edges is globally consistent and independent of any particular walk.  
This global characterization enables reasoning about floor edges without reference to specific walks.
\end{remark}

\begin{definition}[Ex-Pendant Edge] \label{def:ex_pendant_edge}
Let $e \in E(G)$ be an edge with index $i$ in the computation graph $G$.  
\begin{itemize}
    \item $e$ is \emph{left-pendant} if there exists no edge in $G$ adjacent to $e$ with index $i - 1$, and the node of $e$ with index $i$ is not a folding node.  
    \item $e$ is \emph{right-pendant} if there exists no edge  in $G$ adjacent to $e$ with index $i + 1$, and the node of $e$ with index $i+1$ is not a folding node.  
    \item $e$ is \emph{both-pendant} if it is both left-pendant and right-pendant.
    \item $e$ is \emph{ex-pendant} if it is either left-pendant or right-pendant.
\end{itemize}
\end{definition}

\begin{definition}[Ceiling-Adjacent Edge]\label{def:ceiling_adjacent}
Let $e = (v, w)$ be an edge in a computation graph $G$, and let $E_f$ be a set of designated final edges. 
An edge $f = (u, v')$ is said to be \emph{weakly ceiling-adjacent} to $e$ toward $E_f$ if there exists a vertex sequence $(v_0, v_1, \dots, v_n)$ such that:
\begin{itemize} 
    \item $v_0 = v'$ and $v_0$ is a non-folding node (the terminal node of $f$);
    \item for all $0 \le i < n$, $v_i \in \IPrec_G(v_{i+1})$;
    \item every vertex $v_i$ for $0 < i < n$ is a folding node;
    \item either ($v_n = v$ and $v$ is a folding node) or ($v_n = w$ and $e = (v, w) \in E_f$).
\end{itemize}

If, in addition to being weakly ceiling-adjacent, there exists a path from $f$ to $e$ in $G$ such that no edge in the path—other than $f$ itself—shares the same index as $f$, then $f$ is said to be \emph{ceiling-adjacent} to $e$ toward $E_f$.
 This includes cases where $e$ and $f$ are directly adjacent.
When the context is clear, we simply refer to $f$ as being ceiling-adjacent to $e$.

\begin{figure}[H]
  \centering
  \begin{subfigure}[t]{0.4\textwidth}
    \includegraphics[width=0.9\textwidth]{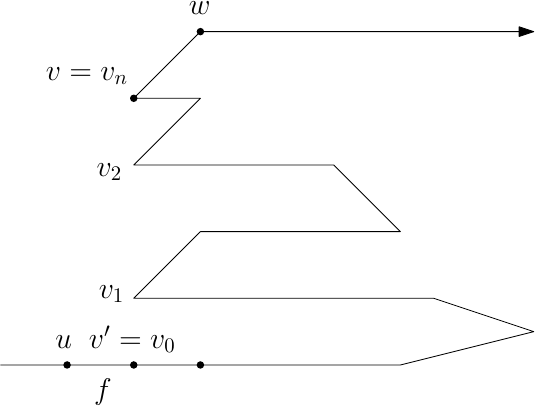}
    \caption{Ceiling-adjacent edge ($v$ is a folding node)}
    \label{fig:ceiling_adjacent_edge_folding_edge}
    \Description{Illustration of a weakly ceiling-adjacent relationship where the vertex sequence $(v_0, \dots, v_n)$ terminates at $v_n = v$. Here, $v$ is a folding node, reflecting a structural transition within the computation graph.}
  \end{subfigure}
  \hfill
  \begin{subfigure}[t]{0.4\textwidth}
    \includegraphics[width=0.9\textwidth]{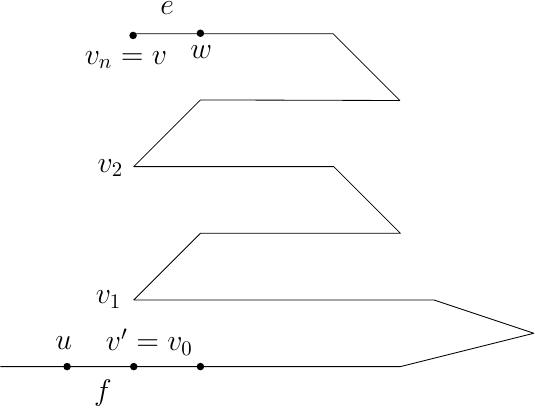}
    \caption{Ceiling-adjacent edge ($e \in E_f$)}
    \label{fig:ceiling_adjacent_edge_final_edge}
    \Description{Ceiling-adjacency toward a final edge $e = (v, w) \in E_f$. In this case, the sequence terminates at $v_n = w$, directly connecting the chain to the termination set of the computation walks.}
  \end{subfigure}
\end{figure}
\end{definition}

\begin{remark}
If $e \in E_f$, there may exist two distinct ceiling-adjacent edges $f$ satisfying
$\indexOf(f) = \indexOf(e) \pm 1$.
For non-final edges (i.e., $e \notin E_f$), any ceiling-adjacent edge $f$ satisfies
$\indexOf(f) = \indexOf(e) - \dir(e)$.
\end{remark}

\begin{lemma}[Characterization of Ceiling Edges] \label{lem:ceiling_edge_characterization}
Let $W$ be a computation walk in a computation graph $G$ with final edge $e_f$, and let $H = G[E(W)]$.
 An edge $e \in W$ is a ceiling edge if and only if it satisfies one of the following:
 \begin{enumerate}
 \item $e = e_f$, or
 \item $e$ is ceiling-adjacent in $H$ to another ceiling edge $e_c \in W$ such that $|\indexOf(e_f) - \indexOf(e_c)| < |\indexOf(e_f) - \indexOf(e)|.$
 \end{enumerate}
 In other words, an edge $e$ is a ceiling edge if and only if it is ceiling-adjacent in $H$ to a ceiling edge $e_c$ that is closer to the final edge $e_f$ in terms of index distance.
 \end{lemma}

\begin{proof}[Proof Sketch]
The proof proceeds by induction on the index distance $| \indexOf(e_f) - \indexOf(e) |$ within the walk $W$.
\begin{itemize}
\item \textbf{Base Case:} When $e = e_f$, it is trivially a ceiling edge as it has no index-successors in $W$.
\item \textbf{Inductive Step:} We assume the property holds for all edges closer to $e_f$. For an edge $e$, we show that its ceiling-adjacency to $e_c$ precludes any index-successors between them: in direct adjacency, no intermediate edges exist, and in indirect adjacency via folding nodes, any potential index-successor would require an impossible transition from a folding node back toward the index of $e$ (\cref{lem:folding_node_property}).
\end{itemize}
By establishing that each index admits at most one ceiling edge in a valid computation walk, we conclude that ceiling-adjacency uniquely propagates the "ceiling" status from the final edge $e_f$ back to the preceding stages. The complete formal induction is provided in \cref{app:proof_ceiling_edge}.
\end{proof}

This lemma serves to propagate the ceiling edge property backwards along a computation walk via ceiling-adjacency.

While floor edge is globally well-defined within the graph structure, independent of any local walk or surface assignment, ceiling edge is solely dependent on the computation walk it belongs to, 
as the same edge may function as a ceiling edge in one walk but not in another.
Thus, we require a graph-theoretic notion of ceiling edge that holds globally as well, rather than being dependent on a particular walk.
To formulate this, we introduce a family of edges that play a symmetric role to floor edges in a computation graph, not in a walk. 

\begin{definition}[Cover and Ex-Cover Edges] \label{def:cover_edges}
Let $G$ be a computation graph and $E_f$ be the set of designated final edges of the computation walks. 

Let $C = (c_0, c_1, \dots, c_k)$ be a finite sequence of edges such that $c_0 = e$ and $c_k \in E_f$. Depending on the connectivity and adjacency properties of this chain, we categorize $e$ as follows:

\begin{enumerate}
    \item $e$ is a \textbf{cover edge} if for every $j = 0, \dots, k-1$, the edge $c_j$ is \emph{ceiling-adjacent} to $c_{j+1}$. The set of all cover edges is denoted by $\widehat{C}$.
    
    \item $e$ is an \textbf{ex-cover edge} if the following broader conditions are satisfied:
    \begin{itemize}
        \item for every $j = 0, \dots, k-1$, $c_j$ is \textbf{weakly ceiling-adjacent} to $c_{j+1}$;
        \item for every $j = 0, \dots, k-1$, there exists a path in $G$ from $c_j$ to $c_k$.
    \end{itemize}
    The set of all such edges is denoted by $\widehat{C}_{ex}$.
\end{enumerate}
We refer to the sequence $C$ as a \emph{cover edge chain} or an \emph{ex-cover edge chain}, respectively.
\end{definition}

\begin{lemma}[Ceiling Edges are Cover Edges]  \label{lem:cover_edge_includes_feasible_ceiling_edges}
Let $\mathcal{W}_f$ be a set of computation walks in a computation graph $G$, and let $E_f$ be the set of designated final edges of all walks in $\mathcal{W}_f$.  
Let $C_f$ denote the set of ceiling edges that appear in walks from $\mathcal{W}_f$.

Then $C_f \subseteq \widehat{C}$, where $\widehat{C}$ is the set of cover edges of $G$ with respect to $E_f$, as defined in \cref{def:cover_edges}.
\begin{proof}[Proof Sketch]The inclusion follows directly from the structural characterization of ceiling edges in \cref{lem:ceiling_edge_characterization}. By definition, any ceiling edge $c \in C_f$ belongs to a walk that forms a backward chain of ceiling-adjacency reaching an element in $E_f$. Since the cover set $\widehat{C}$ is initialized with $E_f$ and is closed under the same backward ceiling-adjacency relations, every such $c$ is necessarily captured in $\widehat{C}$. The formal verification that these adjacency relations are preserved under graph superposition is provided in \cref{app:proof_lem_cover_edge}.\end{proof}
\end{lemma}

With the formalizations of floor and cover edges established, the following definition integrates these vertical (tier-based) and horizontal (index-based) boundary conditions.
This synthesis allows us to define the \textbf{step-pendant edge}, a critical construct used to identify edges that are structurally incapable of sustaining a valid computation walk within the computation graph.

\begin{definition}[Step-Pendant Edge] \label{def:step_pendant_edge}
Let $G$ be a computation graph with initial nodes $V_0$, and $E_f \subseteq E(G)$ be a set of designated final edges. Let $E_{\mathrm{init}}$ denote the set of all outgoing edges from $V_0$, and $\widehat{C}_{ex}$ be a given set of some ex-cover edges toward $E_f$ as defined in \cref{def:cover_edges}.

An edge $e = (u,v) \in E(G)$ is said to be \textbf{step-pendant with respect to $E_f$} if it satisfies any of the following conditions:
\begin{enumerate}
    \item \textbf{Horizontally Step-Pendant:} $e$ is an ex-pendant edge such that ($e \notin E_{\mathrm{init}} \cup E_f$)  or ($e$ is both-pendant and $e \in (E_{\mathrm{init}} \cup E_f) \setminus (E_{\mathrm{init}} \cap E_f)$);
    \item \textbf{Vertically Step-Pendant:} 
    \begin{itemize}
        \item $\IPrec_G(e) = \emptyset$ and $\mathrm{tier}(v) > 0$;
        \item $\ISucc_G(e) = \emptyset$, $e$ is not a cover edge toward $E_f$, and $e \notin \widehat{C}_{ex}$.
    \end{itemize}
\end{enumerate}
Henceforth, we denote an edge satisfying these conditions as an \textbf{$E_f$-step-pendant edge}. By default, $\widehat{C}_{ex} = \emptyset$ if not otherwise specified.
\end{definition}

The following definition incorporates vertical adjacency relationships to ensure consistency with index-based transitions:
\begin{definition}[Step-Adjacency] \label{def:step_adjacency}
An edge $e \in E(G)$ is said to be \textbf{step-adjacent} to an edge $f$ if and only if:
\begin{itemize}
    \item $e$ is adjacent to $f$ in the standard graph-theoretic sense, or
    \item $e \in \ISucc_G(f)$, or 
    \item $e \in \IPrec_G(f)$.
\end{itemize}
\end{definition}

To define the key notion of a feasible graph as a residual component obtained from another, we first introduce the concept of a step-extended component based on the above notion of step-pendant edges.

\begin{definition}[Step-Extended Component]
\label{def:step_extended_component}
Given a computation graph $G$ and a set of \textbf{base edges} $E_R \subseteq E(G)$, the \textbf{step-extended component} of $E_R$ is the subgraph $C_E \subseteq G$ defined recursively as follows:

Let $C_E^{(0)}$ be the subgraph of $G$ consisting of the edges in $E_R$ and their incident vertices. For each $i \ge 0$, define $C_E^{(i+1)}$ by adding to $C_E^{(i)}$ all edges $e \in E(G) \setminus E(C_E^{(i)})$ (and their incident vertices) such that:
\begin{enumerate}
    \item[(i)] $e$ is step-adjacent to some edge in $E(C_E^{(i)})$, and
    \item[(ii)] $e$ is a step-pendant edge in the subgraph $G - E(C_E^{(i)})$.
\end{enumerate}

The process terminates at step $k$ when no such edges can be added. The resulting subgraph $C_E^{(k)}$ is called the \textbf{maximal step-extended component} of $E_R$ in $G$, denoted by $\mathsf{MSEC}_G(E_R)$.
\end{definition}

\begin{definition}[Feasible Graph] \label{def:feasible_graph}
Given a computation graph $G$ with initial nodes $V_0$ and a set of designated final edges $E_f \subseteq E(G)$, let $E_{\mathrm{init}}$ denote the set of all outgoing edges from $V_0$. The set of \textbf{base edges} $E_R$ is defined as:
\[
E_R := \{ e \in E(G) \mid e \text{ is $E_f$-step-pendant}\}.
\]

The \textbf{feasible graph} $G_f$, denoted by $\mathsf{Feasible}(G)$, is obtained from $G$ by removing the edges and internal vertices of the maximal step-extended component $\mathsf{MSEC}_G(E_R)$, and any remaining isolated vertices $I(G)$:
\[
G_f := (G - E(\mathsf{MSEC}_G(E_R))) - I(G).
\]
\end{definition}

With notion of ceiling-adjacent edge and step-adjacent edge, the following notion of adjaceny also used to construct feasible graph algorithmically.

\begin{definition}[Index-Adjacent Edge]\label{def:index_adjacent_edges}
Let $G$ be a computation graph, $V_0$ the set of initial vertices, and $E_f$ the set of designated final edges.  
Let $E_i$ be an edge slice with index $i$, i.e., $E_i = \{ e \in E(G) \mid \indexOf(e) = i \}$.

An edge $f = (u, v)$ in $G$ is said to be \emph{index-adjacent} to $E_i$ with respect to $V_0$ and $E_f$ if it satisfies at least one of the following conditions:
\begin{itemize}
    \item $f$ is adjacent to an edge in $E_i$;
    \item $i = \indexOf(f) - \dir(f)$ and ($u$ is a folding node \textbf{or} $u \in V_0$);
    \item $i = \indexOf(f) + \dir(f)$ and ($v$ is a folding node \textbf{or} $f \in E_f$).
\end{itemize}
When the context is clear or the condition satisfied without $V_0$ and $E_f$, we simply say $f$ is index-adjacent to $E_i$.
\end{definition}

\SetKwFunction{ComputeFeasibleGraph}{ComputeFeasibleGraph}
\SetKwFunction{ComputeCoverEdges}{ComputeCoverEdges}
\subsection{Feasible Graph Construction Algortihm} \label{subsec:feasible_graph_construction}

We present the algorithm \ComputeFeasibleGraph{}, which constructs a feasible graph from a computation graph $G$, a set of start vertices $V_0$, and a set of designated final edges $E_f$. Here, $G$ is intended to be a subgraph of the augmented footmark graph corresponding to computation walks of a previously simulated deterministic machine, extended by the edges in $E_f$.

The output of the algorithm is guaranteed to satisfy the conditions of a feasible graph as defined in \cref{def:feasible_graph}.
Specifically, the algorithm removes all maximal step-extended components that do not contain any edge in $E_f$, yielding a subgraph
\[
G' = \ComputeFeasibleGraph{}
\]
that preserves all valid computation walks starting from $V_0$ and reaching the edges of $E_f$.

This algorithm uses a subroutine \ComputeCoverEdges{} to compute cover edges for each walk, which is invoked once at the beginning. We first present this subroutine and prove its correctness and time complexity. Then, we present the main algorithm and verify its correctness and complexity.

The following subalgorithm computes the cover edges defined in \cref{def:cover_edges}. 

\begin{algorithm}
\caption{Compute Cover Edges of subgraph of Footmarks of Computation Walks} \label{alg:compute_cover_edges}
\Input  {$G$: subgraph of footmark graph of walks, $E_f$: designated final edges}
\Output {Set $C$ of cover edges of $G$ with respect to $E_f$}
\Function{\ComputeCoverEdges{$G, E_f$}}{
\State{Let $C \gets E_f$} \label{alg_line:compute_cover_edges:initial_cover_edge}
\State{Let $Q \gets$ a queue with all the edge of $E_f$}
\While{$Q$ is not empty}{
    \State{Dequeue $e$ from $Q$}
    \State{Let $E_c$ be the set of ceiling-adjacent edge to  $e$}
    \ForAll{edge $f$ in $E_c \setminus C$}{
        \If{there exists a  $f$--$e$ path $P$ such that $\indexOf(f') \neq \indexOf(f)$ for all edges $f' \in P \setminus \{f\}$} {
            \State{Add $f$ to $C$ and Add $f$ to $Q$} \label{alg_line:compute_cover_edges:add_connected_cover_edge}
        }
    }
}
\State{\Return $C$}
}
\end{algorithm}

\begin{sublemma}[Time Complexity of Computing Cover Edges]
\label{lem:time_computing_cover_edges}
Let $h$ be the height of the computation graph and $w$ be its width.
Then the time complexity of the algorithm
\ComputeCoverEdges{} in \cref{alg:compute_cover_edges}
is bounded by $\bigO(h^5 w^2)$.
\begin{proof}
The computation graph contains at most $\bigO(h^2 w)$ edges.
Each edge is inserted into the queue $Q$ at most once, since once an edge
is added to the cover set $C$, it is marked and never re-enqueued.

Consider a single iteration of the while-loop when an edge $e$ is dequeued
from $Q$.
If $e$ is a folding edge, the algorithm computes the set of weakly
ceiling-adjacent edges to $e$.
By \cref{sublem:time_complexity_weakly_ceiling_adjacent_edges}, this step takes
$\bigO((h^2+\log w) \log h)$ time and returns at most $\bigO(h)$ edges.

For each such edge $f$, the algorithm checks whether there exists an
$f$--$e$ path $P$ such that
$\indexOf(f') \neq \indexOf(f)$ for all edges
$f' \in P \setminus \{f\}$.
This reachability test can be performed by a graph traversal over at most
$\bigO(h^2 w)$ edges, and therefore takes $\bigO(h^2 w)$ time per edge $f$.
Hence, the total cost of the reachability checks for a fixed edge $e$ is
$\bigO(h \cdot h^2 w) = \bigO(h^3 w)$.

Combining the above, the total work per edge processed from $Q$ is
\[
O((h^2+\log w) \log h  + h^3 w) = \bigO(h^3 w).
\]

Since at most $\bigO(h^2 w)$ edges are processed in the queue $Q$, the overall
time complexity of the algorithm is
\[
O(h^2 w) \cdot \bigO(h^3 w) = \bigO(h^5 w^2).
\]
\end{proof}
\end{sublemma}

The main algorithm proceeds as described in \cref{alg:feasible_graph}.

\begin{algorithm} \small
\SetKwFunction{SweepEdges}{SweepEdges}
\SetKwFunction{StepDownEdges}{StepDownEdges} 
\SetKwFunction{StepUpEdges}{StepUpEdges}
\caption{Feasible Graph of Subgraph of Extended Footmarks of Computation Walks} \label{alg:feasible_graph}
\Input  {$G$: a computation graph (subgraph of the extended footmark graph $F(\mathcal{W}) + E_f$), \\$V_0$: initial vertices, $E_f$: designated final edges}
\Output {Feasible graph $H$ of $G$ with respect to $E_f$}
\Function{\ComputeFeasibleGraph{$G, V_0, E_f$}} {
    \StateC{Let $C \gets$ \ComputeCoverEdges{$G, E_f$}} \Comment{Initialize cover edge set from designated final edges}
    \StateC{Let $H \gets G$ and $n \gets 0$} \Comment{Initialize feasible graph and previous edge count}
    \While(\tcc*[f]{Until no further changes or empty}){$|E(H)|>0$ \textbf{and} $n \ne |E(H)|$} {
        \StateC{Let $n \gets |E(H)|$} \Comment{Number of edges before sweep}
        \State{Let $i \gets$ (the minimum edge index of $E(H)$)}
        \StateC{$H \gets$ \SweepEdges{$H, C, V_0, E_f, i, +1$}}\Comment{Update graph by sweeping edges from left to right}
	\If{|E(H)|==0} {
		\State{\Return $H$}
	}
        \State{Set $i \gets$ (the maximum edge index of $E(H)$)}
        \StateC{$H \gets$ \SweepEdges{$H, C, V_0, E_f, i, -1$}}\Comment{Update graph by sweeping edges from right to left}
    }
    \StateC{\Return $H$}\Comment{Return the final feasible graph}
}
\Function{\SweepEdges{$G, C, V_0, E_f, i, d$}}{
    \StateC{Let $H \gets$ an empty graph}\Comment{Initialize empty graph to store feasible edges}
    \State{Let $E_i$ be the edge slice of $G$ with index $i$}
    \While{$E_i$ is not empty }{ 
        \State{Let $E_j$ be the edge slice of $H$ where $j=i-d$}
        \StateC{Let $I \gets$ \StepUpEdges{$G, H, E_i, E_j, E_f, V_0, E_f$}}\Comment{Expand edges upward from previous index layer}
        \StateC{Let $H \gets$ \StepDownEdges{$G, H, I, C_i$}} \Comment{Add edges downward to form feasible graph}
        \StateC{$i \gets i + d$}\Comment{Move to the next index in direction $d$}
        \State{Set $E_i$ be the edge slice of $G$ with index $i$}
    }
    \StateC{\Return $H$}\Comment{Return the constructed feasible graph}
}

\Function{\StepDownEdges{$H, I, C$}}{
\State{Let $Q$ be a queue with all the edges of $C \cap I$} \label{alg_line:step_down_edges:ceiling edge}
\State{Let $E_v \gets \emptyset$ and $I' \gets \emptyset$}
\WhileC{Step downward}{$Q$ is not empty} {
    \State{Dequeue $e$ from $Q$, and add $e$ to $E_v$}
    \State{Add $e$ to $I'$} \label{alg_line:step_down_edges:add_to_I_prime}
    \State{Enqueue all the edges of $\IPrec_{G}(e) \setminus E_v$ to $Q$}  \label{alg_line:step_down_edges:add_index-precedentss}
}
\State{Set $H \gets H \cup I'$}
\State{\Return $H$}
}

\Function{\StepUpEdges{$G, H, Es, Hs, V_0, E_f$}} {
\StateC{Let $Eb \gets$ all floor edges in edge slice $Hs$} \Comment{Edges $e=(u,v)$ in $Hs$ with  $tier(v)=0$, see \cref{lem:floor_edge_condition}}
\State{Let $Q$ be a queue with all the edges of $Eb$} \label{alg_line:stepupedges:compute_floor_edges}
\State{Let $E_v \gets \emptyset$ and $I \gets \emptyset$}
\WhileC{Step upward}{$Q$ is not empty}{
    \State{Dequeue $e$ from $Q$ and add $e$ to $E_v$}
    \If{$e$ is index-adjacent to edge slice $Hs$ for $V_0, E_f$ } {
        \State{Add $e$ to $I$} \label{alg_line:step_up_edges:add_to_I}\label{alg_line:step_up_edges:add_index_adjacent_edges}
        \State{Enqueue all the edges of $\ISucc_{G}(e) \setminus E_v$ to $Q$}\label{alg_line:step_up_edges:enqueue_index-succedent_edges}
    }
}
\State{\Return $I$}
}
\end{algorithm}

\begin{lemma}[Absence of Non-Designated Step-Pendant Edges] 
\label{lem:no_step_pendant_edge_in_feasible_graph}
Let $H = \mathsf{Feasible}(G)$ be the graph constructed from a computation graph $G$ with an initial vertex set $V_0$ via \cref{alg:feasible_graph}, with respect to the designated final edge set $E_f$. Let $C = \widehat{C}$ be the set of cover edges computed with respect to $E_f$. Then, $H$ contains no $E_f$-step-pendant edges.
\end{lemma}
\begin{proof}[Proof Sketch]
The proof establishes that \SweepEdges{} acts as a structural filter that preserves only edges with bidirectional connectivity. 

First, the algorithm ensures that no \textbf{horizontally step-pendant} edges remain in $H$. By \cref{sublem:index_adj_equivalence}, an edge is horizontally step-pendant if and only if it fails the index-adjacency criteria. Since the \StepUpEdges{} phase explicitly requires index-adjacency for inclusion (\cref{alg_line:step_up_edges:add_index_adjacent_edges}), such edges are systematically excluded from the forward-reachable set $I$. 

Second, the recursive propagation from anchor sets (floor edges and cover edges) prevents the inclusion of \textbf{vertically step-pendant} edges. An edge lacking an upward precedence ($\IPrec$) or downward succession ($\ISucc$) fails to be enqueued during the structural expansion of $I$ and $I'$, respectively. 

Consequently, the intersection $H = I \cap I'$ is guaranteed to be free of any structural discontinuities, whether sequential (horizontal) or hierarchical (vertical). The detailed exhaustive case analysis is provided in \cref{app:proof_lem_no_step_pendant}.
\end{proof}

\begin{sublemma}[No Pruning beyond Step-Pendant Edges]\label{lem:only_step_extended_component_removed}
Given a computation graph $G$ with an initial vertex set $V_0$ and a designated final edge set $E_f$, let $H = \mathsf{Feasible}(G)$ be the graph constructed via \cref{alg:feasible_graph}. Let $C=\widehat{C}$ be the set of cover edges. Then, any edge $e \in G \setminus H$ is contained in some step-extended component $C_E$ of $G$ based on $E_R$, where $E_R$ is the set of all $E_f$-step-pendant edge.
\end{sublemma} 
\begin{proof}[Proof Sketch]
This property establishes the precision of the algorithm by showing that any pruned edge $e \in G \setminus H$ must belong to some \textbf{step-extended component} $C_E$ rooted at the base set $E_R$. We establish this by contradiction. Suppose there exists an edge $e_0 \notin H$ that does not belong to any $C_E$. 

We analyze the iteration $k$ where $e_0$ is removed ($e_0 \in H^{(k-1)}$ but $e_0 \notin H^{(k)}$) through the following cases:

\begin{itemize}
    \item Case 1: $e_0$ is not step-pendant in $H^{(k-1)}$. 
    By \cref{sublem:index_adj_equivalence}, $e_0$ maintains bidirectional index-adjacency and possesses both an upward index-precedent ($\IPrec$) and a downward index-succedent ($\ISucc$), unless it is a floor or cover edge. In the \StepUpEdges{} and \StepDownEdges{} phases, such an edge necessarily satisfies the inclusion criteria and is added to $I \cap I' = H^{(k)}$. This directly contradicts the assumption that $e_0$ was removed.

    \item Case 2: $e_0$ is step-pendant in $H^{(k-1)}$ but not step-adjacent to any previously removed edge in $C^{(k-1)}$.
    In this case, $e_0$'s pendant status must be inherent to the original graph $G$. Thus, $e_0 \in E_R$, making it a \textbf{base edge} ($C_E^{(0)}$) of a step-extended component. This contradicts the assumption $e_0 \notin C_E$.

    \item Case 3: $e_0$ is step-pendant in $H^{(k-1)}$ and is step-adjacent to some $e' \in C^{(k-1)}$.
    If $e'$ already belongs to a step-extended component (by inductive hypothesis on previously removed edges), then $e_0$ satisfies the recursive definition of the same component ($C_E^{(i)} \to C_E^{(i+1)}$) because it is step-adjacent to a member and is step-pendant in $G \setminus C^{(k-1)}$. This again contradicts $e_0 \notin C_E$.
\end{itemize}
Since all cases lead to a contradiction, every removed edge must belong to a step-extended component. Detailed formal derivations and inductive steps are provided in \cref{app:proof_lem_only_step_extended}.
\end{proof}

\begin{corollary}[Equality between Extended Component and Removed Set] \label{cor:removed_is_step_extended}
The entire set of edges removed by \cref{alg:feasible_graph} constitutes a step-extended component rooted at $E_R$  where $E_R$ is the set of all $E_f$-step-pendant edges in $G$.
\end{corollary}
\begin{proof}
Let $C_{\text{Removed}} \subset E(G) \setminus E(H)$ be the set of edges removed up to a certain point in the algorithm, and assume $C_E = C_{\text{Removed}} \cup E_R$ is a step-extended component, Let $e$ be the next edge selected for removal. Since $e$ is removed by the algorithm, it must be $E_f$-step-pendant in the current graph $G - C_E$.

To show that $C_E' = C_E \cup \{e\}$ remains a step-extended component, we consider the following cases:
\begin{enumerate}
    \item If $e \in E_R$, then $e$ is already a base edge of the component. Thus, $C_E' = C_E$ trivially satisfies the definition of a step-extended component.
    \item If $e \notin E_R$, then $e$ was not $E_f$-step-pendant in the original graph $G$. Its $E_f$-step-pendant status in $G - C_E$ must therefore have been induced by the removal of its neighbors in $C_{\text{Removed}}$. This implies that $e$ is step-adjacent to at least one edge in $C_E$.
\end{enumerate}
Since $e$ is $E_f$-step-pendant in $G - C_E$ and is step-adjacent to $C_E$, the set $C_E \cup \{e\}$ satisfies the recursive construction of a step-extended component according to \cref{def:step_extended_component}. By induction, the final set of removed edges $E(G) \setminus E(H)$, being the union of all such $e$ and $E_R$, constitutes a step-extended component.
\end{proof}

\begin{lemma}[Correctness of Feasible Graph Construction Algorithm] \label{lem:constructing_feasible_graph}
Let $G$ be a computation graph with a set of initial vertices $V_0$, and $E_f$ the set of designated final edges. Let $H$ be the graph constructed by \cref{alg:feasible_graph} with respect to $E_f$, and let $G_f = \mathsf{Feasible}(G)$ be the feasible graph as defined in \cref{def:feasible_graph}. 
Then, $H = G_f$. 
Specifically, an edge $e \in E(G)$ is contained in $H$ if and only if it does not belong to the maximal step-extended component $\mathsf{MSEC}_G(E_R)$, where $E_R$ is the set of all $E_f$-step-pendant edges of $G$.
\end{lemma}
\begin{proof}
Let $C_{\text{removed}} = E(G) \setminus E(H)$ denote the set of edges removed by \cref{alg:feasible_graph}. 

First, by \cref{sublem:correctness_cover_edges}, the set of cover edges $\widehat{C}$ is correctly identified. Based on this, \cref{cor:removed_is_step_extended} ensures that $C_{\text{removed}}$ constitutes a step-extended component rooted at the base edges $E_R$. By definition, $E_R$ consists of all edges that are $E_f$-step-pendant in $G$.

Suppose, for the sake of contradiction, that $C_{\text{removed}}$ does not constitute the \emph{maximal} step-extended component $\mathsf{MSEC}_G(E_R)$. This implies that there exists some edge $e \in E(H)$ that is step-adjacent to an edge $e' \in C_{\text{removed}}$ and is $E_f$-step-pendant in the subgraph $H = G - C_{\text{removed}}$. 

However, by \cref{lem:no_step_pendant_edge_in_feasible_graph}, the graph $H$ constructed by the algorithm contains no $E_f$-step-pendant edges outside the anchor sets. Since $e$ is step-adjacent to a removed edge $e' \in C_{\text{removed}}$, it cannot be an anchor edge, contradicting the existence of such a step-pendant edge $e$ in $H$.

Therefore, $C_{\text{removed}}$ must be the maximal step-extended component stemming from $E_R$. It follows that $H = G - \mathsf{MSEC}_G(E_R)$, which exactly satisfies the definition of the feasible graph $G_f$.
\end{proof}

\begin{lemma}[Time Complexity of \ComputeFeasibleGraph{}] \label{lem:feasible_graph_time_complexity}
Let $G$ be the input computation graph with width $w$ and height $h$, and $m = |E(G)| = \bigO(wh^2)$ be the number of edges. The worst-case time complexity of \ComputeFeasibleGraph{} is:
\[
T_f = O\bigl(w^{2} h^{4}(h \log h + \log w)\bigr).
\]

\begin{proof}
The complexity consists of a one-time initialization followed by an iterative refinement process:
\begin{itemize}
    \item \textbf{Initialization}: \ComputeCoverEdges{} runs once in $\bigO(h^{5} w^{2})$ time.
    \item \textbf{Iterative Refinement}: Each iteration of the refinement loop removes at least one edge, leading to at most $m = \bigO(wh^2)$ iterations. Each iteration invokes \SweepEdges{} twice (bidirectionally), with each call costing $\bigO(w h^2 (h \log h + \log w))$ as per \cref{lemma:sweepedges_complexity_final}.
\end{itemize}
Summing these costs, we obtain:
\[
T_f = \bigO(h^5 w^2) + O\bigl(m \cdot (wh^2(h \log h + \log w))\bigr).
\]
Substituting $m = \bigO(wh^2)$ into the iterative term:
\[
O(wh^2 \cdot wh^2(h \log h + \log w)) = \bigO(w^2 h^4 (h \log h + \log w)).
\]
Since the initialization cost $\bigO(h^5 w^2)$ is dominated by $\bigO(w^2 h^5 \log h)$, the total complexity simplifies to:
\[
T_f = O\bigl(w^{2} h^{4}(h \log h + \log w)\bigr).
\]
Given that both $w$ and $h$ are polynomially bounded by the input size, the algorithm's execution time is strictly polynomial.
\end{proof}
\end{lemma}
Throughout the remainder of the paper, $T_f$ denotes the worst-case running time of \ComputeFeasibleGraph{}.

\subsection{Feasible Walk Preservation}
In this subsection, we focus on the structural properties of the feasible graph constructed by the previous algorithm. We introduce several key concepts related to walks and edges within this graph, which are essential for understanding how the feasible graph captures valid computation paths.

We distinguish different types of computations walks and edges according to their relationship with feasibility and their role in the overall graph structure. These distinctions will be crucial for subsequent analysis and proofs, particularly in establishing that all feasible walks ending with the designated final edges are fully contained in the feasible graph.

The formal definitions follow below, after which we proceed with important lemmas that characterize the properties of these walks and edges.

Now we can define the category of walks and edges in the feasible graph.

\begin{definition}[Classification of Walks and Edges] \label{def:classification_walks_edges}
Let $G$ be a feasible graph with respect to a set of designated final edges $E_f$. We classify the computation walks and edges in $G$ as follows:

\begin{itemize}
    \item A \textbf{feasible walk} is a computation walk in $G$ that contains at least one edge from $E_f$.
    \item A \textbf{feasible edge} is any edge that belongs to at least one \textit{feasible walk}.
    \item An \textbf{embedded walk} is a maximal computation walk in $G$ that is not a feasible walk but consists entirely of \textit{feasible edges}. (These represent valid prefix paths that fail to reach $E_f$ within the current graph $G$ but are composed of structurally necessary edges.)
    \item An \textbf{obsolete walk} is a maximal computation walk in $G$ that is not a feasible walk and contains at least one non-feasible edge. 
    \item An \textbf{obsolete edge} is an edge that belongs to an \textit{obsolete walk} but is not a \textit{feasible edge}.
    \item An \textbf{orphaned edge} is an edge in $G$ that does not belong to any valid computation walk (neither feasible nor obsolete).
\end{itemize}

For precision, we use the phrase \textbf{with respect to $E_f$} to denote this classification. When $E_f = \{e_f\}$ is a singleton, we refer to these as walks \textbf{to $e_f$}. For consistency, the term \textbf{feasible walk} and \textbf{feasible edge} is applied not only to the computed feasible graph but also to the original computation graph from which the feasible graph is derived.
\end{definition}

It is clear that the categories of walks defined above are mutually disjoint. The same holds for the corresponding categories of edges.

\begin{remark}[Context-Dependency of Classifications] \label{rem:context_dependency}
The classification of computation walks and edges exhibits a fundamental distinction in their dependence on the target edge set $E_f$ (designated as the final edges):
\begin{itemize}
    \item \textbf{Goal-Oriented Categories:} The definitions of \emph{feasible} and \emph{obsolete} walks (and their corresponding edges) are strictly relative to the designated final edge set $E_f$. They distinguish between trajectories that successfully certify a computation and those that fail to reach the target within the current graph $G$.
    \item \textbf{Structural Category:} In contrast, the definition of an \textbf{orphaned edge} is independent of $E_f$. An edge is orphaned if it cannot participate in \emph{any} valid computation walk, regardless of its eventual destination. This represents a more fundamental, structural invalidity where the local configuration cannot be part of any consistent execution trace.
\end{itemize}
\end{remark}

\begin{figure}[htbp]
    \centering
    \begin{subfigure}[b]{0.45\textwidth}
        \centering
        \includegraphics[width=\textwidth]{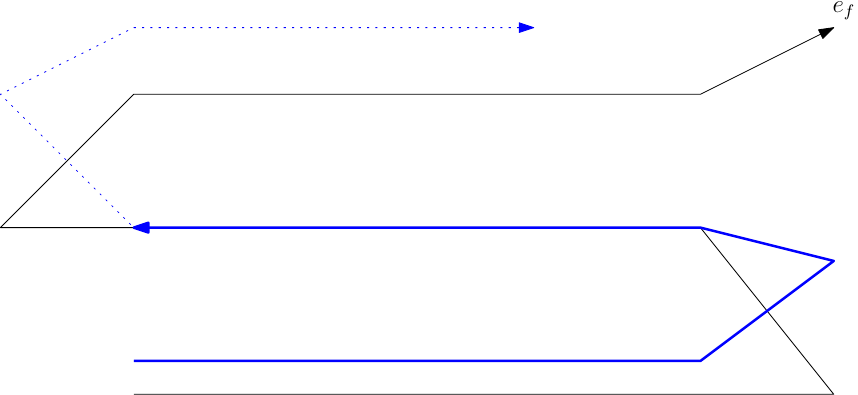}
        \caption{Obsolete Walk and Edge}
        \label{fig:obsolete_walk_and_edge}
        \Description{An obsolete edge starts on a feasible walk—a path terminating at the designated final edge—but deviates and ends prematurely at a vertex that is not part of any feasible trajectory, representing a structural dead end.}
    \end{subfigure}
    \hfill
    \begin{subfigure}[b]{0.45\textwidth}
        \centering
        \includegraphics[width=\textwidth]{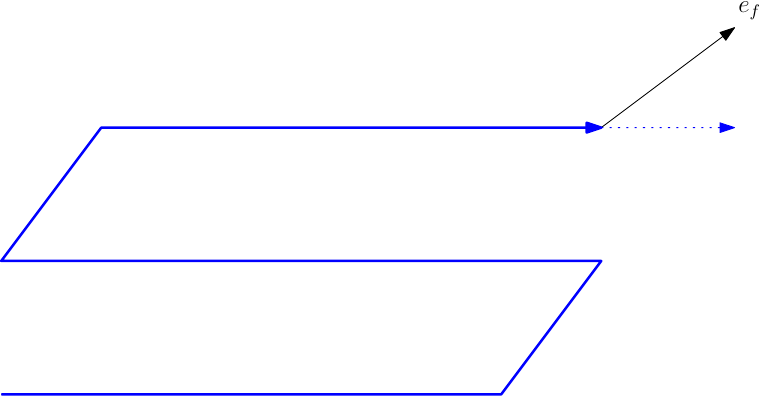}
        \caption{Embedded Walk}
        \label{fig:embedded_walk}
        \Description{An embedded walk consists of edges that are entirely subsumed within a feasible walk. It terminates at a vertex already visited by a feasible trajectory, showing that while the walk itself is truncated, its edges do not violate the feasible graph's structure.}
    \end{subfigure}

    \vspace{0.5cm}

    \begin{subfigure}[b]{0.45\textwidth}
        \centering
        \includegraphics[width=\textwidth]{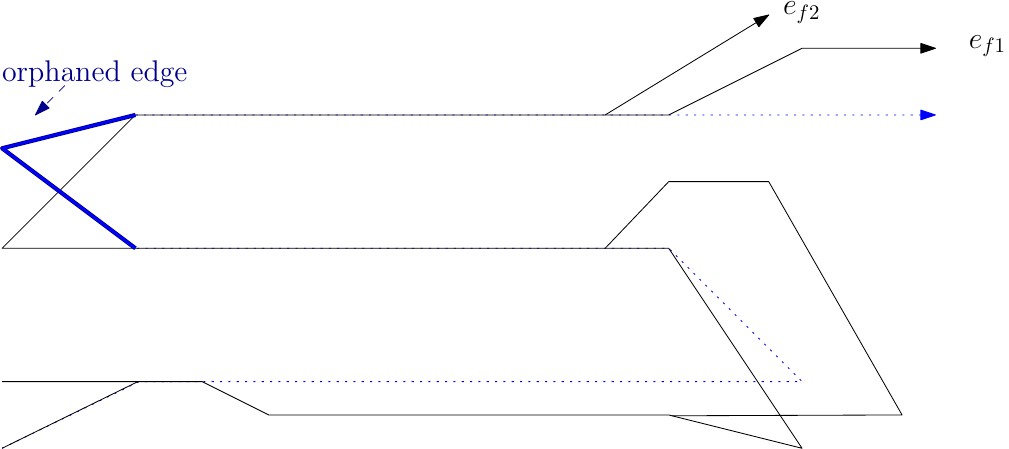}
        \caption{Orphaned Edge}
        \label{fig:orphaned_walk_and_edge}
        \Description{An orphaned edge begins at an intermediate vertex of a feasible walk, diverges into a sequence of non-feasible edges, and subsequently re-converges at a later vertex on a feasible walk. Despite the re-connection, the divergent segment remains inconsistent with the designated final edge.}
    \end{subfigure}

    \caption{Illustrations of different types of walks and edges in a feasible graph: obsolete, embedded, and orphaned.}
    \label{fig:walk_edge_types}
\end{figure}

We introduce the following types of edges, which are used to manipulate and analyze the graph structure in the process of re-constructing a feasible graph.

\begin{definition}
A \textbf{merging edge} is an incoming edge to a computation node $v$ whose in-degree is greater than $1$ and out-degree is not $0$ in a computation graph.
A \textbf{splitting edge} is an outgoing edge from a computation node $v$ whose out-degree is greater than $1$ and in-degree is not $0$ in a computation graph.
\end{definition}

\textbf{Note.} The notion of a \emph{feasible walk} plays a central role in both pruning the computation graph and verifying the validity of computation paths. As we progressively remove infeasible portions of the graph, feasible walks serve as the structural backbone that preserves all valid computation paths. This makes them essential in the design of our deterministic simulation of the verifier. The following lemma shows that the feasible graph always preserves any computations walks ending with the given designated final edges.

\begin{lemma}[Preservation of Feasible Walks] \label{lem:preservation_of_feasible_walks}
Let $G$ be a computation graph with a set of initial nodes $V_0$, and $E_f \subseteq E(G)$ a set of designated final edges. Let $G_f = \mathsf{Feasible}(G)$ be the feasible graph defined in \cref{def:feasible_graph}. Then, every computation walk $W$ in $G$ that starts at some $v \in V_0$ and ends with an edge $e \in E_f$ is fully contained in $G_f$ and remains a feasible walk in $G_f$.
\end{lemma}

\begin{proof}
Suppose, for the sake of contradiction, that there exists a computation walk $W = (e_0, e_1, \dots, e_k)$ such that $e_0 \in E_{\mathrm{init}}$ and $e_k \in E_f$, but $W \not\subseteq G_f$.
Since $G_f$ is a feasible graph, it is by definition the result of recursively removing step-extended components. 
Any edge $e \in G \setminus G_f$ must therefore have been an $E_f$-step-pendant edge at the moment of its removal. 
Let $e_i$ be the \textbf{first edge} of $W$ to be removed during the construction of the feasible graph of $G$ with respect to $V_0$ and $E_f$. Let $G'=G$ if $e_i$ belongs to the base edges of the maximal step-extended component; otherwise, let $G'=G-C_E$, where $C_E$ is the set of edges in the step-extended components removed right before the inclusion of $e_i$.

By the definition of the feasible graph (\cref{def:feasible_graph,def:step_extended_component}), only $E_f$-step-pendant edges can be removed from $G'$, regardless of whether $G'=G$ or not. We check the conditions of \cref{def:step_pendant_edge} for $e_i$:

\begin{enumerate}
    \item \textbf{Ex-pendant condition}: In a computation walk, only the initial edge $e_0$ and the final edge $e_k$ can be ex-pendant. However, $e_0 \in E_{\mathrm{init}}$ and $e_k \in E_f$ are explicitly excluded from the $E_f$-step-pendant edge set and thus cannot be the first removed edge. For any $0 < i < k$, $e_i$ is connected to $e_{i-1}$ and $e_{i+1}$, so it is not ex-pendant.
    
    \item \textbf{Precedent condition}: 
If $\tier(v) > 0$ for $e_i = (u,v)$, it must have an index-precedent edge in $W$ (as it is a computation walk). Since $e_i$ is the \textit{first} edge of $W$ to be removed, its index-precedent still exists in $G'$, hence $\IPrec_{G'}(e_i) \neq \emptyset$. 
If $\tier(v) = 0$, $e_i$ is a floor edge, which by definition (\cref{def:step_pendant_edge}) does not satisfy the step-pendant condition even though it has no index-precedent edges.
    
    \item \textbf{Succedent/Cover condition}: Similarly, if $e_i$ is not a ceiling edge, it has an index-successor in $W$ which still exists in $G'$, so $\ISucc_{G'}(e_i) \neq \emptyset$. If $e_i$ \textit{is} a ceiling edge, it is a cover edge with respect to $E_f$ by \cref{lem:cover_edge_includes_feasible_ceiling_edges}. Since cover edges (and ex-cover edges $\widehat{C}_{ex}$) are protected in the step-pendant definition, $e_i$ cannot satisfy the third condition.
\end{enumerate}

In all cases, $e_i$ cannot be a step-pendant edge at the time of its removal. This contradicts the definition of $\mathsf{MSEC}$. Therefore, no edge of $W$ can be removed, and the entire walk is preserved in $G_f$.
\end{proof}

\begin{remark}[Structural Persistence of Obsolete and Orphaned Edges]
The persistence of obsolete and orphaned edges in the feasible graph $G_f$ of a computation graph $G$ stems from the structural configuration where the cover edge set is defined more broadly than the ceiling edge set. Consequently, an edge may not be identified as step-pendant even if it lacks a valid successor within the ceiling constraints, allowing such computationally invalid structures to remain structurally supported.
While \cref{lem:preservation_of_feasible_walks} guarantees that all valid computation walks to the designated final edges are preserved within $G_f$, the converse does not hold: the mere inclusion of a designated final edge in $G_f$ does not guarantee the existence of a feasible walk. In other words, while the collapse of a structure in $G_f$ confirms the non-existence of a feasible walk, its survival provides a necessary but insufficient condition for feasibility. This gap necessitates the Deterministic Pruning Mechanism introduced in the next section to verify whether a designated final edge actually forms a complete feasible walk.
\end{remark}

\subsection{Total Collapse}
In this section, we establish the `Total Collapse' of the feasible graph structure under constrained graph conditions. 
By demonstrating that certain elimination of vital edges necessarily triggers a cascade of structural contradictions—mediated by the topological constraints of grid-aligned edges and the non-existence of step-pendant edges—we prove that the system must reduce to an empty state upon the removal of particular essential edges that belong to all feasible walks.

\begin{definition}[Certificate-Oblivious Verifier TM]
A verifier Turing machine $M$ is \textbf{certificate-oblivious} with respect to a problem instance $L_{fixed}$ if the tape head trajectory $H_M(L_{fixed}, Y, t)$ is uniquely determined by the certificate length $m = |Y|$ and the time step $t$, 
independent of the content of $Y$, for all $t$ up to the halting time $T(L_{fixed}, Y)$.
Formally, for any two certificates $Y, Y'$ with $|Y| = |Y'|$, the following holds:
$$H_M(L_{fixed}, Y, t) = H_M(L_{fixed}, Y', t) \quad \text{for all } t \le \min(T(L_{fixed}, Y), T(L_{fixed}, Y'))$$
\end{definition}

\begin{remark}
\begin{itemize}
\item Scope of Obliviousness: The prefix "Certificate-" denotes that the obliviousness property is specific to the certificates, allowing for potential dependence on the fixed problem instance $L_{fixed}$.
\item Generalization: Any standard Oblivious Turing Machine (OTM) is inherently a Certificate-Oblivious Verifier TM, as its trajectory is strictly independent of the input content $Y$ and halting time is generally fixed or independent of $Y$.
\end{itemize}
\end{remark}
By the properties of Certificate-Oblivious Verifier TM, the head position $Pos(t)$ is invariant across all computational branches.
 Given our Grid-aligned mapping, this invariance implies that all nondeterministic choices are mapped onto identical spatial coordinates in the footmark graph. Consequently, the set of all potential computation paths collapses into a single, deterministic trajectory.

\begin{definition}[Grid-Aligned Footmarks]
A footmark graph $G = (V, E)$ is said to be \textbf{Grid-Aligned} if for any two computation walks $W_1, W_2 \in \mathcal{W}$ and any step $k$ such that $|W_1| \ge k$ and $|W_2| \ge k$, 
the vertices $v_k \in W_1$ and $u_k \in W_2$ satisfy:
$$(\tier(v_k), \indexOf(v_k)) = (\tier(u_k), \indexOf(u_k))$$
Specifically, for any vertex $v$ reached at the $k$-th step of a computation walk, we define $\timeOf_G(v) = k$, which is well-defined due to the grid-aligned property. 
For any subgraph $G' \subseteq G$, the coordinates and time step of any vertex $u \in G'$ are defined relative to the original context of $G$.
To simplify notation, we omit the subscript $G$ when the context is clear; thus, $\timeOf(u)$ refers to $\timeOf_G(u)$ for any $u \in G' \subseteq G$.
\end{definition}

\begin{lemma}[Certificate-Oblivious to Grid-Aligned Correspondence]\label{lem:cotm_to_grid_aligned}
A Certificate-Oblivious Verifier TM $M_{\mathrm{COTM}}$ generates a footmark graph $G$ of $\mathcal{W}$ that satisfies the Grid-Aligned property where $\mathcal{W}$ is a set of all computation walks corresponding to all certificates for a particular problem instance.
\end{lemma}

\begin{proof}
By the definition of a Certificate-Oblivious Verifier, the tape head position $H_M(\Lfixed, Y, t)$ is uniquely determined by the time step $t$ and is independent of the certificate content $Y$.
Let $W = (v_0, v_1, \dots, v_n)$ and $W' = (v'_0, v'_1, \dots, v'_m)$ be two computation walks in $G$. Let $W_j$ be the subwalk of $W$ from $v_0$ to $v_j$. Let $c_{i, j}$ denote the last node visited in the $i$-th cell by the walk $W_j$. 
We prove the claim that $(\tier(v_j), \indexOf(v_j))$ is identical for all walks $W$ at each step $j$, and $\tier(c_{i, j}) = \tier(c'_{i, j})$ for all $i, j$.

Suppose, for the sake of contradiction, that $G$ is not Grid-Aligned. Then there exists a minimum index $k$ such that $(\tier(v_k), \indexOf(v_k)) \neq (\tier(v'_k), \indexOf(v'_k))$ or $\tier(c_{i, k}) \neq \tier(c'_{i, k})$ for some $i$.
By the minimality of $k$, it holds that for all $j < k$, the coordinate trajectories are identical: $(\tier(v_j), \indexOf(v_j)) = (\tier(v'_j), \indexOf(v'_j))$ and $\tier(c_{i, j}) = \tier(c'_{i, j})$ for all $i$.

By the definition of a CSOTM, the tape head position—and consequently the $\indexOf$—is identical for all $t$ up to the halting time. Thus, $\indexOf(v_k) = \indexOf(v'_k)$. Let $l = \indexOf(v_k)$.
Furthermore, the tier at step $k$ is given by $\tier(v_k) = \tier(c_{l, k-1}) + 1$. Since $\tier(c_{l, k-1}) = \tier(c'_{l, k-1})$ due to the minimality of $k$, it follows that $\tier(v_k) = \tier(c'_{l, k-1}) + 1 = \tier(v'_k)$.
For all $i \neq l$, the state of the $i$-th cell remains unchanged: $\tier(c_{i, k}) = \tier(c_{i, k-1}) = \tier(c'_{i, k-1}) = \tier(c'_{i, k})$.
For $i = l$, since $v_k$ is the last node of $W_k$ in the $l$-th cell and $\tier(c_{l, k}) = \tier(c'_{l, k})$, it follows that $\tier(v_k) = \tier(v'_k)$.

This contradicts the assumption that the coordinates or tier values differ at step $k$. Thus, for any two computational walks $W$ and $W'$, the coordinate trajectory $(\indexOf, \tier)$ is identical, implying that $G$ is Grid-Aligned.
\end{proof}

\begin{lemma}[Total Collapse of the Feasible Graph]\label{lem:total_collapse}
Let $G$ be a subgraph of grid-aligned footmark graph with a unique initial node hosting at least one valid feasible walk with respect to the designated final edge $e_f$. 
Let $G_f = \mathsf{Feasible}(G -e)$ denote the residual feasible graph obtained after removing a single edge $e \in E(G)$. 
If $e$ is an essential edge (i.e., contained in every valid feasible walk) situated prior to the second splitting edge of any feasible walk, then the residual structure undergoes an unconditional total collapse: $G_f = \emptyset$.
\end{lemma}

\begin{proof}
We prove this structural collapse by establishing a stronger, generalized claim: \emph{For any subgraph $G$ of a footmark graph,
if a set of removed feasible edges $E_r$ contains an essential edge $e$ (which belongs to every valid computation walk in $G$ terminating in $E_f$) located prior to the second splitting edge, then the resulting feasible graph $\mathsf{Feasible}(G - E_r)$ is empty.
This holds under the structural constraint that all target nodes $v$ of edges in $E_f$ share an identical time value $\timeOf(v)$ according to the grid-aligned property.}
Note that due to the grid-aligned property, an identical time value $\timeOf(v)$ implies an identical spatial coordinate $(index(v), tier(v))$. Consequently, all target nodes of $E_f$ reside on the same grid point.

We proceed by induction on the total number of edges $n = |E(G-E_r)|$.

\paragraph{Base Case ($n=0$):} 
if  $|E(G-E_r)|=0$, trivially feasible graph of $G-E_r$ is empty, hence $G_f = \emptyset$.

\paragraph{Inductive Hypothesis (IH):} 
Assume that for any context graph $G-E_r$ containing exactly $i$ edges for such essential edge $e \in E_r$, the removal of such an essential edge guarantees total collapse, yielding $G_f = \emptyset$.

\paragraph{Inductive Step ($n=i+1$):} 
The proof proceeds in two stages: first, we establish the existence of at least one step-pendant edge $e'$ within $G' \setminus E_r$ via contradiction. Second, we apply mathematical induction by pruning $e'$, showing that the remaining graph collapses while preserving the consistency of the feasible graph structure with respect to $E_f$.

Consider a graph $G' - E_r$ containing $i+1$ edges with an essential edge $e \in E_r$. 
To trigger the inductive step, we must show that $G' - E_r$ necessarily contains \emph{at least one} step-pendant edge. 
Suppose, for contradiction,  that $G' - E_r$ contains no step-pendant edges. Let $F$ be the set of all edges whose target nodes possess the maximum time value in the grid-aligned graph. 
By the maximality of time, any edge $f \in F$ has neither next edges nor index-succedent edges (see \cref{lem:ex-pendant_edges_with_no_index-succedent}). 
Consequently, $f$ must belong to $E_f$; otherwise, by \cref{def:step_pendant_edge}, $f$ would be flagged as a step-pendant edge, triggering an immediate contradiction. 
Furthermore, if $e$ were a non-splitting edge, its immediate predecessors $\Prev(e)$ would become step-pendant upon $e$'s erasure. Thus, $e$ must be a splitting edge.

We categorize the removal of an edge $e$ into three exhaustive cases based on its essentiality in $G = G' - F$ and its relation to the set $F$:
(1) $e$ is essential, (2) $e \in F$ (non-essential), and (3) $e \notin F$ and is non-essential. We demonstrate that each case leads to a structural contradiction regarding the existence of step-pendant edges.

\begin{itemize}[leftmargin=*]
	\item \textbf{Case 1: $e \notin F$, and $e$ is essential in $G$.} \\
	Let $f$ be an edge in $F$. Note that throughout the following sub-cases, any shared incident node or step-adjacent structure between $e$ and $f$ forces identical $(index, tier)$ coordinates due to the grid-alignment property, which we use to derive structural contradictions.
	
	Suppose $f$ shares step-adjacent edges with any edge in $E_r$; then, by grid-alignment, these edges must share identical $(index, tier)$ coordinates with an incident node of $f$. This is impossible, as all edges with the maximum time value are constrained to belong to $F$.
	
	\begin{itemize}
	    \item \textbf{Case where $e$ and $f$ share index-precedent edges:} 
	    Since $e$ and $f$ both become index-succedents of a common index-precedent edge $e_p$, the $(index, tier)$ coordinates must be shared between the source nodes of $e$ and $f$. 
	    
	    \item \textbf{Case where $e$ and $f$ share index-succedent edges:} 
	    Since $e$ and $f$ both become index-precedents of a common index-succedent edge $e_p$, the $(index, tier)$ coordinates must be shared between the target nodes of $e$ and $f$. 
	    
	    \item \textbf{Case where $f$ is a merging edge to $e$ (sharing a next edge):} The target (head) nodes of $e$ and $f$ are identical.
	    
	    \item \textbf{Case where $f$ is a splitting edge from $e$ (sharing a previous edge):} The source nodes of $e$ and $f$ are identical. 
	\end{itemize}
	
	Since sharing any step-adjacent edge with $e$ leads to a structural contradiction, $f$ cannot share any such edges with $e$. Let $E_f^*$ denote the set of all previous edges of $F$. 
	Now, we examine whether a step-pendant edge $e'$ of $G$ remains step-pendant in $G' = G + F$ by investigating the step-adjacent edges of $f \in F$:
	
	\begin{itemize}[leftmargin=*]
	    \item \textbf{Case where $e' \in \ISucc(f)$ or $e' \in \Next(f)$:} This case is impossible, as $f$ is chosen to be an ex-pendant edge with no index-succedent edges.
	    
	    \item \textbf{Case where $e' \in \Prev(f)$:} Since $e' \in E_f^*$, the edge $e'$ cannot be step-pendant due to the presence of next-edges. Consequently, the addition of $f$ does not alter the step-pendancy status of $e'$, because the role of $f$ as a next-edge for $e'$ is already accounted for in the feasibility constraints.
	    
	    \item \textbf{Case where $e' \in \IPrec(f)$:} By \cref{index-succedent_from_cover_edge}, if $f$ is an index-succedent edge of $e'$, then $e'$ is a cover edge in $G$ toward $E_f^*$. This implies that $e'$ is not a step-pendant edge in $G$ due to the lack of index-succedent edges; thus, the addition of $f$ does not alter this status.
	\end{itemize}
	
	Consequently, the step-pendancy status of the step-adjacent edges of $e'$ in $G' - E_r$ remains identical to their configuration in $G - E_r$. Since the feasible graph of $G - E_r$ is empty by the IH, $G' - E_r$ must contain a step-pendant edge, which contradicts the initial assumption that $G' - E_r$ is step-pendant free.

	\item \textbf{Case 2: $e \in F$, implying $e$ is not essential in $G$.} \\
	By definition, any $e \in F$ is a designated final edge in $E_f$ since $e$ is the essential edge. As established, $e$ must be a splitting edge; otherwise, its removal would trigger an immediate step-pendancy in its previous edges, contradicting our assumption. 
	
	Since the first essential splitting edge in our grid-aligned model must be a floor edge that terminates in $E_f$ (by \cref{lem:no_infeasible_first_splitting_edge}), any essential splitting edge $e$ must coincide with an edge in $E_f$. 
	
	Furthermore, if $e$ and $f$ (where $f \in F$) share a previous edge as a splitting junction, we examine the following structural sub-cases:
	\begin{itemize}[leftmargin=*]
	    \item \textbf{Both are Floor Edges:} The existence of $f$ provides an alternative feasible walk, which directly contradicts the essentiality of $e$.
	    \item \textbf{Either is a non-floor edge:} This configuration violates \cref{lem:no_infeasible_first_splitting_edge}, which explicitly excludes the existence of infeasible edges at the first splitting junction, rendering this configuration structurally impossible.
	\end{itemize}
	Consequently, the removal of $e$ leads to an immediate structural contradiction, confirming that this case is impossible.
   
	\item \textbf{Case 3: $e \notin F$ and $e$ is non-essential in $G$.} \\
		In this scenario, $e$ is non-essential in $G$ due to the existence of alternative feasible paths. However, the introduction of the terminal set $F$ prunes these alternatives, promoting $e$ to an essential edge in $G' - E_r$. This leads to the infeasibility of $W_f + f$ for some feasible walk $W_f$ previously feasible in $G$, as the history constraints imposed by an edge $f \in F$ render the computation walk invalid.
	    
	    The infeasibility arises at some ceiling edge $e_c$ of $W_f$, where $\init(e_c) \notin \IPrec(\term(f))$. If $e_c$ were not a cover edge, it would necessarily become a step-pendant edge, contradicting the assumption that $G' - E_r$ is step-pendant free. 
	    
	    Crucially, $e_c$ cannot be a cover edge because $G'$ is a grid-aligned footmark graph: the target nodes of all edges in $F$ share an identical time step and $(index, tier)$ coordinates. This synchronization renders the existence of a valid cover edge for such a configuration structurally impossible. Thus, $e_c$ is forced to be step-pendant, leading to an immediate contradiction. 
	   
\end{itemize}
Through these exhaustive cases, we demonstrate that the assumption of a step-pendant-free graph $G' - E_r$ is untenable. 
This structural necessity proves that the pruning of an essential edge consistently results in the emergence of a step-pendant edge, thereby satisfying the inductive requirement.
Consequently, $G' - E_r$ is mathematically guaranteed to host at least one step-pendant edge, say $e'$. By designating $e'$ as the first edge extracted by the fixed-point pruning sequence, we reduce the graph $G'$ to a smaller configuration that remains within the scope of our inductive hypothesis.

\begin{itemize}[leftmargin=*]
    \item \textbf{Case A: $e'$ is not contained in any feasible walk.} \\
    In this case, removing $e'$ does not eliminate any valid feasible traces.
    Thus, $e \in E_r$ remains a strictly essential edge in the reduced graph $G'' = G' -e'$. Since $|E(G''-E_r)| = i$, the Inductive Hypothesis (IH) applies directly to $G''$. Consequently, the feasible graph of $G'' -E_r$ is empty, implying that the feasible graph of $G' - E_r$ is also empty.
    \item \textbf{Case B: $e'$ is contained within some feasible walks.} \\
	By including $e'$ in the set of pruned edges, we define $E_r' = E_r \cup \{e'\}$. Since $e'$ is a step-pendant edge, its removal reduces the total number of edges in the graph, i.e., $|E(G' - E_r')| = i$. This strictly satisfies the condition for the Inductive Hypothesis. As the removal of a step-pendant edge does not invalidate the essentiality of $e$, the reduced graph $G' - E_r'$ collapses to an empty feasible graph by the IH, which in turn confirms the structural collapse of $G' - E_r$.
\end{itemize}

By mathematical induction, the removal of an essential edge $e$ located before the second splitting boundary triggers the total structural collapse of the graph: $G_f = \emptyset$ for $G-E_r$ where $E(G-E_r)=n \ge 1$.
\end{proof}

\begin{remark}[Necessity of the Grid-Aligned Footmarks Property]
While the preceding lemmas establish local dependency properties, \cref{lem:total_collapse} necessitates the global constraint of a \textbf{Grid-Aligned Footmarks} structure. This condition is fundamental, acting as the linchpin that triggers the total structural collapse of the feasible graph.
It is important to note that \cref{lem:total_collapse} strictly applies to grid-aligned footmark graphs. 
The essentiality of an edge $e$ and the subsequent collapse of the feasible graph $G_f$ rely on the fact that every edge in $G$ is part of a valid computation walk. 
If the graph were a general augmented graph where $e_t$ does not belong to a deterministic computation walk, or if it lacks the \textbf{Grid-Aligned} property, the removal of an "essential" edge \textbf{might not} lead to an empty feasible graph, as the structural consistency guaranteed by the Turing machine's transition function would be absent.
\end{remark}

\section{Verification of Computation Walks}\label{sec:walk_verification}

In this section, we present a deterministic procedure to verify the existence of a valid computation walk containing a \textit{verification target edge}. While the static \ComputeFeasibleGraph{} algorithm identifies a subgraph $G_f$ of a computation graph containing all feasible walks to designated final edges, the feasible graph $G_f$ often retains persistent obsolete and orphaned structures. To resolve these, we employ a dynamic pruning strategy focusing on \textbf{implausible edges}.

Our approach proceeds through three targeted stages:

\begin{itemize}
   \item \textbf{Exploratory Selection of Candidate Walks.} We first attempt to select an arbitrary valid computation walk from the initial nodes on the feasible graph. If it successfully reaches the verification target edge, the participating edges are confirmed to be included in valid computation walks to the verification target edge.

   \item \textbf{Pruning of Implausible Merging Edges.} A first splitting edge is categorized as an \textbf{implausible edge} if its necessity for maintaining a computation walk containing a verification target edge is not yet established. Instead of assuming their invalidity \textit{a priori}, we treat them as structural suspects. We test these edges by observing the graph's stability upon their trial removal; if the removal triggers the collapse of the target computation walk, the edge is deemed essential for maintaining a feasible walk to the verification target edge."

    \item \textbf{Detection and Isolation of Futile Structures.} We further refine the graph by isolating clusters that mimic walk-like connectivity but lack global validity. This is achieved by identifying critical walks whose removal would induce a structural collapse. This stage systematically strips away the ``structural noise'' caused by obsolete walks that survived the initial feasibility computation.

    \item \textbf{Verification through Cascading Collapse.} The final verification is achieved through a controlled, recursive elimination process. By invoking \ComputeFeasibleGraph{} after the removal of a futile edge, we induce a \textbf{cascading collapse} of all structures tethered to it. In this stage, we systematically target edges that are redundant to the existence of a targeted computation walk. If the verification target edge remains stable and reachable via a consistent path after such removal, we confirm that the remaining structure is progressively refined toward a unique, minimal computation walk.
\end{itemize}

\paragraph{The Pruning Mechanism}
The core of the pruning strategy lies in identifying \textbf{computing-futile edges}---those that belong exclusively to \textbf{computing-futile walks}. An edge is deemed removable only if its removal does not violate the feasibility of any \textit{targeted} walks designated for preservation.

To refine the graph without losing essential computation walks, we distinguish between walks based on their strategic necessity during the pruning process. 

\begin{definition}[Computing-Redundant and Computing-Futile Edges]\label{def:computing-redundant-futile-edge}
Let $G$ be a subgraph of the $e_t$-augmented footmarks $G_U$ for a verification target edge $e_t$.
\begin{itemize}
    \item \textbf{Computing-targeted walk:} A valid computation walk in $G$ that reaches a \textit{verification target edge} in $e_t$.
    \item \textbf{Computing-futile walk:} A maximal computation walk in $G$ that does not reach any edge in $e_t$.
    \item \textbf{Computing-embedded walk:} A maximal computation walk in $G$ composed entirely of edges from other \textit{computing-targeted} walks, yet functioning as a \textit{computing-futile walk} itself. 
    \item \textbf{Nested computing-futile walk:} A walk $W$ in $G$ is a \textit{nested computing-futile walk} if it is a subwalk of another computing-futile or computing-targeted walk, and there exists some edge $f \in E(G_U) \setminus E(G)$ such that the $f$-augmented walk $W + f$ forms a valid computation walk in $G_U$. 
    \item \textbf{Computing-effective edge:} An edge $e \in E(G)$ that belongs to some computing-targeted walk in $G$ to $e_t$.
    \item \textbf{Computing-essential edge:} An edge $e \in E(G)$ that is strictly contained within every valid computing-targeted walk in $G$ to $e_t$.
    \item \textbf{Computing-redundant edge:} A computing-effective edge $e \in E(G)$ for some computing-targeted walk in $G$ to $e_t$, such that there exists another computing-targeted walk to $e_t$ in the feasible graph of $G - e$.
    \item \textbf{Computing-futile edge:} An edge $e \in E(G)$ that does not belong to any computing-targeted walk in $G$ to $e_t$. Equivalently, removing $e$ from $G$ does not destroy the existence of any computing-targeted walk to $e_t$.
\end{itemize}
\end{definition}
\begin{figure}[htbp]
        \centering
        \includegraphics[width=0.5\textwidth]{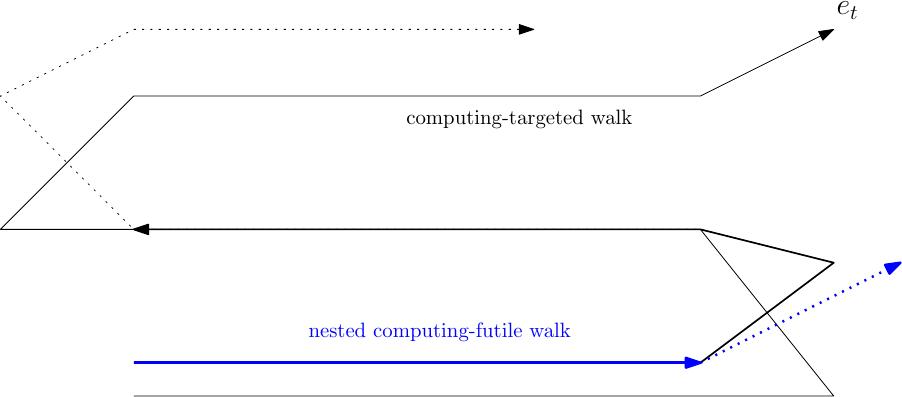}
        \caption{Nested Computing-futile Walk}
        \label{fig:nested_futile_walk}
        \Description{A walk that is subsumed within an obsolete walk, branching from a feasible trajectory but terminating prematurely within the graph's observable boundary. A dotted line extends from its terminal vertex, indicating a potential continuation that remains outside the feasible structure and leads to a structural dead end.}
\end{figure}

The term ``computing-futile'' does not indicate relative disjointness from a single computing-targeted walk, but rather that the edge is excluded from all computing-targeted walks.

\begin{remark}[Necessity of Critical Computation Walk Exploration]
The fundamental justification for exploring the \textit{critical attempted walk} lies in the discrepancy between local and global consistency. In our framework, an edge may satisfy all \textit{local consistency} requirements (i.e., transition rules and local reachability), yet fail to maintain \textit{global consistency} (i.e., forming a complete valid computation that reaches $e_t$). 

An \textbf{embedded walk} serves as a prime example: while it is composed of edges from computation walks containing verification target edges that are locally valid, the walk itself is functionally futile because it provides no novel contribution to the global verification. Therefore, the algorithm must traverse the critical walk to its terminus to definitively identify whether a walk is truly effective or merely an embedded futile sequence. This global verification, conducted within polynomial bounds, ensures that the pruning process is both sound and exhaustive.
\end{remark}

\SetKwFunction{PruneWalk}{PruneWalk}
\SetKwFunction{FindFirstSplittingEdgeOrFinalEdge}{FindFirstSplittingEdgeOrFinalEdge}
\SetKwFunction{ExtendFutileWalks}{ExtendFutileWalks}

\subsection{Pruning Walk Strategy for Implausible Edges}

In this subsection, we define the instrumental pruning mechanisms designed to isolate and test \textbf{implausible edges}. Here, pruning is treated not merely as a reduction step, but as a diagnostic tool to probe the structural necessity of specific edges. We focus on the most vulnerable points of a computation walk: the initial \textit{splitting edge} or, in its absence, the \textit{final edge}. By targeting these structural junctions, we can observe how the global reachability of the graph responds to their removal.

To provide a robust foundation for subsequent identification, we introduce two distinct pruning modalities:

\begin{enumerate}
    \item \textbf{Target-Oriented Pruning:} This mode selectively preserves only \textbf{computing-targeted walks}, systematically eliminating walks identified as futile using feasible graph. This is used to extract the minimal core required for verification.
    \item \textbf{Conservative Futile-Preserving Pruning:} This more nuanced mode prunes specific futile structures while intentionally retaining other \textbf{computing-futile walks}. By maintaining a subset of futile structures, this tool allows the algorithm to distinguish between \textit{embedded} noise and \textit{truly futile} edges in the later stages of verification.
\end{enumerate}

These two pruning tools serve as the operational basis for the subsections to follow. By manipulating the presence of implausible edges under these different constraints, we can definitively categorize an edge as truly futile if its removal consistently fails to trigger a cascading collapse of the designated target structures.

\begin{algorithm}[!ht]
\caption{Pruning an Edge Given a Computing-Futile Walk} \label{alg:pruning_an_edge_of_walk}
\Input{$G$: Feasible Graph, $e_t$: Verification target edge, $G_U$: $e_t$-augmented footmarks, \\ $W$: Computation Walk, $preserveFutile$: flag indicating whether to preserve computing-futile walks}
\Output{The feasible graph $G'$ in which an edge of $W$ removed.}

\Function{\PruneWalk{$G_U, G, v_0, e_t, W, preserveFutile$}} {
\State{Let $E_o \gets \emptyset$ and $E_f \gets \{e_t\}$}
\State{Set $e' \gets$ \FindFirstSplittingEdgeOrFinalEdge{$G,W$}}\label{alg_line:find_first_splitting_edge}
\If{preserveFutile} {
	\tcp{Add the extended futile edges to the designated final edges}
	\State{\ExtendFutileWalks{$G_U, G, E_o, E_f$}}
}
\StateC{Let $G' \gets$ \ComputeFeasibleGraph{$G-e', \{v_0\}, E_f \cup E_o$}}\Comment{Remove $e'$ from feasible graph if exists}\label{alg_line:find_obsolete_edge:recompute_feasible_graph}
\StateC{\Return $G'-E_o$}\Comment{Do not recover removed edge, this ensure polynomial time complexity}
}

\Function{\FindFirstSplittingEdgeOrFinalEdge{$G, W$}} {
\State{Let $e$ be the first edge of computation walk $W$}
\While{$e$ is not the final edge of walk $W$} {\label{alg_line:find_first_splitting_edge:while_start}
    \If{$e$ is splitting edge} {
        \State{\Return $e$}
    }
    \State{$e \gets \nextOf_W(e)$ }
}
\StateC{\Return $e$} \Comment{It returns final edge of computation walk if no splitting edge found}
}

\Procedure{\ExtendFutileWalks{$G_U:\In, G :\InOut, E_o :\Out, E_f: \In$}} {
   \ForAll{edge $e=(u,v) \in E(G_U)$} {
        \If{ $e \not\in E(G)$ \textbf{and} $u \ne t$ for any edge $(s,t) \in E_f$} {
	        \If{$u$ is incident to any node in $V(G)$ \textbf{and} and $tier(v)>0$ and $\IPrec_G(v) \ne \emptyset$ } {
	             \State{Let  $G \gets G+e$}
	             \State{Let $E_o \gets E_o \cup \{e\}$}
	        }
	}
    }
}
\end{algorithm}

\begin{definition}[Extendable and Extended Computing-Futile Edge]\label{def:extended_futile_edge}
Let $\mathcal{W}$ be a set of computation walks, and let $G_U$ denote $e_t$-augmented footmark graph for $\mathcal{W}$ with target verification edge $e_t$.  
Let $G \subseteq G_U$ be the current feasible graph, and let $W \subset W' \in \mathcal{W}$ be a computing-futile walk with respect to $e_t$.  

\begin{itemize}
    \item An edge $e$ of $E(W') \setminus E(G)$ is called an \emph{extendable computing-futile edge} if $W$ is not a nested computing-futile walk, and the sequence $W+e$ (the concatenation of $W$ and $e$) is the valid computation walk.  
   
    \item An \emph{extended computing-futile edge} is an extendable futile edge (from $E_o$) that is included in the augmented graph $G + E_o$ where $E_o$ is the set of all extendable edges in $G$.
\end{itemize}
\end{definition}

The formal establishment of the correctness of extending extendable computing-futile edges via the procedure \ExtendFutileWalks{}, including \cref{lem:add_final_edges_correctness} and its proof, is deferred to the Appendix.

\begin{lemma}[Correctness of Pruning a Walk]\label{lem:correctness_of_prunewalk}
Let $G$ be a subgraph of $e_t$-augmented footmarks with a unique initial node $v_0$, and let $W$ be a computing-futile walk to $e_t$ in $G$. If $\PruneWalk{}$ is executed, the resulting graph $G'$ satisfies the following:
\begin{itemize}
    \item If the \textsf{preserveFutile} flag is true and at least two such walks exist, then at least one computing-targeted or computing-futile walk from $G$ is preserved in $G'$.
    \item $G'$ is a subgraph of $G$ with at least one fewer edge than $G$, where the first splitting edge (or the final edge if no such edge exists) on a computing-targeted or computing-futile walk is removed.
    \item A nested computing-futile walk is not preserved in $G'$ unless the computing-futile walk containing it is also preserved.
    \item If $G$ is a grid-aligned footmark and the computation walk $W$ contains a computing-essential edge as its first splitting edge, then $G'$ is an empty graph.
\end{itemize}
\begin{proof}
By \cref{lem:correctness_of_find_first_splitting_edge}, a removable edge $e'$ is selected on $W$ as the first splitting edge that does not belong to all computation walks in $G$ if more than one computation walk exists in $G$. 
If the \textsf{preserveFutile} flag is true, by \cref{lem:add_final_edges_correctness}, all next edges of computing-futile walks are extended and stored in $E_o$ except for nested computing-futile walks. 
By \cref{lem:preservation_of_feasible_walks}, this ensures that computing-futile walks not containing $e'$ remain valid after \ComputeFeasibleGraph{} is invoked; otherwise, all computing-targeted walks not containing $e'$ are preserved.

Then, by \cref{lem:constructing_feasible_graph}, the final output is a feasible subgraph of $G - e'$, defined with respect to $E_o \cup E_f$, and all step-pendant edges are removed, where $E_o$ is the set of extended computing-futile edges if \textsf{preserveFutile} is true, and empty otherwise.

Therefore, $\PruneWalk{}$ correctly computes a feasible graph $G'$ in which the first splitting edge $e'$ is removed from a computing-futile walk $W$, while preserving at least one computation walk (either computing-targeted or computing-futile). 
Since no additional edges are added, $G'$ has at least one fewer edge than $G$. If \textsf{preserveFutile} is true, there is at least one preserved computing-futile walk that is not a nested computing-futile walk of another.

Regarding the fourth property, suppose the first splitting edge of $W$ is a computing-essential edge. By definition, a computing-essential edge is contained in every computing-targeted walk, and thus in every feasible walk. 
By \cref{lem:total_collapse}, the existence of such an edge imposes a structural constraint that forces a total collapse, yielding an empty feasible graph $G'$
\end{proof}
\end{lemma}

\begin{remark}
The algorithm ensures that at least one edge is removed during each iteration of the outer loop by removing extended computing-futile edges before it returns. By maintaining the invariant that a removed edge is never recovered in any subsequent phase, we guarantee a monotonic reduction of the feasible graph $G$. This structural decay is essential to ensure that the number of feasible graph updates is polynomially bounded, preventing already removed edges in the feasible graph from re-entering the computation and ensuring convergence toward a minimal feasible walk.
\end{remark} 

\begin{lemma}[Time Complexity of \PruneWalk{}]
\label{lem:prune_walk_time_complexity}

Let $G$ be a feasible graph of width $w$ and height $h$, with
$|E(G)| = \bigO(wh^2)$.  
Let $T_f$ denote the worst-case time complexity of \ComputeFeasibleGraph{}.

Then, the worst-case time complexity of \PruneWalk{} is bounded by $T_f$.

\begin{proof}
The procedure \PruneWalk{} consists of the following steps:

\begin{itemize}
    \item \textbf{Calling \FindFirstSplittingEdgeOrFinalEdge{}:}  
    This function scans the given computation walk $W$ and returns either
    the first splitting edge or the final edge.
    Since the length of $W$ is at most $\bigO(wh)$, this step takes
    $\bigO(wh)$ time.

    \item \textbf{Calling \ExtendFutileWalks{}:}  
    By \cref{lem:add_final_edges_time_complexity}, this step runs in
    $\bigO(wh^3)$ time.

    \item \textbf{Calling \ComputeFeasibleGraph{}:}  
    This recomputes the feasible subgraph of $G - e'$ using the updated
    sets $E_o$ and $\{v_0\}$.
    By \cref{lem:feasible_graph_time_complexity}, this step runs in
    \[
    T_f = O\bigl(w^{2}h^{4}(h\log h+\log w)\bigr).
    \]
\end{itemize}

All other operations incur lower-order costs.
Therefore, the total worst-case running time of \PruneWalk{}
is dominated by the call to \ComputeFeasibleGraph{}, and is
bounded by $T_f$.
\end{proof}
\end{lemma}

\SetKwFunction{FindDisjointEdge}{FindDisjointEdge} 
\subsection{Detecting Computing-Redundant or Computing-Futile Edges} \label{subsec:redundant_futile_edge_detection}

To prune unnecessary or redundant edges from a feasible graph $G$, we first aim to identify computation walks that do not contribute uniquely to any target configuration. This is achieved by locating either a \emph{computing-futile edge} or a \emph{computing-redundant edge} within computing-futile or computing-embedded walks. 

The core detection mechanism utilizes the pruning tools established in the previous subsection. The process is centered on the observation of a \textbf{cascading collapse} triggered by the experimental removal of a computing-futile walk. The strategy operates as follows:

\begin{enumerate}
    \item \textbf{Select an Arbitrary Computation Walk:} We take an arbitrary computation walk. If it is a computing-targeted walk, further detection is not necessary, otherwise it proceeds to the following step with the selected computing-futile walk.
    \item \textbf{Experimental Elimination of Futile-Computing Walk:} We temporarily remove a computing-futile walk from the current feasible graph even if it contains computing-effective edges. 
    \item \textbf{Detection of Strategic Necessity via Collapse:} If the removal of the computation walk $W_c$ triggers a total collapse—specifically, the emergence of an empty graph—it confirms that $W_c$ was a vital constituent of a \textbf{computing-targeted walk}.
    \item \textbf{Isolation through Futile-Preservation:} To isolate the truly futile components, we re-invoke the \textbf{Conservative Futile-Preserving Pruning} tool. By intentionally preserving known \textit{computing-futile walks} while $W_c$ is absent, we create a structural contrast.
    \item \textbf{Identification of Redundancy:} In this contrastive state, an edge $e \in G$ that remains within the structure, yet was not part of the computation walk $W_c$ itself, is identified as a \textbf{computing-futile edge} (or a \textbf{computing-redundant edge}). 
\end{enumerate}

This sophisticated triangulation—testing the necessity of some edges of the $W_c$ while shielding futile backgrounds— allows the algorithm to pinpoint ``structural noise'' that local consistency checks could never identify. 
By observing what remains functional even when the futile background is preserved, we ensure that the pruning process is both surgically precise and globally sound.

\begin{definition}[Pruned Graphs and Critical Attempted Walk]\label{def:pruned_graph_and_attempted_walk}
Let $G^{(0)}$ be a feasible graph with respect to a set of designated final edges $E_f=\{e_t\}$ where $e_t$ is a verification target edge.

For each $i \ge 0$, define $G^{(i+1)}$ as the graph obtained by applying
\PruneWalk{} to $G^{(i)}$ without preserving computing-futile walks, using a maximal computation walk $W_i$ in $G^{(i)}$, chosen arbitrarily among those that are not compuing-targeted walk to $e_t$.

Let $m < |E(G)|$ be the minimal integer such that, in $G^{(m)}$, either every remaining computation walk is feasible with respect to $E_f$, or $E(G^{(m)}) = \emptyset$.

Let $R := W_m$, called the \emph{critical attempted walk}, obtained on the maximal pruned graph at which the pruning process terminates.

We refer to:
\begin{itemize}
  \item $G^{(i)}$ as the $i$-th \textbf{pruned graph},
  \item $W_i$ (for $0 \le i < m$) as an \textbf{attempted walk},
  \item $G^{(m)}$ as the \textbf{maximal pruned graph},
  \item $R$ as the \textbf{critical attempted walk}.
\end{itemize}
\end{definition}

The following \cref{alg:find_target_redundant_futile_edge} attempts to isolate an edge that can be safely pruned without eliminating all computing-targeted walks.
It first takes a computation walk, if it finds a computing-targeted walk by investigating whether it contains the target edge $e_t$, then it returns immediately.
 Otherwise, it iteratively removes edges in the walks from the graph using \PruneWalk{}, detecting a critical attempted walk which removes all the computing-targeted walk to $e_t$.

When a critical attempted walk is detected, the algorithm recomputes the feasible graph by removing the first splitting edge of the detected walk while preserving computing-futile walks from the last graph that still contained a computing-targeted walk.
The critical attempted walk must contain the feasible walk up to the removed first splitting edge; as otherwise, the feasible graph would not become empty. 
Subsequently, the algorithm identifies a computing-futile walk using the next attempted walk on the re-pruned graph. 
Finally, it invokes \FindDisjointEdge{} to locate an edge that is not shared with the critical walk (represented by graph $R$), ensuring its redundancy or disjointness.
Since the first edge of a computing-futile walk that does not belong to the feasible walk is inherently either part of a computing-futile walk or at least a computing-redundant walk, its identification is guaranteed.

\begin{algorithm}[!ht]
\caption{Find Computing-Redundant or Computing-Disjoint Edge for Pruning} \label{alg:find_target_redundant_futile_edge}
\SetKwFunction{TakeArbitraryWalk}{TakeArbitraryWalk}
\SetKwFunction{FindTargetRedundantFutileEdge}{FindTargetRedundantFutileEdge}
\Input{$e_t$: The verification target edge, $G_U$: $e_t$-augmented footmarks, \\$G$: Original Feasible graph, $v_0$: unique Initial node} 
\Output{Computing-redundant or computing-futile Edge of feasible graph $G$ to $e_t$}
\Function{\FindTargetRedundantFutileEdge{$G_U, G, v_0, e_t$}} {
\StateC{Let $E_f \gets \{e_t\}$}\Comment{$E_f$: set of verification target edge}
\State{Let $R \gets$ an empty graph}
\While(\tcc*[f]{Loop until a computing-targeted/computing-futile edge found}){$G$ is not empty} {\label{alg_line:find_target_redundant_futile_edge:start_loop}
    \State{Let $W \gets$ \TakeArbitraryWalk{$G, v_0$}}\label{alg_line:take_arbitrary_walk}
    \If(\tcc*[f]{$W$ is computing-targeted walk}){$e_t$ in $W$} {
        \StateC{\Return {$e_t$}} \Comment{$e_t$ is not always at the end of $W$} 
    }
    \Else(\tcc*[f]{$W$ can be computing-embedded walk}) {
        \State{Set $H \gets$ \PruneWalk{$G_U, G, v_0, e_t, W, \False$}} 
	 \If(\tcc*[f]{$H$ contains no computing-targeted walk}){$H$ is empty} {
            \State{Set $G \gets$ \PruneWalk{$G_U, G, v_0, e_t, W, \True$}} \label{alg_line:prune_walk}
	        \StateC{\Return \FindDisjointEdge{$W, G$}}\label{alg_line:find_disjoint_edge}\Comment{Return futile or redundant edge}
        }
        \Else{
	     \State{Set $G \gets H$}
        }
    }
} \label{alg_line:find_target_redundant_futile_edge:end_loop}
\State{\Return $\NIL$}
}

\Function(\tcc*[f]{$R$ is a critical attempted walk in $G$}){\FindDisjointEdge{$R, G$}} { 
\State{Let $e$ be the first edge of walk $R$}
\While{$e$ is not NIL} {\label{alg_line:search_disjoint_edge:while_start}  
    \If{$e \not\in E(G)$} { \label{alg_line:search_disjoint_edge:select_disjoint_edge}
		\StateC{Let $E_N \gets \Outgoing(\init(e))$} \Comment{ $E_N$ does not contain $e$}
        \State{\Return an edge $f \in E_N$ if $|E_N|>0$ otherwise \NIL}
    }
    \StateC{$e \gets \nextOf_R(e)$}\Comment{If $\nextOf_R(e)$ does not exist then $e$ is \NIL}
}
\State{\Return \NIL}
}
\end{algorithm}

\begin{algorithm}[!ht]
\caption{Deterministic Walk Generation} \label{alg:take_arbitrary_walk}
\SetKwFunction{TakeArbitraryWalk}{TakeArbitraryWalk}
\Input{$G$: A feasible graph, $v_0$: Initial node}
\Output{$W$: A computation walk in $G$}
\Function{\TakeArbitraryWalk{$G, v_0$}} {
    \State{Let $S$ be an empty dynamic array representing the transition surface}
    \State{Let $W$ be an empty list of edges}
    \State{Let $e$ be the first available edge in $G$ incident with node in $v_0$,  otherwise \NIL}
    \While{$e \ne \NIL$} {
        \State{Update surface $S[\text{index}(u)]$ with the transition case of node $u$}
        \State{Append $e$ to walk $W$}
        \tcp{$S[i]=\IPrec_G(v')$ where $(u',v')$ is a next edge}
        \State{Set $e \gets$ the first edge of $\text{Next}_G(e)$ consistent with surface $S$, or $\NIL$ if none exists}
    } 
    \State{\Return $W$}
}
\end{algorithm}
\begin{remark}[Determinism and Reproducibility] 
Although \TakeArbitraryWalk{} involves selection steps, the overall pruning process remains \emph{deterministic} provided a consistent rule (e.g., "the first available edge") is applied. Crucially, any such consistent selection rule leads to the same correctness outcome, ensuring that the algorithm is both predictable and reproducible.
\end{remark}

\begin{sublemma}\label{sublem:correctness_find_futile_edge}
Let $R$ be the critical attempted walk that prunes all computing-targeted walks from the $i$-th pruned graph $G^{(i)}$ by applying \PruneWalk{} without preserving computing-futile walks, where $G^{(0)}$ is the initial feasible graph. 
Let $W$ be a computing-futile walk on the feasible graph $G'$ obtained by applying \PruneWalk{} to $G^{(i)} - e_r$ while preserving computing-futile walks, where $e_r$ is the first splitting edge of $R$ (if one exists), or the last edge of $R$ otherwise. 
Let $e_s$ be the first edge of walk $W$ that is not contained in walk $R$. Then, $e_s$ is a computing-futile edge unless it is a computing-redundant edge.
\end{sublemma}

\begin{proof}
\begin{figure}[ht]
  \centering
	\includegraphics[width=0.8\textwidth]{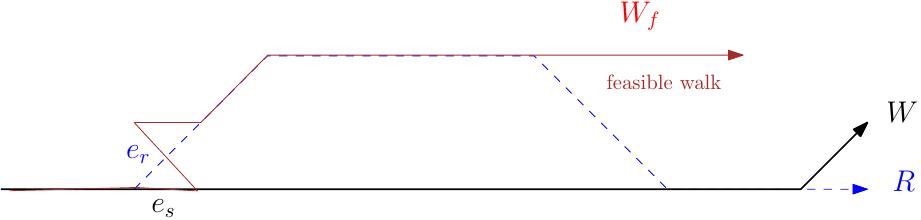}
	\caption{A computing-effective edge on computing-futile walk after all computing-targeted walks removed}
	\label{fig:effective_edge_on_futile_walk}
	\Description{A diagram showing a contradiction where a computing-targeted walk connects edge $e_s$ on walk $W$ to edge $e_r$ on walk $R$. }
\end{figure}

If $e_s$ is a \textit{computing-redundant edge}, the lemma holds trivially. Assume, for the sake of contradiction, that $e_s \in W$ (where $e_s \notin R$) is a \textit{computing-effective edge}.

Note that $e_s$ is an essential edge in the critical attempted walk $R$ by \cref{lem:total_collapse}, implying that $e_s$ is also a computing-effective edge in the initial feasible graph $G^{(0)}$. 
Let $W_f$ be a computing-targeted walk containing $e_s$. Since $e_s$ is not a computing-redundant edge, $W_f$ must contain both $e_r$ and $e_s$; otherwise, the feasible graph obtained from $G^{(i)} - e_r$ would fail to prune all computing-targeted walks, which contradicts the definition of $R$ as the critical attempted walk.

As illustrated in \cref{fig:effective_edge_on_futile_walk}, we demonstrate that any computation walk containing $e_s$ cannot contain $e_r$. Let $u$ be the common initial node of $e_s$ and $e_r$.

\begin{itemize}
    \item[1)] \textbf{Case where $e_s$ follows $e_r$ in $W_f$:} Any node occurring after $e_r$ must have a strictly higher tier than $u$ by the tier consistency of the computation walk defined in \cref{def:computation_walk}. However, the initial node of $e_s$ is $u$, which shares the same tier, leading to a contradiction.
    
    \item[2)] \textbf{Case where $e_s$ precedes $e_r$ in $W_f$:} Similarly, any node occurring after $e_s$ must have a strictly higher tier than $u$. Since the initial node of $e_r$ is $u$, this again leads to a contradiction.
\end{itemize}

In both cases, we reach a contradiction. Therefore, $e_s$ must be either a computing-futile edge or a computing-redundant edge, as defined in \cref{def:computing-redundant-futile-edge}.
\end{proof}

\begin{lemma}[Correctness of \FindTargetRedundantFutileEdge{}]
\label{lem:find_feasible_or_disjoint_walk}

Let $G_U$ be an $e_t$-augmented footmark graph. Given a feasible graph $G \subseteq G_U$ with a unique initial node $v_0$ and a verification target edge $e_t$,
\FindTargetRedundantFutileEdge{} in \cref{alg:find_target_redundant_futile_edge} satisfies the following:

\begin{enumerate}
    \item \textbf{If a computing-targeted walk exists in $G$,} then the algorithm returns :
    \begin{itemize}
 	\item the target verification edge $e_t$ contained in a computing-targeted walk in $G$, or
        \item a computing-redundant edge or a computing-futile edge to $e_t$ in $G$, provided that $G_U$ is grid-aligned footmark graph.
    \end{itemize}

    \item \textbf{If no computing-targeted walk exists in $G$,} then the algorithm returns :
    \begin{itemize}
        \item a computing-futile edge, or
        \item $\NIL$.
    \end{itemize}
\end{enumerate}

Moreover, the algorithm always terminates in finite time.

\begin{proof}
Let $G^{(i)}$ denote the graph obtained after the $i$-th iteration of the main
\texttt{while} loop
(from line \ref{alg_line:find_target_redundant_futile_edge:start_loop} to
line \ref{alg_line:find_target_redundant_futile_edge:end_loop} in \cref{alg:find_target_redundant_futile_edge}).
Let $n_w(G)$ denote the number of computing-targeted walks in $G$, and let
$n'_w(G)$ denote the number of computing-targeted or computing-futile walks in $G$.

Assume, for the sake of contradiction, that the algorithm violates the statement
of the lemma; namely, either a computing-targeted walk exists in $G$ but the
algorithm returns neither the verification target edge nor a computing-redundant/futile edge even though $G_U$ is a grid-aligned footmark graph,
or no computing-targeted walk exists in $G$ but the algorithm returns neither
a computing-futile edge nor $\NIL$.

Let $G^{(k-1)} \neq \emptyset$ be a pruned graph. Suppose that after applying \PruneWalk{} to $G^{(k-1)}$ without preserving computing-futile walks, the resulting graph becomes $G^{(k)} = \emptyset$. 

Specifically, if the removal of computing-essential edges during iteration $k$ triggers a total collapse leading to $G^{(k)} = \emptyset$, it follows from \cref{lem:correctness_of_prunewalk} that the preceding state $G^{(k-1)}$ must have contained at least one feasible walk, provided that the $e_t$-augmented footmark graph $G_U$ was constructed such that $e_t$ is contained in some computation walks. 

This implies that $G^{(k)}$ serves as a \textbf{maximal pruned graph}; any pruning beyond this point is pointless (see \cref{def:pruned_graph_and_attempted_walk}), as it would destroy the last remaining structures capable of sustaining a computing-targeted walk in $G^{(0)}$.

  By construction, this pruning step strictly decreases the total number of computation walks $n'_w(G^{(k-1)})$, effectively terminating the process only after all potential computing-targeted walks have been exhausted. This implies:
\[
n'_w(G^{(k-1)}) > n'_w(G^{(k)}),
\]
and
\[
n_w(G^{(k)}) = 0.
\]

Let $W_c$ be the walk selected by \TakeArbitraryWalk{} during iteration
$k$. Let $H$ be the feasible graph obtained from $G^{(k-1)}$ by applying
\PruneWalk{} while preserving computing-futile walks.
If $W_c$ contains a target verification edge $e_t$, the algorithm returns
such an edge, contradicting the assumption.
Therefore, $W_c$ is a computing-futile walk. By definition, $W_c$ is
either a computing-embedded walk consisting solely of computing-effective edges, or a computing-futile walk containing at least one computing-futile edge.

We consider the following exhaustive cases.
\begin{itemize}
\item \textbf{Case 1:} $n_w(G^{(k)}) > 0$.\\
This contradicts the choice of $k$, since applying \PruneWalk{} to
$W_c$ in iteration $k$ results in an empty graph, which implies $n_w(G^{(k)}) = 0$.

\item \textbf{Case 2:} $n_w(G^{(k)}) = 0$ and $n'_w(H) > 0$.\\
Since all computing-targeted walks have been eliminated, $W_c$ serves as the critical attempted
walk $R$.
In this state, the algorithm proceeds to apply \PruneWalk{} to $G$ while preserving computing-futile walks, 
and subsequently calls \FindDisjointEdge{} on the resulting graph. 
By \cref{sublem:correctness_find_futile_edge}, \FindDisjointEdge{} returns a computing-redundant or computing-futile edge, which contradicts the initial assumption.

\item \textbf{Case 3:} $n_w(G^{(k)}) = 0$ and $n'_w(H) = 0$.\\
In this case, the algorithm returns $\NIL$.
Assume for contradiction that a computing-targeted walk existed in the
original graph $G^{(0)}$.
First, the walk $W_c$ selected by \TakeArbitraryWalk{} cannot be a computing-targeted walk;
otherwise, the algorithm would return a target verification edge contained in it.
Therefore, there must exist a computation walk $W' \neq W_c$ that diverges from
$W_c$ at some edge. Otherwise, $W_c$ would be the only computation walk in $G^{(k-1)}$ and would
necessarily be computing-targeted, contradicting the assumption that the algorithm did
not return the target edge.

Since $G^{(k-1)}$ contains no computing-futile walk by the assumption (otherwise $H$ contains a computing-futile walk), 
$W'$ cannot be an computing-futile walk in $G^{(k-1)}$.

Hence, $W'$ must be a computing-targeted walk distinct from $W_c$, since any computation walk must belong to either computing-targeted or computing-futile walks.

This contradicts the maximality of the pruning process, which eliminates all
computing-targeted walks.
Therefore, no computing-targeted walk existed in the input graph $G=G^{(0)}$, contradicting the assumption.

\end{itemize}
In all cases, we obtain a contradiction. Hence, the algorithm always returns one
of the valid outputs. Since each pruning step strictly decreases the number of
computation walks and the graph is finite, the algorithm terminates in finite
time.
\end{proof}
\end{lemma}

\begin{lemma}[Time Complexity of Computing a Redundant or Futile Edge]
\label{lem:find_target_or_futile_edge_time}

Let $G$ be a feasible graph of width $w$ and height $h$, with
$|E(G)| = \bigO(wh^2)$.
Let $T_f$ denote the worst-case time complexity of
\ComputeFeasibleGraph{}.

Then, the worst-case time complexity of
\FindTargetRedundantFutileEdge{} is bounded by
\[
O(wh^2 \cdot T_f).
\]

\begin{proof}
We analyze the cost of a single iteration of the main loop in
\FindTargetRedundantFutileEdge{}.

\begin{itemize}
    \item \textbf{\TakeArbitraryWalk{}:}
    This procedure constructs a computation walk of length at most
    $\bigO(wh)$, since the walk spans at most $w$ index positions and at most
    $h$ vertices per index. Hence, this step takes $\bigO(wh)$ time.

	\item \textbf{\PruneWalk{}:} By \cref{lem:prune_walk_time_complexity}, this step runs in $O(T_f)$ time. Since it is invoked at most twice per iteration, the total time for this step remains $O(T_f)$.

    \item \textbf{\FindDisjointEdge{} (executed once upon termination):}
	This procedure scans a walk $W$ of length $O(wh)$. Upon encountering a missing edge $e \in W$ (i.e., $e \notin E(G)$), it performs an outgoing edge lookup and membership check in $R$, taking $O(\log h)$ time, and then terminates. 
    As the scan performs at most $O(wh)$ membership checks, the total complexity is $O(wh \log h)$
\end{itemize}

All steps other than \PruneWalk{} incur lower-order costs.
Thus, each iteration of the main loop runs in $\bigO(T_f)$ time.

Let $k$ be the number of iterations.
By \cref{lem:correctness_of_prunewalk}, at least one edge is removed from
the feasible graph in each iteration.
Since $|E(G)| = \bigO(wh^2)$, the number of iterations is bounded by
$k = \bigO(wh^2)$.

Therefore, the total worst-case running time of
\FindTargetRedundantFutileEdge{} is
\[
O(k \cdot T_f) = \bigO(wh^2 \cdot T_f).
\]
\end{proof}
\end{lemma}

\SetKwFunction{VerifyExistenceOfWalk}{VerifyExistenceOfWalk} 
\SetKwFunction{FindTargetRedundantFutileEdge}{FindTargetRedundantFutileEdge}
\SetKwFunction{ComputeFeasibleGraph}{ComputeFeasibleGraph}

\subsection{Verifying Computation Walk}\label{subsec:verification_of_computation_walk}

We now present a formal algorithm to verify the existence of a valid computation walk containing the verification target edge within a given computation graph. This verification procedure is designed to ensure both the \emph{soundness} (only valid walks are preserved) and \emph{completeness} (no valid walk is excluded) of the underlying grid-aligned graph semantics.

To determine whether a \textit{computing-targeted walk} exists containing the verification target $e_t$, we define the procedure \VerifyExistenceOfWalk{} in \cref{alg:walk_verification}. 
The algorithm first constructs a feasible subgraph with respect to the target edge $e_t$. It then enters a refinement loop, iteratively identifying and removing edges that are categorized as either \textit{computing-redundant} or \textit{computing-futile} through the detection mechanism established in the previous subsection.

If the process confirms that $e_t$ is stable and reachable even after the exhaustive removal of non-essential structures, or if $e_t$ is directly validated by \FindTargetRedundantFutileEdge{}, it certifies the existence of a computing-targeted walk. Conversely, if the graph collapses such that $e_t$ is no longer reachable, it proves that no such walk existed in the original configuration.

This verifier algorithm plays a role analogous to an \textbf{NP verifier}: it operates in deterministic polynomial time over a computation graph whose size is strictly bounded by $\bigO(wh^2)$ (where $w$ and $h$ denote its width and height, respectively). Since the number of edges is polynomial, the iterative pruning of \textit{computing-futile} or \textit{redundant} edges must terminate within polynomial bounds, effectively reducing the graph to its minimal, unique computation walk if one exists.

\begin{algorithm}
\caption{Verification of Walk} \label{alg:walk_verification}
\Input{$G$: a computation graph, $v_0$: unique initial node, $e_t$: verification target edge}
\Output{The Verification Result \True/\False}
\Function{\VerifyExistenceOfWalk{$G_U, v_0, e_t$}} {
\State{Let $E_f \gets \{e_t\}$ and $V_0 \gets \{v_0\}$}
\State{Set $G \gets$ \ComputeFeasibleGraph{$G_U, V_0, E_f $}}
\While{$G$ contains any edge of $E_f$} { \label{alg_line:walk_verification:start_of_while}
    \State{Let $e \gets$ \FindTargetRedundantFutileEdge{$G_U, G, v_0, e_t$}} \label{alg_line:walk_verification:find_obsolete_edge}
    \If{$e = e_t$} {
        \State{\Return \True}
    }
    \ElseIf{$e =NIL$} {
        \State{\Return \False}
    }
    \State{Set $G \gets G-e$} \label{alg_line:walk_verification:removing_an_edge}
    \State{Set $G \gets$ \ComputeFeasibleGraph{$G, V_0, E_f$}}\label{alg_line:walk_verification:recompute_feasible_graph}
} \label{alg_line:walk_verification:end_of_while}
\State{\Return \False}
}
\end{algorithm}

\begin{lemma}[Correctness of Verifying Existence of Walk] \label{lem:verify_walk_correctness}
The procedure \VerifyExistenceOfWalk{} in \cref{alg:walk_verification} correctly decides the existence of a computing-targeted walk containing $e_t$ within an $e_t$-augmented footmark graph $G$ with a unique initial node $v_0$.
Specifically, the algorithm returns \True only if there exists a computation walk containing $e_t$ in $G$. Furthermore, if $G$ is a grid-aligned footmark graph, the existence of such a walk ensures that the algorithm returns \True.
\end{lemma}

\begin{proof}
Let $G^{(0)}$ be the initial feasible graph constructed by \ComputeFeasibleGraph{} from the input graph $G$ with the initial vertex set $V_0=\{ v_0 \}$, and the target edge set $E_f = \{e_t\}$.
Let $G^{(i)}$ denote the feasible graph at \cref{alg_line:walk_verification:recompute_feasible_graph} in $i$-th iteration of the \texttt{while} loop.

We define the following loop invariant, maintained at the beginning of each iteration of the \texttt{while} loop in Line~\ref{alg_line:walk_verification:start_of_while}.

\textbf{Invariant:} A computation walk from some node in $V_0$ to $e_t$ exists in $G^{(i)}$ if and only if one existed in the initial graph $G^{(0)}$.

The correctness is established by verifying that the invariant holds initially, is preserved through each iteration, and ultimately guarantees the validity of the final result.

\paragraph{(1) Initialization.}
Initially, $G^{(0)} \gets$ \ComputeFeasibleGraph{}, which preserves all feasible walks with respect to $E_f$ by \cref{lem:preservation_of_feasible_walks} where $e_t\in E_f$. Hence, a computing-targeted walk to $e_t$ exists in $G^{(0)}$ if and only if one existed in the original graph $G$. Thus, the invariant holds prior to the first iteration.

\paragraph{(2) Maintenance.}
At each iteration, if an edge $e$ returned by the function \FindTargetRedundantFutileEdge{} is not the verification target edge, the algorithm removes it; otherwise, it terminates and returns \True.

By \cref{lem:find_feasible_or_disjoint_walk}, if $e \ne e_t$, then $e$ is either computing-redundant or computing-futile. Therefore, removing $e$ (i.e., constructing $G - e$) preserves at least one computing-targeted walk to $e_t$.

After removing $e$, the algorithm recomputes the feasible subgraph using \ComputeFeasibleGraph{}, which preserves all feasible walks with respect to $E_f$ by \cref{lem:preservation_of_feasible_walks}. Thus, any computing-targeted walk to $E_f$ that existed prior to the removal of $e$ continues to exist in the updated graph. In particular, since edges are only removed (never added), no new walks to $e_t$ can be introduced during pruning. Therefore, the invariant is maintained in $G^{(i+1)}$.

\paragraph{(3) Termination.}
The algorithm must terminate since each iteration removes at least one edge from $G^{(i)}$, and the total number of edges is finite. Therefore, the number of iterations is bounded by $|E(G)|$, and $G^{(i)}$ will eventually become empty when no further edges remain.

The algorithm terminates in one of two cases:
\begin{itemize}
    \item \textbf{Case 1:} The algorithm returns \textbf{True} when the function \FindTargetRedundantFutileEdge{} returns the verification target edge $e_t$. By the invariant, a feasible walk with respect to $e_t$ must exist in $G^{(i)}$, and therefore must have existed in $G^{(0)}$. (Soundness)

    \item \textbf{Case 2:} The algorithm returns \textbf{False} when \FindTargetRedundantFutileEdge{} returns \NIL, or if the computed feasible graph $G^{(i)}$ contains no edges from $E_f$.
 By the invariant, no feasible walk with respect to $e_t$ exists in $G^{(i)}$, and hence no computing-targeted walks existed in $G^{(0)}$. (Completeness)
\end{itemize}

Hence, \VerifyExistenceOfWalk{} returns \textbf{True} if and only if a computing-targeted walk to $e_t$ exists in the initial feasible graph $G^{(0)}$.

Therefore, the algorithm correctly verifies the existence of a computation walk containing $e_t$, completing the proof.

\end{proof}

\begin{lemma}[Time Complexity of Verification of Walk]
\label{lem:verify_existence_of_walk_time}
Let $G$ be a computation graph of width $w$ and height $h$, with
$|E(G)| = \bigO(wh^2)$.
Let $T_f$ be the worst-case time complexity of
\ComputeFeasibleGraph{}, and let $T_v$ be the worst-case
time complexity of \VerifyExistenceOfWalk{}.
Then,
\[
T_v = \bigO(w^{2}h^{4} \cdot T_f).
\]

\begin{proof}
The algorithm \VerifyExistenceOfWalk{} consists of a
\textbf{while}-loop that iteratively removes edges from the feasible
graph until either the target edge $e_t$ is found or no further edges
remain.

We analyze the cost of a single iteration.

\begin{itemize}
    \item \textbf{\FindTargetRedundantFutileEdge{}:}
    This step runs in $\bigO(wh^{2} \cdot T_f)$ time by
    \cref{lem:find_target_or_futile_edge_time}.

    \item \textbf{Edge removal and update:}
    Removing an edge from the adjacency structure costs $\bigO(h)$ time,
    which is negligible compared to the dominant terms.

    \item \textbf{Feasible graph recomputation:}
    Recomputing the feasible graph incurs an additional $T_f$ cost.
\end{itemize}

Thus, the per-iteration cost is dominated by
$\bigO(wh^{2} \cdot T_f)$.

At least one edge is removed in each iteration.
Since the computation graph $G$ contains at most $\bigO(wh^2)$ edges, the total
number of iterations is bounded by $\bigO(wh^2)$.

Therefore, the overall worst-case time complexity is
\[
T_v = \bigO(wh^2 \cdot wh^2 \cdot T_f)
    = \bigO(w^{2}h^{4} \cdot T_f).
\]
\end{proof}
\end{lemma}

In this section, we established that the proposed algorithm \VerifyExistenceOfWalk{} correctly and completely decides the existence of a computing-targeted walk ending at a given final edge within the computation graph. 
Furthermore, we proved that the algorithm runs in polynomial time with respect to the input size, providing an effective and efficient procedure for verifying computation walks.

\section{Proof of \textsf{P=NP}}\label{sec:np_is_p}

In this section, we establish that $P=NP$ by constructing a deterministic polynomial-time algorithm that simulates the branching behavior of a deterministic verifier for multiple symbols. The proof is organized into three progressive subsections:

\begin{itemize}
    \item \textbf{Incremental Extension of Footmarks Graph:} We describe an incremental edge extension procedure to the verified footmark graph by selectively adding edges corresponding to valid transitions. This process leverages the \VerifyExistenceOfWalk{} mechanism developed in the previous section, forming the deterministic core for simulating multiple transitions from a single computation node.
    \item \textbf{Transformation to Polynomial-Time Simulation:} We demonstrate how the traditionally exponential-time simulation of all certificates can be reduced to polynomial time by utilizing the above edge-extension strategy, effectively bypassing the brute-force search space.
    \item \textbf{Generalization to NP:} We show that this simulation strategy is applicable to all problems within the NP class, thereby concluding that $NP \subseteq P$, which implies $P = NP$.
\end{itemize}

\subsection{Extending Footmarks of Computation Walks}\label{subsec:extending_footmarks}

In this subsection, we present a method to incrementally construct a \textbf{footmark graph} representing all valid computation walks by extending a previously verified subgraph of the footmarks. The algorithm begins with a known subgraph $H$---defined as the verified history where all edges belong to at least one valid computation walk---and identifies its \textit{boundary edges}, which are potential transitions connecting $H$ to unexplored nodes within the full computation space $G$.

The extension proceeds through a rigorous \textit{selection-and-verification} cycle:
\begin{enumerate}
\item \textbf{Candidate Selection:} The boundary edges which have index-precedent edges in the verified footmarks $H$, or which are floor edges, are identified as potential candidates for extending the current verified domain $H$. At this stage, these edges are considered under minimal restrictions to ensure the completeness of the search space.
\item \textbf{Targeted Verification:} Each candidate edge is designated as a \textit{verification target edge} ($e_t$). We then invoke the \VerifyExistenceOfWalk{} procedure to determine if $e_t$ can be part of a globally consistent \textit{computing-targeted walk} to the target edge $e_t$.
\item \textbf{Promotion to Footmarks:} An edge is promoted to a \textit{footmarks edge} and integrated into $H$ only if its feasibility is confirmed. This ensures that every edge in $H$ is anchored to a valid path reaching the final accept or reject states.
\end{enumerate}

This iterative process continues until an accepting node is reached or no further extension is possible because no candidate edges can be further expanded. By transforming the exploration of an exponential certificate space into a sequence of deterministic verification steps, this strategy provides a polynomial-time mechanism for exploring valid computation paths without encountering exponential branching.

\begin{definition}
Given computation graph $G$ and its subgraph $H$, a \emph{boundary edge} of $H$ in $G$ is defined as any edge $(u, v)$ in $G$ such that $u \in V(H)$ and $v \notin V(H)$.
\end{definition}

\begin{algorithm}[ht]
\caption{Compute Footmarks of Computation Walks and Determine Acceptance} \label{alg:compute_footmarks}
\Input{$G$: Dynamic Computation Graph, $H$: Graph of Verified Edges, $v_0$: Unique Initial Node}
\Output{Whether an accept state $\qacc$ is reachable by a maximal computation walk}
\SetKwFunction{IsAcceptedOnFootmarks}{IsAcceptedOnFootmarks}
\SetKwFunction{CollectBoundaryEdges}{CollectBoundaryEdges}
\SetKwFunction{ExtendByVerifiableEdges}{ExtendByVerifiableEdges}
\Function{\IsAcceptedOnFootmarks{$G, H, v_0, \qacc, \qrej$}} {
\State{Let $Q \gets$ \CollectBoundaryEdges{$G, H, \qrej$}}
\While(\tcc*[f]{Extend $H$ by valid computation edges}){$Q \ne \emptyset$} {
    \State{Let $Ev \gets \emptyset$}
    \StateC{\ExtendByVerifiableEdges{$v_0, Q, H, E_v$}}\Comment{$E_v$: verified extension edges}
    \If(\tcc*[f]{No feasible edge extended}){$E_v = \emptyset$} {
        \State{\Return \False}
    }
    \ElseIf{there exist $state(v)=\qacc$ for some $(u,v)\in E_v$} {
        \State{\Return \True}
    }
    \StateC{$Q \gets$ \CollectBoundaryEdges{$G,H$}}\Comment{newly collected boundary edges}
}
\State{\Return \False}
}

\Function{\CollectBoundaryEdges{$G, H, \qrej$}} {
\State{Let $Q$ be the empty set of edges}
\ForAll(\tcc*[f]{Collect boundary edges}){node $v$ in $V(H)$} {
    \State{Let $h \gets$ the maximum tier of all node $w \in V(H)$ with $\indexOf(w)=\nextIndex(v)$}
    \ForAll{edge $e'=(v,w) \in \Outgoing_G(v)$ with $\tier(w) \le h+1$ and $\state(w)\ne \qrej $} {
        \If{$e' \notin E(H)$ and ($e'$ is a floor edge or $\IPrec_{H'}(e')$ is not empty) where $H'=H+e'$} {
             \State{Add $(v,w)$ to $Q$}
        }
    }
}
}

\Procedure{\ExtendByVerifiableEdges{$v_0$: \In, $Q$:\In, $H$:\InOut, $E_v$:\Out}} {
\ForAll{edge $e \in Q$} {
    \If{\VerifyExistenceOfWalk{$H+e, v_0, e$}} {
        \State{Add $e$ to $E_v$}
        \State{$H \gets H + e$}
    }          
}
}
\end{algorithm}

\begin{lemma}[Correctness of \ExtendByVerifiableEdges{}]
\label{lem:traverse-feasible-edges}
Let $G$ be a computation graph, $H \subseteq G$ a grid-aligned footmark graph, and $v_0$ a unique initial vertex.
Let $Q$ be the set of boundary edges $(u,v)$ such that $u \in V(H)$ and $v \in V(G) \setminus V(H)$.
Let $E_v$ denote the set of edges verified to be computing-effective by \ExtendByVerifiableEdges{} and subsequently added to $H$.
Then, the algorithm correctly collects into $E_v$ exactly those edges $e = (u,v) \in Q$ for which a computation walk to $e$ exists within the resulting incremented grid-aligned footmark graph $H$.

\begin{proof}
We prove correctness by invariant-based induction on the number of processed edges from the queue $Q$.
\begin{itemize}
\item \textbf{Invariant:}  
The set $E_v$ contains exactly those edges $(v,w)$ such that a computation walk exists from the initial vertex $v_0$ to $e = (v,w)$ in the current graph $H$ (the grid-aligned footmark graph) constructed so far.  
Moreover, $H$ contains exactly those edges for which a computation walk exists from $v_0$.

\item \textbf{Base Case:}  
Initially, $E_v = \emptyset$ and $Q$ consists of the set of boundary edges from $V(H)$ to $V(G) \setminus V(H)$.  
No edges have been verified, and $H$ contains only the initially verified subgraph. Hence, the invariant holds trivially.

\item \textbf{Maintenance:}  
Assume that after the $k$-th iteration, the invariant holds. Now consider the $(k+1)$-th edge $e = (v,w) \in Q$:
\begin{itemize}
\item[-] If \VerifyExistenceOfWalk($H + e, v_0, e$) returns \texttt{True}, then by \cref{lem:verify_walk_correctness}, there exists a computation walk to $e$ in $H + e$. In this case, $e$ is added to both $E_v$ and $H$.
\item[-] Otherwise, $e$ is skipped, and $E_v$ remains unchanged. 
\end{itemize}

Thus, only feasible edges are added, and no infeasible edge is mistakenly included. The invariant remains true.

\item \textbf{Termination:}  
Since each edge in $Q$ is processed at most once and $Q$ is finite (bounded by $|E(G)|$), the procedure terminates after a finite number of steps.

\item \textbf{Conclusion:}  
By induction, upon termination, $E_v$ contains exactly all edges in $Q$ for which a computation walk exists in $H + e$ from $v_0$ to $e$.  
Thus, the procedure is both sound (only grid-aligned footmark edges are included) and complete (all grid-aligned footmarks edges are found).  
Moreover, the updated graph $H$ contains exactly the edges reachable via computation walks from $v_0$.
\end{itemize}
\end{proof}
\end{lemma}

\begin{lemma}[Correctness of Checking Acceptance of Footmarks] \label{lem:correctness_of_acceptance_in_footmarks}
Let $G$ be a computation graph, and let $H \subseteq G$ be a subgraph representing the grid-aligned footmarks of some computation walks;
 that is, $H$ is initialized to contain exactly those vertices reachable from $v_0$ via previously verified computation walks.  
Then the algorithm \IsAcceptedOnFootmarks{} terminates and satisfies:

\begin{enumerate}
    \item It expands $H$ by adding all feasible edges and vertices reachable from $v_0$ via computation walks in $G$.
    \item It returns \texttt{True} if and only if there exists such a walk that reaches an accepting state $\qacc$.
\end{enumerate}

\begin{proof}
We prove correctness by maintaining a loop invariant across iterations of the \texttt{while} loop in the algorithm.
\begin{itemize}
\item \textbf{Invariant:} At the beginning of each iteration:
\begin{itemize}
    \item $H$ contains exactly those vertices and edges that are reachable from $v_0$ by a computing-targeted walk,
    \item $Q$ contains only boundary edges that have not yet been verified as feasible,
    \item If an accepting state $q_\text{acc}$ has been reached by any verified walk, the algorithm returns immediately with \texttt{True}.
\end{itemize}

\item \textbf{Base Case:}  
Initially, $Q$ contains all boundary edges of $H$ with index-precedent consitency collected by \CollectBoundaryEdges{}.  
No edge has yet been verified as feasible.  
Since $H$ is assumed to be constructed from valid prior computation steps, the invariant holds.

\item \textbf{Inductive Step (Maintenance):}  
Suppose the invariant holds at the start of an iteration.
\begin{itemize}
\item The algorithm invokes \ExtendByVerifiableEdges{} with the current $Q$ and $H$.
\item By \cref{lem:traverse-feasible-edges}, this procedure correctly identifies and adds all feasible footmark edges from $Q$ to $H$, and records them in $E_v$.
\item If any such edge leads to a node in accepting state $\qacc$, the algorithm terminates and returns \texttt{True}.
\item If no feasible footmark edge is found ($E_v = \emptyset$), then no further expansion is possible, and the algorithm correctly returns \texttt{False}.
\item Otherwise, the verified footmark edges in $E_v$ are removed from $Q$, and the algorithm proceeds.
\end{itemize}

Hence, the invariant continues to hold: $H$ grows only via footmark edges, $Q$ shrinks by removing only verified entries, and accepting paths are immediately detected.

\item \textbf{Termination:}  
In each iteration, at least one edge is removed from $Q$ and never re-added.  
Since $Q$ contains only boundary edges, and the number of edges in $G$ is finite, the total number of iterations is bounded by $|E(G)|$.  
Therefore, the loop terminates in finite time.

\item \textbf{Conclusion:}  
Upon termination, $H$ contains all vertices and edges reachable from $v_0$ by computing-targeted walks in $G$.  
The function returns \texttt{True} if and only if there exists such a walk that reaches an accepting state.  
\end{itemize}
\end{proof}
\end{lemma}

\begin{remark}[Re-evaluation of Candidate Edges]
It is important to note that an edge $e$ that fails the validation test in a given iteration is not permanently discarded. As the subgraph $H$ expands through the addition of other footmark edges, the structural conditions for $e$ (such as being part of a computation-targeted walk) may subsequently be satisfied. 
Therefore, the algorithm ensures that such candidate edges are re-evaluated in light of the updated footmarks $H$, allowing previously unverified edges to be identified as feasible in later steps.
\end{remark}

\begin{lemma}[Time Complexity of Checking Acceptance of Footmarks]  
\label{lem:time-complexity-compute_footmarks}  
Let $H$ be a constructed dynamic computation graph with width $w$ and height $h$ via iterative edge extension,  
and let $T_v$ denote the time complexity of \VerifyExistenceOfWalk{}. 
Then, the total time complexity of \IsAcceptedOnFootmarks{} in \cref{alg:compute_footmarks} is bounded by
\[
O(w^2 h^4 \cdot T_v).
\]
\begin{proof}
In \IsAcceptedOnFootmarks{}, the graph is extended at most along the edges of the footmark graph of all computation walks of all certificates (see \cref{def:footmarks_graph}),
 whose total number of edges $|E(H)|$ is $\bigO(wh^2)$ by definition.

At each iteration, at least one edge is extended; otherwise, the loop terminates because $E_v$ becomes empty.  

Each call to \ExtendByVerifiableEdges{} costs $\bigO(wh^2 \cdot T_v)$ by \cref{lem:time_complexity_extending_edges_by_verification},  
and each edge is processed at most once.  
Hence, the cumulative cost of all verification calls is $\bigO(|E(H)| \cdot (wh^2) \cdot T_v)=\bigO(w^2 h^4 \cdot T_v)$.  

Similarly, each call to \CollectBoundaryEdges{} costs $T_c = \bigO(wh^3)$ and can be invoked at most $\bigO(wh^2)$ times,  
resulting in a total cost of $\bigO(wh^2 \cdot T_c) = \bigO(w^2 h^5)$.
This is asymptotically dominated by the verification cost $\bigO(w^2 h^4 \cdot T_v)$ (see \cref{lem:verify_existence_of_walk_time}).  

The computation graph $H$ is maintained as a dynamic data structure.  
When a tape cell index is accessed for the first time, extending the underlying structure may require copying existing data, incurring a worst-case cost of $\bigO(wh^2)$.  
Since this occurs at most once per cell index, the overall overhead from dynamic graph expansion is $\bigO(wh^2)$,  
which is also asymptotically dominated by $\bigO(w^2 h^4 \cdot T_v)$.  

Therefore, the time complexity of \IsAcceptedOnFootmarks{} is $\bigO(w^2 h^4 \cdot T_v)$.
\end{proof}
\end{lemma}

In this subsection, we presented algorithms that, starting from already visited nodes, explore boundary edges to expand and verify all valid computation paths within polynomial time. We proved their correctness and analyzed their time complexity by leveraging the verification algorithm for computation walks to specific edges.

\SetKwFunction{SimulateVerifierForAllCertificates}{SimulateVerifierForAllCertificates}
\subsection{Polynomial-Time Algorithm for NP Problems: Simulating All Certificates}\label{subsec:simulate_verifier_poly}

In this subsection, we present a deterministic polynomial-time algorithm that simulates an NP verifier over a dynamic computation graph. This simulation replaces the non-deterministic certificate-guessing process with a structured traversal of all valid computation paths. By doing so, the algorithm provides a deterministic decision procedure for NP verification, supporting the foundational claim that \textsf{P} = \textsf{NP}.

The primary objective here is to bridge the gap between verifying a computation walk in a fixed graph and simulating a verifier over the entire space of potential certificates. To prove \textsf{P} = \textsf{NP}, it is not enough to show that a computation walk can be found in an arbitrary graph; we must demonstrate that our dynamically constructed footmark graph $H$ faithfully encapsulates the footmarks of \emph{all} possible NP-certificates.

During the simulation, the algorithm maintains two synchronized structures:
\begin{itemize}
    \item \textbf{NP Dynamic Computation Graph $G$:} This represents the universal computation space that encodes every potential transition corresponding to any valid NP-certificate and its associated computation walk.
    \item \textbf{Footmarks Subgraph $H$:} This records the verified history of previously explored computation walks, acting as a repository of "guaranteed" paths to ensure efficiency and prevent duplicated verification.
\end{itemize}

Unlike exponential-time simulations that explicitly construct a complete computation tree for each certificate, our method expands the grid-aligned footmark graph $H$—representing the union of computation walks for all certificates—\emph{on demand} based on verified feasibility. This controlled expansion ensures that we only generate the specific subset of the universal graph $G$ necessary to decide the language, thereby maintaining polynomial-time execution.

By the definition of \textsf{NP}, every language $\mathcal{L}$ is governed by a verifier $M$. Rather than reducing $\mathcal{L}$ to a specific problem, our framework directly analyzes the state-transition graph of $M$ over the symbolic space of all possible certificates $Y \in \Sigma^{q(|X|)}$. 
Furthermore, for every deterministic Turing machine, there exists an equivalent oblivious Turing machine with the same input and at most a quadratic overhead in running time (see \cref{thm:pippenger_fischer} and \cref{rem:verifier_otm}).
This transformation allows us to impose a grid-aligned structure on the corresponding footmarks, which is necessary for ensuring the polynomial-time efficiency of our graph expansion algorithm."

We represent this space as a grid-aligned Footmarks Graph where each node encapsulates the 6-tuple configuration of $M$ at a given time and tape position. 

We adopt the deterministic certificate-oblivious \textsf{NP} verifier $M = (Q, \Sigma, \Gamma, \delta, q_0, \qacc, \qrej)$ and the input configuration—consisting of the problem instance, the fixed input string $\Lfixed$, and the certificate length $m$—as formally defined in \cref{subsec:simulate_verifier_exp}.
 The simulation processes the fixed problem instance $\Lfixed$ followed by a symbolic certificate region, as illustrated in \cref{fig:verifier_tape_area}. 
 The underlying machine model and input tape structure are identical to those employed in the brute-force algorithm detailed in Appendix \cref{subsec:simulate_verifier_exp}.

The transition from a non-deterministic search to a deterministic simulation rests upon two formal pillars:
\begin{itemize}
    \item \textbf{Soundness:} If a computation walk is identified in the dynamically constructed graph $G$, there must exist at least one specific certificate string that induces that sequence of transitions in the original verifier Turing Machine.
    \item \textbf{Completeness (The Universal Coverage):} The dynamic construction process must ensure that every edge capable of leading to an accepting configuration under \emph{any} certificate is included in $G$. This ensures that no valid certificate is overlooked during the deterministic traversal.
\end{itemize}

The algorithm \SimulateVerifierForAllCertificates{} iteratively applies the verification mechanisms developed in previous sections to expand the frontier of $H$ within the bounds of $G$. This approach relies on a crucial structural observation: rather than re-evaluating every certificate explicitly, it suffices to incrementally extend the verified boundary from known footmarks. By transforming the exponential breadth of certificates into a polynomial growth of verified edges, we achieve a deterministic simulation of the non-deterministic verifier.

Crucially, since the verifier $M$ operates deterministically once a certificate symbol is fixed, the transition from each node in $G$ is uniquely determined by the choice of the symbol. Consequently, we refer to the \textbf{outgoing edges} of a node as its \textbf{next edges}. It is important to note that because multiple certificate symbols (e.g., $\{T, F\}$) may be available at a given configuration, a single node can give rise to multiple next edges, each serving as a \textbf{floor edge} for its respective branching computation path.

\SetKwFunction{Initialize}{Initialize}
\SetKwFunction{GetNextEdges}{GetNextEdges}
\begin{algorithm}[H]
\caption{Simulate Verifier For All Certificates} \label{alg:compute_verifier_in_polynomial}
\Input{
$\Lfixed$: problem instance string ending the delimiter `\#', $m$: certificate length,\\ 
$q_0$: initial state of verifier TM, $\Sigma$: input alphabet, $\delta$: transition function, \\
$Q$: Set of states of machine, $\qacc$: Accepting state, $\qrej$: Rejecting state
}
\Output{
  The decision result: \texttt{Yes} if any certificate makes verifier accept, otherwise \texttt{No}
}
\Function{\SimulateVerifierForAllCertificates{$\Lfixed, m, q_0, \Sigma, \delta, Q, \Gamma, \qacc, \qrej$}} {
\State{Let $G$ be a NPDynamicComputationGraph}
\State{Let $G$.\Initialize{$\Lfixed, m, q_0, Q, \Sigma, \delta, \Gamma, \{ \qacc, \qrej \}$}}
\StateC{Let $s \gets \Lfixed[0]$} \Comment{Problem Instance is not empty string}
\StateC{Let $v_0$ be the unique node in $V_{0,0}^{q_0,s}$} \Comment{$v_0$: vertex at index 0, tier 0, state $q_0$, symbol $s$}
\StateC{Let $E_0 \gets$ $G$.\GetNextEdges{$v_0, 0$}}\Comment{$\Outgoing_G(v_0)$}
\State{Let $V_0 \gets \{v_0\}$}
\StateC{Let $H \gets G(V_0,E_0)$}\Comment{Computation Graph}
\If(\tcc*[f]{$v_0$:Initial vertex where $state(v_0)=q_0$}){\IsAcceptedOnFootmarks{$G, H, V_0, \qacc, \qrej$}} {
    \State{\Return \texttt{Yes}}
}
\State{\Return \texttt{No}}
}
\end{algorithm}

\begin{algorithm}
\SetKwFunction{Initialize}{Initialize}
\SetKwFunction{GetFloorNextEdges}{GetFloorNextEdges}
\SetKwFunction{GetNonFloorNextEdges}{GetNonFloorNextEdges}
\SetKwFunction{GetNextEdges}{GetNextEdges}

\caption{NP Dynamic Computation Graph } \label{alg:np_dynamic_computation_graph}
\AlgDescription{This Graph is a Dynamic Computation Graph for NP verifier input string and all certificates}
\SetKwFunction{NPDynamicComputationGraph}{NPDynamicComputationGraph}
\Class{\NPDynamicComputationGraph \textbf{extends} \textsf{DynamicComputationGraph}} {
\State{\Field $q_0$ : a Turing machine state representing the initial state}
\State{\Field $Q$ : Set of all states of the Turing machine}
\State{\Field $\Lfixed$ : Fixed tape input string (problem instance string ending with the delimiter `\#')}
\State{\Field $m$ : Integer representing length of certificate}
\State{\Field $\Sigma$: Set of all input symbols of the Turing machine}
\State{\Field $\Gamma$: Set of all symbols of the Turing machine}
\State{\Field $\delta$:  Transition function of the Turing machine}
\State{\Field $V$: Dynamic Array representing computation node}
\State{\Field $F$: Set of final state of the Turing machine}

\Function{\Initialize{$\Lfixed', m', q_0', Q', \Sigma', \delta', \Gamma', F'$}} {
	\State{Set $(q_0,\Lfixed, m, \Sigma, Q, \delta, \Gamma, F) \gets (q_0', \Lfixed', m', \Sigma', Q', \delta, \Gamma', F')$}
	\State{Set $\Gamma \gets \Sigma' \cup \{\epsilon\}$ if $\Gamma'$ is $\NIL$ otherwise $\Gamma'$}
}
 
\Function{\GetNextEdges{$v, t_m$}} {
        \State{Let $E \gets$ \GetFloorNextEdges{$v$}}
        \State{Let $E' \gets$ \GetNonFloorNextEdges{$v, t_m$}}
	\State{\Return $E \cup E'$}
}

\Function{\GetNonFloorNextEdges{$v, t_m$}} {
	\State{Let $E \gets \emptyset$}
	\State{Let $(i', q') \gets (\nextIndex(v), \nextState(v))$}
	\ForAll{$(s', t')$ such that $s' \in \Gamma$ \textbf{and} $0<t' \le t_m$} {
		\ForAll{$v' \in V_{i',0} ^{q',s'}$} {
	                \State{Add $(v,v')$ to $E$} 
		}
	}
	\State{\Return $E$}
}

\Function{\GetFloorNextEdges{$v$}} {
	\State{Let $E \gets \emptyset$}
	\State{Let $E$ be an empty set of computation edges}
	\State{Let $(i', q') \gets (\nextIndex(v), \nextState(v))$}
	\If{$\Lfixed[i']$ is defined} { 
	  	\State{Let $v' \gets$ the unique node $u \in V_{i',0}^{q',s'}$ with $s'=\Lfixed[i']$}
                \State{Add $(v,v')$ to $E$} 
	}
	\ElseIf{$i'<0$ \textbf{or}  $i' > |\Lfixed|+m$} {
		\State{Let $v' \gets$ the unique node in $V_{i',0}^{q',s'}$ with $s'=\epsilon$}
                \State{Add $(v,v')$ to $E$} 
        }
        \Else{
		\ForAll{$s' \in \Sigma$} {
			\State{Let $v' \gets$ the unique node in $V_{i',0}^{q',s'}$}
	                \State{Add $(v,v')$ to $E$} 
		}
	}
	\State{\Return $E$}
}
}
\end{algorithm}

\begin{remark}[Structural Inheritance of Computation Graphs]
For the convenience of formal proof and to clarify the relationship between different graph models, we adopt the concept of \textbf{inheritance} from object-oriented programming. The graph $G$ utilized in this simulation is treated as a \textit{specialized subclass} of the general computation graph defined in \cref{sec:comp_model}. It inherits all fundamental properties of the base graph while incorporating specific constraints—such as the deterministic transition rules of $M$ and the branching floor edges dictated by certificate symbols—to model the verifier’s execution space precisely.
\end{remark}

\begin{sublemma}\label{lem:floor_edge_only_for_valid_input_string}
The method \GetFloorNextEdges() returns exactly the set of floor edges
that are consistent with the problem instance string $\Lfixed$ and all possible certificate
strings of length $m$.
\begin{proof}
The method \GetFloorNextEdges() generates a floor edge $(u,v)$ where $v$
corresponds to the tape cell indexed by $i'$, at which the next read or write
operation occurs.
The following cases are distinguished:
\begin{itemize}
\item If $i' < |\Lfixed|$ (input region), the algorithm deterministically
returns the unique edge to the node labeled with $s' = \Lfixed[i']$.

\item If $i' \in [|\Lfixed|, |\Lfixed|+m)$ (certificate region), the
algorithm returns all edges to nodes associated with with each $s' \in \Sigma$, thereby
enumerating all possible certificate symbols at that position.

\item If $i' < 0$ or $i' \ge |\Lfixed|+m$ (out-of-bounds access), the algorithm
returns exactly the edge labeled with the blank symbol $\epsilon$.
\end{itemize}

Hence, every floor edge consistent with the problem instance string $\Lfixed$ and some certificate
string of length $m$ is generated by \GetFloorNextEdges(), and no other
edges are returned.
\end{proof}
\end{sublemma}

\begin{lemma}[Completeness of Simulating Verifier for All Certificates] 
\label{lem:completeness_of_simulate_verifier_for_all_certificate}
Let $M$ be a certificate-oblivious NP verifier Turing machine and $\Lfixed=L\#$ be an input string for a problem instance.
For any certificate $Y$ on a problem instance $\Lfixed$, \SimulateVerifierForAllCertificates{} in \cref{alg:compute_verifier_in_polynomial} correctly simulates all valid accepting computation walks. That is, for every accepting computation walk corresponding to some $Y$, the constructed grid-aligned footmark graph $H$ contains all necessary edges to simulate it. 

\begin{proof}
We proceed by contradiction.
Note that the footmark graph generated by the certificate-oblivious TM is grid-aligned by \cref{lem:cotm_to_grid_aligned}.
Suppose there exists a valid accepting computation walk $W$ for some certificate $Y$ such that $M(\Lfixed \#Y)$ accepts, but $W$ cannot be fully simulated. Then there exists a node $v$ in $W$ and a valid next node $v'$ such that the edge $e=(v,v')$ is missing from the constructed graph $H$.
\begin{itemize}
\item \textbf{Case 1: $e$ is a floor edge.}  
By \cref{lem:floor_edge_only_for_valid_input_string}, every floor edge that is
consistent with the input string $\Lfixed$ and some certificate of length $m$ is necessarily generated by \GetFloorNextEdges{}.
Since $W$ is a valid accepting computation walk for some certificate $Y$, the corresponding floor edge $e$ must be included in $G$, yielding a contradiction.
Hence, every valid floor edge is included in $H$.

\item \textbf{Case 2: $e$ is a non-floor edge.}  
These are transitions involving the surface of the computation walk:
\begin{itemize}
    \item \GetNonFloorNextEdges{} returns all edges consistent with the verifier’s transition function up to tier $t_m = h+1$, where $h$ is the maximum tier among the nodes in the \textbf{already expanded footmarks} $H$ at the corresponding index. These edges constitute the candidate set collected in \CollectBoundaryEdges{}, ensuring a controlled and incremental expansion of the frontier.
    \item The surface of the walk uniquely determines the tiers of non-floor edges, ensuring no valid edge is omitted.
\end{itemize}

In both cases, the edge $e$ must be present in $G$, contradicting the assumption that it is missing.

\item \textbf{Conclusion.}  
Since the floor edges corresponding to the problem instance string and all potential certificates are exhaustively included based on the input instance $\Lfixed$, and the non-floor edges are uniquely determined by the surface, every accepting computation walk for any certificate is represented in $G$.
Thus, the algorithm is complete.
\end{itemize}
\end{proof}
\end{lemma}
\begin{remark}[Remark on Algorithmic Dependency]
The case distinction above is enforced by the structural constraints established in the preceding subsections. In particular, \CollectBoundaryEdges{} restricts candidate boundary edges to those whose head tier is at most $h+1$, where $h$ denotes the maximum tier among nodes in the \textbf{already expanded footmarks} $H$ at the corresponding next index. As a result, floor edges are treated separately since their validity depends only on consistency with the input string $\Lfixed$ (cf.~\cref{lem:floor_edge_only_for_valid_input_string}), whereas non-floor edges are generated only when justified by the surface via the predicate $\IPrec_{H'}(e')$.
\end{remark}

\begin{lemma}[Correctness of Simulating Verifier for All Certificates]\label{lem:soundness_of_simulate_verifier_for_all_certificate}
Let $M$ be a certificate-oblivious NP verifier and $\Lfixed=L\#$ be an input string for a problem instance.
Then, \SimulateVerifierForAllCertificates{} in \cref{alg:compute_verifier_in_polynomial} returns \texttt{Yes} if and only if there exists a certificate $Y$ of length $m$ such that $M(L\#Y) = \texttt{accept}$.
\begin{proof}

First, suppose that \SimulateVerifierForAllCertificates{} returns \texttt{Yes}.
By \cref{lem:correctness_of_acceptance_in_footmarks}, there exists a computation walk reaching an accepting state $\qacc \in F$.
Moreover, by \cref{lem:floor_edge_only_for_valid_input_string}, every floor edge of this walk corresponds to a symbol of some certificate $Y$ of length $m$.
Hence, the walk simulates a valid accepting execution of the verifier $M$ on input $L\#Y$, and therefore $M(L\#Y)=\texttt{accept}$.

Now, we prove that if there exists a computation walk $W$ for some certificate $Y$ such that the verifier $M$ on input $L\#Y$ reaches an accepting configuration, then \SimulateVerifierForAllCertificates{} constructs a corresponding path in its dynamically constructed computation graph and returns \texttt{Yes}.

Assume, for contradiction, that \SimulateVerifierForAllCertificates{} includes a computation walk in the graph $H$ (constructed via \IsAcceptedOnFootmarks{}) that is not realizable by any valid certificate $Y$ of length $m$ for input $\Lfixed$.

Let $W = (e_0, e_1, \dots, e_k)$ be a partial valid computation walk, and suppose $W' = (e_0, e_1, \dots, e_k, e_{k+1})$ contains an edge $e_{k+1}$ that is not consistent with any certificate. 
The footmark graph generated by the certificate-oblivious TM is grid-aligned by \cref{lem:cotm_to_grid_aligned}.
Since \IsAcceptedOnFootmarks{} only explores edges returned by \GetNonFloorNextEdges{}, and that procedure generates only edges consistent with the verifier’s transition rules, this leads to a contradiction.

Now note that, except for the floor edges of $W$, all other edges are uniquely determined by the \emph{surface} of the computation walk, 
i.e., the sequence of the last transition cases indexed by cell positions as defined in \cref{def:surface}.

By \cref{lem:floor_edge_only_for_valid_input_string}, the returned edge $(v, v')$ by \GetFloorNextEdges{} is consistent with a feasible configuration of the verifier on some input $\Lfixed$ within the certificate length $m$, meaning no invalid edge can be introduced into $H$. 
This contradicts the assumption that $H$ contains a computation walk not realizable by any valid certificate.

Finally, since the algorithm returns \texttt{Yes} only when a node labeled $\qacc \in F$ is reached,
and such a node is reachable only via a valid computation walk corresponding to some certificate $Y$,
the procedure is sound.

Furthermore, every accepting computation walk is expanded per \cref{lem:completeness_of_simulate_verifier_for_all_certificate}, ensuring that it returns \texttt{Yes} if such a walk exists, as \IsAcceptedOnFootmarks{} returns \texttt{True} by \cref{lem:correctness_of_acceptance_in_footmarks}.

\end{proof}
\end{lemma}

\begin{lemma}[Time Complexity of Simulating Verifier for All Certificates]  
\label{lem:time-complexity-simulate}  
Let $H$ be the footmark graph with width $w$ and height $h$ constructed based on the NP computation graph $G$, and let
$T_f$ be the time complexity of \ComputeFeasibleGraph{}.

Then the total time complexity of \SimulateVerifierForAllCertificates{} is bounded by
\[
 O\bigl(w^{4} h^{8} T_f\bigr).
\]
\begin{proof}
The dominant computation in \SimulateVerifierForAllCertificates{} in \cref{alg:compute_verifier_in_polynomial} is the call to \IsAcceptedOnFootmarks{}, which systematically explores and extends verified edges to construct the footmark graph $H$.

By \cref{lem:time-complexity-compute_footmarks}, \IsAcceptedOnFootmarks{} performs at most $\bigO(w^2h^4)$ calls to \VerifyExistenceOfWalk{}, each of which costs $\bigO(w^{2}h^{4} T_f)$ time by \cref{lem:verify_existence_of_walk_time}.

All other steps in \SimulateVerifierForAllCertificates{}, including graph initialization and dynamic edge generation, are polynomially bounded in $w$ and $h$ and are asymptotically dominated by the cost of the main subroutine.

Therefore, the total time complexity is
\[
O\bigl(w^2h^4 \cdot w^{2}h^{4} T_f\bigr)
= O\bigl(w^{4} h^{8} T_f\bigr).
\]
\end{proof}
\end{lemma}

\begin{corollary}[Polynomial-Time Complexity of Simulating Verifier for All Certificates]  
\label{cor:poly_time_simulate}  
Let $M$ be a polynomial-time certificate-oblivious NP verifier that halts in at most $p(n)$ steps on any input of size $n$, and let $G$ be the corresponding NP computation graph. 
Then \SimulateVerifierForAllCertificates{} runs in polynomial time with respect to $n$, bounded by  
\[
O\bigl(p(n)^{20}\bigr).
\]

\begin{proof}  
Let $T_f$ be the time complexity of \ComputeFeasibleGraph{}, and $w$ and $h$ be the width and height of the footmark graph constructed based on $G$, respectively.
By \cref{lem:poly-bounded-graph}, the footmark graph $H$ satisfies $w = \bigO(p(n))$ and $h = \bigO(p(n))$.  

According to \cref{lem:feasible_graph_time_complexity}, the time complexity of \ComputeFeasibleGraph{} is:
\[
T_f = \bigO(w^2 h^4 (h \log h + \log w)) = \bigO(p(n)^6 \cdot p(n) \log p(n)) = \bigO(p(n)^7 \log p(n)).
\]  

From \cref{lem:time-complexity-simulate}, the total time complexity of \SimulateVerifierForAllCertificates{} is:
\[
O(w^4 h^8 T_f) = \bigO(p(n)^{12} \cdot p(n)^7 \log p(n)) = \bigO(p(n)^{19} \log p(n)).
\]  

Since $\bigO(p(n)^{19} \log p(n)) \subset \bigO(p(n)^{20})$, the entire simulation runs in polynomial time in the input size $n$.  
\end{proof}  
\end{corollary}

\subsection{Reduction from NP to P via Feasible Graph Simulation} \label{subsec:reduction_from_np_to_p}

Given any problem in \textsf{NP}, let $M_0$ be its polynomial-time verifier. There exists an oblivious Turing machine $M$, which generate grid-aligned footmark graph, with the same input and polynomial overhead for $M_0$.
We construct a universal computation graph $G$ that encapsulates all potential transitions of $M$ on the problem instance input $\Lfixed$ across the entire space of certificates of length $m$. 
This graph is not explicitly enumerated; rather, it is incrementally explored and extended by generating only those edges verified as feasible through the \textit{walk verification mechanism} established in \cref{sec:walk_verification} and the \textit{feasible graph construction} in \cref{subsec:feasible_graph_construction}.

The algorithm \SimulateVerifierForAllCertificates{} subsequently determines whether an accepting computation path exists within the extended graph. By identifying the \textbf{footmarks} of all valid computation walks, the algorithm effectively collapses the non-deterministic existential search over an exponential number of certificates into a deterministic, boundary-guided traversal of the computation graph. This transformation ensures that the search for a valid certificate is conducted within strictly polynomial time-complexity bounds, thereby completing the reduction from \textsf{NP} to \textsf{P}.

\begin{theorem}[$\mathsf{P} = \mathsf{NP}$]
The algorithm \SimulateVerifierForAllCertificates{} decides any language $\mathcal{L} \in \mathsf{NP}$ in deterministic polynomial time. Specifically, the integration of iterative footmark graph extension and global walk verification ensures that an accepting certificate is identified if and only if one exists. Consequently,  
\[
\mathsf{P} = \mathsf{NP}.
\]
\end{theorem}

\begin{proof}
We prove that \SimulateVerifierForAllCertificates{} solves any $\mathsf{NP}$ problem in polynomial time by leveraging the preceding lemmas and the time complexity analysis. 

Recall that $\mathsf{NP}$ admits a \emph{deterministic polynomial-time verifier}: there exists a verifier $M_0$ such that for every language $\mathcal{L} \in \mathsf{NP}$, membership of an instance $X$ in $\mathcal{L}$ is decided by the verifier $M_0$ on the concatenated input $X \# Y$, where $\#$ denotes a delimiter and $Y$ denotes the certificate. Furthermore, there exists oblivious (especially certificate-oblivious) Turing Machine with the same input, the same halting state and quadratic overhead.
 Therefore, it suffices to show that \SimulateVerifierForAllCertificates{} correctly and efficiently simulates this verifier $M$.

\begin{enumerate}
    \item \textbf{Correctness and Certificate Coverage:}  
    The algorithm systematically simulates the verifier for all possible certificates $Y$ of length $m$ for the problem instance $\Lfixed=X\#$ and $m=|Y|$. 
    By \cref{lem:soundness_of_simulate_verifier_for_all_certificate,lem:completeness_of_simulate_verifier_for_all_certificate}, the simulation is sound and complete: it explores all computation-targeted walks corresponding to every potential concatenated string $X\#Y$ and returns \texttt{Yes} if and only if there exists a certificate $Y$ such that $M(X \# Y)$ accepts. Thus, the algorithm correctly identifies the existence of valid certificates.

    \item \textbf{Polynomial-Time Complexity:}  
    From \cref{cor:poly_time_simulate},  \SimulateVerifierForAllCertificates{} runs within the bound
    \[
    T(n) = \bigO(p(n)^{20}),
    \]
    where $p(n)$ is a polynomial function of the total input size $n = |X \# Y| =  |X| + 1 + m$. Since $m$ is polynomially bounded by $|X|$, $T(n)$ remains a polynomial function of the instance size $|X|$.
    By letting $n' = |X|$, we obtain the final complexity bound:$$T(n') = \bigO(p'(n')^{20}).$$
    Furthermore, $M$ incurs at most quadratic overhead compared to $M_0$, and any RAM algorithm running in polynomial time can be simulated by a deterministic Turing machine with at most cubic overhead. Hence, the algorithm belongs to the class $\mathsf{P}$.

    \item \textbf{Conclusion:}  
    Since \SimulateVerifierForAllCertificates{} decides any language $\mathcal{L} \in \mathsf{NP}$ deterministically in polynomial time, it follows that $\mathsf{NP} \subseteq \mathsf{P}$. Combined with the established inclusion $\mathsf{P} \subseteq \mathsf{NP}$, we establish:
    \[
    \boxed{\mathsf{P} = \mathsf{NP}}.
    \]
\end{enumerate}
\end{proof}
\begin{remark}
It is worth noting that a verifier can be constructed to be certificate-oblivious without external transformation, further simplifying the simulation process.
\end{remark}

\section{Implications and Discussion}

The proof presented in this paper establishes a deterministic polynomial-time simulation of NP verification via feasible graph construction, thereby implying that $P = NP$ within the proposed computational framework.
This result has substantial implications for complexity theory as well as for the broader study of algorithmic computation.

First, while this paper focuses on the theoretical complexity of the reduction, the practical realization of certificate-oblivious Turing machines remains a natural avenue for future investigation. 
Theoretical models for such architectures suggest that deterministic overhead can be effectively mitigated, further reinforcing the viability of the proposed reduction.

Second, from a theoretical standpoint, the result clarifies the relationship between decision problems that admit polynomial-time verification and those that admit polynomial-time deterministic computation.
In particular, it demonstrates that the nondeterministic decision process can be systematically replaced by a structured deterministic exploration of computation walks, without incurring superpolynomial overhead.
This provides a new perspective on the role of nondeterminism in classical complexity theory.

A direct consequence of $P = NP$ is that NP-complete problems, previously regarded as intractable in the worst case, admit polynomial-time deterministic decision procedures.
As a result, the distinction between the complexity classes $\mathsf{P}$ and $\mathsf{NP}$ collapses, yielding
\[
\mathsf{P} = \mathsf{NP} = \mathsf{co\text{-}NP}.
\]
This collapse reshapes the standard hierarchy of time-bounded complexity classes and necessitates a reexamination of several foundational assumptions in computational complexity.

Since the correctness proof is formally established by the existence of a deterministic $\mathsf{NP}$ verifier, the proposed simulation framework is not restricted to any specific $\mathsf{NP}$-complete problem. By the definition of $\mathsf{NP}$, for all $\mathsf{NP}$ problems whose certificates can be represented as a sequence of symbols, the \SimulateVerifierForAllCertificates{} algorithm can be applied directly to their problem-specific verifiers without modification. In this sense, the simulation—functioning through the iterative extension of the footmark graph via computation walk verification with feasible graphs—is not merely a theoretical device tied to a particular $\mathsf{NP}$ verifier, but a general computational framework for $\mathsf{NP}$ verification. This suggests that the proposed construction captures the structural essence of $\mathsf{NP}$ computation, rather than relying on problem-specific encodings, thereby providing a uniform algorithmic perspective across the entire class $\mathsf{NP}$.

Importantly, this result does not imply that all problems become efficiently solvable in practice.
The polynomial-time bounds established in this work arise from a simulation framework involving the construction and traversal of computation graphs whose width and height are polynomially bounded but potentially large.
Consequently, while the algorithm operates in polynomial time in the formal sense, the associated constants and exponents may render direct implementations impractical for large input sizes.

In particular, many cryptographic systems rely on average-case hardness assumptions and specific algebraic structures rather than worst-case NP-hardness alone.
Accordingly, the implications for cryptography are primarily theoretical: existing cryptographic constructions are not immediately invalidated by this result, but their underlying assumptions warrant careful reexamination in light of the equivalence between $\mathsf{P}$ and $\mathsf{NP}$.

Similarly, in areas such as combinatorial optimization, artificial intelligence, and machine learning, NP-Complete formulations frequently arise.
While this work establishes that such problems are decidable in polynomial time in principle, translating the proposed simulation framework into practically efficient algorithms remains an open challenge.
Bridging this gap will require further investigation into optimization strategies, structural restrictions, and potentially new computational paradigms.

It is also important to note that this result does not collapse higher complexity classes.
In particular, the separation between polynomial time and exponential time remains intact: $\mathsf{P} \subsetneq \mathsf{EXPTIME}$ continues to hold.
Thus, the result preserves the broader stratification of time complexity classes while resolving the specific relationship between $\mathsf{P}$ and $\mathsf{NP}$.

Finally, this work opens several directions for future research.
These include refining the feasible graph construction to reduce polynomial overhead, investigating interactions with probabilistic and interactive complexity classes such as $\mathsf{BPP}$ and $\mathsf{IP}$, and exploring whether similar simulation techniques can yield new insights into space-bounded complexity classes such as $\mathsf{PSPACE}$.
More broadly, the feasible graph perspective suggests a unifying structural approach to computation that may inform both theoretical analysis and future algorithmic design.

\section{Conclusion}
In this paper, we introduced a new computation model that enables a deterministic polynomial-time simulation of NP verification, thereby establishing $P = NP$ within the proposed framework.
Rather than attempting to simulate nondeterministic Turing machines directly, our approach focuses on the deterministic simulation of polynomial-time verifiers.
This shift provides a novel perspective on the relationship between nondeterminism and deterministic computation in complexity theory.

Central to our construction are the notions of a computation graph, a feasible graph, and the grid-aligned footmarks of all computation walks. 
Given a verifier $M$ for an $\mathsf{NP}$ language, we construct a computation graph representing all possible transitions over the input and certificate space, 
extract its grid-aligned footmarks subgraph via polynomial-time verification, and deterministically explore all valid computation walks. 
This yields an explicit polynomial-time reduction of the nondeterministic decision process---from existential verification over exponential certificate spaces to a deterministic decision. 

As a consequence, every NP problem admits a deterministic polynomial-time decision procedure, and the long-standing open problem of whether $P = NP$ is resolved within this framework.
Beyond the resolution of this question, the proposed approach provides structural insight into NP computation, revealing how nondeterministic behavior can be captured and simulated through deterministic graph-based exploration.

More broadly, the reduction-based methodology developed in this work suggests a general framework for the deterministic simulation of nondeterministic computation.
This perspective may prove useful in analyzing complexity classes beyond NP and in developing new tools for understanding the fine structure of computational complexity.

Future work will focus on refining the feasible graph construction to reduce polynomial overhead, investigating practical optimizations, and exploring the implications of this framework for related areas such as cryptography, combinatorial optimization, and artificial intelligence.
In addition, extending the model to encompass broader classes of nondeterministic computation may yield further insights into the foundations of complexity theory.

\clearpage
\bibliographystyle{ACM-Reference-Format}
\bibliography{crlee}  

\clearpage
\appendix
\section{Terminology and Definitions}\label{appendix:terminology}
\begin{table}[htbp]
\centering
\caption{Summary of Key Terms and Definitions (Computation Graph)}
\begin{tabularx}{\textwidth}{l|X|l}
\toprule
\textbf{Term} & \textbf{Description} & \textbf{Reference} \\
\midrule 
\textbf{Computation Node} & A 6-tuple $(i, t, q, \sigma, \vdown{q}, \vdown{\sigma})$ representing a cell's local configuration at cell index $i$ and tier $t$. & \cref{subsec:comp_graph} \\
\textbf{Computation Graph} & A directed graph $G=(V, E)$ where vertices are computation nodes and edges represent unit head displacements ($|\Delta \text{index}| = 1$). & \cref{subsec:comp_graph} \\
\textbf{Edge Index/Dir} & $\indexOf(e) = \min(\indexOf(u), \indexOf(v))$ and $\dir(e) = \indexOf(v) - \indexOf(u)$, representing the head's position and movement. & \cref{subsec:comp_graph} \\
\textbf{Edge Slice ($E_i$)} & The formally indexed set of all edges in $G$ sharing the same index $i$. & \cref{subsec:comp_graph} \\
\textbf{Tier} & A hierarchical level of computation nodes and transition cases indicating the number of visiting of the cell. &\cref{subsec:comp_graph}\\
\textbf{Folding Node} & A computation node where incoming and outgoing incident edges share the same cell index. & \cref{subsec:comp_graph} \\
\textbf{Width / Height} & The span of cell indices ($w$) and the maximum tier reached ($h$) in the graph $G$. & \cref{subsec:comp_graph} \\
\textbf{Footmarks $F(\mathcal{W})$} & The subgraph formed by the union of all vertices and edges in a set of computation walks $\mathcal{W}$. & \cref{subsec:comp_graph} \\
\textbf{$e$-augmented Footmarks} & The graph $F(\mathcal{W}) + e$, representing the footmarks expanded by a specific edge $e$ or edge set $E$. & \cref{subsec:comp_graph} \\
\makecell[tl]{\textbf{Dynamic} \ \textbf{Computation Graph}} & A graph constructed incrementally during simulation, adding edges only as they are visited or verified. & \cref{subsec:comp_graph} \\
\textbf{Index-Predecessor} & The last node/edge appearing before the current element on a walk $W$ with the same cell index. & \cref{subsec:comp_graph} \\
\textbf{Index-Successor} & The first node/edge appearing after the current element on a walk $W$ with the same cell index. & \cref{subsec:comp_graph} \\
\textbf{Computation Walk} & A sequence of edges on the computation graph representing a Turing machine execution path. Also referred to as a computation path due to injectivity. & \cref{subsec:comp_graph} \\
\textbf{Transition Case} & A set of computation nodes sharing the same cell index, current state, symbol, and tier. & \cref{subsec:comp_graph} \\
\textbf{Index-Precedent} & A transition case matching the last state and symbol of a computation node (tier difference is +1). & \cref{subsec:comp_graph} \\
\textbf{Index-Succedent} & Set of nodes $v'$ for which a given node $v$ is the index-precedent. & \cref{subsec:comp_graph} \\
\makecell[tl]{\textbf{Previous/Next Edge} \\ \textbf{(Walk-based )}} & For edge $e_i$ in a computation walk, the previous is $e_{i-1}$ and the next is $e_{i+1}$. & \cref{subsec:comp_graph} \\
\makecell[tl]{\textbf{Previous/Next Edges} \\ \textbf{(Graph-based)}} & For edge $e = (u, v)$, all incoming edges to $u$ and outgoing edges from $v$ in a computation graph. & \cref{subsec:comp_graph} \\
\textbf{Surface} & The sequence of latest transition cases for each cell position, as visited by a computation walk. & \cref{subsec:comp_graph} \\
\bottomrule
\end{tabularx}
\label{tab:terms_definitions_computation_graph}
\end{table}

\begin{table}[htbp]
\centering
\caption{Summary of Key Terms and Definitions(Feasible Graph \& Work Verification)}
\begin{tabularx}{\textwidth}{l|X|l}
\toprule
\textbf{Term} & \textbf{Description} & \textbf{Reference} \\
\midrule
\textbf{Floor Edge} & An edge that has no index-predecessor edge, representing the bottom boundary of a computation walk. & \cref{subsec:feasible_graph_concept} \\
\textbf{Ceiling Edge} & An edge that has no index-successor edge, representing the top boundary of a computation walk. & \cref{subsec:feasible_graph_concept} \\
\textbf{Cover Edge} & An edge for which there exists a ceiling-adjacent sequence containing ceiling edges ending with the designated final edge set in the graph. & \cref{subsec:feasible_graph_concept} \\
\textbf{Ex-pendant Edge} & Edge incident to either a source or a sink node. & \cref{subsec:feasible_graph_concept} \\
\textbf{Step-pendant Edge} & Edge that is ex-pendant, or has no index-precedents or no index-succedents in the graph. & \cref{subsec:feasible_graph_concept} \\
\makecell[tl]{\textbf{Step-extended} \\ \textbf{Component}} & A recursively built component consisting of step-pendant edges that are step-adjacent to a given edge set. & \cref{subsec:feasible_graph_concept} \\
\textbf{Index-Adjacent Edge} & An edge adjacent to an edge slice $E_i$ via direct adjacency, folding nodes, initial vertices ($V_0$), or final edges ($E_f$). & \cref{subsec:feasible_graph_concept} \\
\textbf{Feasible Walk} & A computation walk that ends with an edge in the designated final edge set. & \cref{subsec:feasible_graph_concept} \\
\textbf{Feasible Edge} & An edge belonging to at least one feasible walk. & \cref{subsec:feasible_graph_concept} \\
\textbf{Embedded Walk} & A maximal computation walk that is not a feasible walk but consists entirely of feasible edges. & \cref{subsec:feasible_graph_concept} \\
\textbf{Obsolete Walk} & A maximal computation walk that is not a feasible walk and contains at least one non-feasible edge. & \cref{subsec:feasible_graph_concept} \\
\textbf{Obsolete Edge} & An edge that belongs to an obsolete walk but is not a feasible edge. & \cref{subsec:feasible_graph_concept} \\
\textbf{Orphaned Edge} & An edge that does not belong to any valid computation walk (neither feasible nor obsolete). & \cref{subsec:feasible_graph_concept} \\
\textbf{Merging Edge} & Incoming edge at a node with in-degree greater than 1 and out-degree not zero. & \cref{subsec:feasible_graph_concept} \\
\textbf{Splitting Edge} & Outgoing edge at a node with out-degree greater than 1 and in-degree not zero. & \cref{subsec:feasible_graph_concept} \\
\makecell[tl]{\textbf{Computing-} \ \textbf{targeted Walk}} & A valid computation walk  that contains a verification target edge. (Functionally equivalent to a feasible walk). & \cref{sec:walk_verification} \\
\makecell[tl]{\textbf{Computing-} \ \textbf{futile Walk}} & A maximal computation walk that does not reach any verification target edge. & \cref{sec:walk_verification} \\
\makecell[tl]{\textbf{Computing-} \ \textbf{effective Edge}} & An edge that belongs to at least one computing-targeted walk. & \cref{sec:walk_verification} \\
\makecell[tl]{\textbf{Computing-} \ \textbf{redundant Edge}} & A computing-effective edge whose removal does not eliminate the existence of all computing-targeted walks. & \cref{sec:walk_verification} \\
\makecell[tl]{\textbf{Computing-} \ \textbf{futile Edge}} & An edge that belongs to no computing-targeted walk. (Functionally equivalent to non-feasible edges). & \cref{sec:walk_verification} \\
\textbf{$i$-th Pruned Graph} & The graph obtained after prunning a walk $i$ times on original feasible graph. & \cref{subsec:redundant_futile_edge_detection} \\
\textbf{Attempted Walk} & An arbitrary computing-futile walk in the (pruned) feasbile graph. & \cref{subsec:redundant_futile_edge_detection} \\
\makecell[tl]{\textbf{Maximal Pruned}\\ \textbf{Graph}} & The pruned graph at the minimal stage where no further pruning occurs. & \cref{subsec:redundant_futile_edge_detection} \\
\makecell[tl]{\textbf{Critical Attempted} \\ \textbf{Walk}} & The attempted walk that prunes all computing-targeted walks in the computation graph. & \cref{subsec:redundant_futile_edge_detection} \\
\textbf{Boundary Edge} & The outgoing edge from subgraph $H$ to outside of $H$ in computation graph. & \cref{subsec:extending_footmarks} \\
\bottomrule
\end{tabularx}
\label{tab:terms_definitions_feasible_graph}
\end{table}

\clearpage
\section{Computation Walk Simulation Algorithm}
\SetKwFunction{SimulateTuringMachine}{SimulateTuringMachine}
\subsection{Computation Walk Construction Algorithm} \label{subsec:construct_computation_walk}

To provide a formal foundation for simulating the verifier across the entire certificate space, we first describe a general-purpose simulation procedure for a deterministic Turing machine. 
This algorithm models the machine's computation as a walk over the computation graph introduced in \cref{subsec:comp_graph}, explicitly tracking the evolving \textbf{surface}—effectively the \textbf{sequence of the latest transition cases} (precedents) that determine the \textbf{prior state and symbol} for each node.

We consider a standard deterministic Turing machine $M = (Q, \Sigma, \Gamma, \delta, q_0, F)$. Given an initial state $q_0$, an input string $L \in \Sigma^*$, the transition function $\delta$, and a set of halting states $F$, 
the procedure \SimulateTuringMachine{} returns both the final halting state and the resulting computation walk $W$, providing a rigorous trace of the machine's state-symbol evolution.
\SetKwFunction{ProceedNextTransition}{ProceedNextTransition}

\begin{algorithm}
\caption{SimulateTuringMachine($q_0, L, \delta, F$)} \label{alg:simulation_turing_machine}
\Input{Initial State $q_0$, Input string $L$, Transition functions $\delta$, \newline
   \indent  A set of halting states $F$ (including the accept state)}
\Output{The final halting state reached by the machine and the computation walk $W$}
\AlgDescription{\SimulateTuringMachine{} simulates the behavior of the Turing machine and constructs its corresponding computation walk.
The detailed structure and implementation of computation nodes and transition cases are provided in \cref{sec:appendix_computation_graph}.}

\Function{\SimulateTuringMachine($q_0, L, \delta, F$)}{
\StateC{Let $V$ is a 2D dynamic array of 2D array of transition cases containing computation nodes}\Comment{$V[i][t][q][l]$: $V_{i,t}^{q,l}$, See \cref{alg:transition_case} in \cref{sec:appendix_computation_graph}}
\State{Let $s_0 \gets L[0]$ if $L$ is non-empty; otherwise  $s_0 \gets \epsilon$ }
\StateC{Let $v_0 \gets$ the node $v \in V_{0,0}^{q_0,s_0}$} \Comment{unique node at tier-0 transition case}
\StateC{Let $S \gets$ an empty dynamic array of transition cases}	\Comment{$S$: the surface (the last transition cases)}
\StateC{Let $W \gets$ an empty list of edges}	\Comment{$W$ is computation walk}
\State{Let $v \gets v_0$; Let $S[0] \gets V_{0,0}^{q_0,s_0}$}
\While{$state(v) \not\in F$}{
    \State{Set $(v,W) \gets$ \ProceedNextTransition{$v, \delta, L, S, W, V$}}
}
\State{\Return $state(v), W$}
}
\Function{\ProceedNextTransition($v, \delta, L, S, W, V$)}{
    \StateC{Let $(q, l, d) \gets \delta(state(v),symbol(v))$} \Comment{$q$: next state,\ $l$: output symbol,\ $d$: direction ($+1$ or $-1$)}
    \StateC{Set $output(S[j]) \gets l$ where $j=index(v)$} \Comment{Update tape symbol to the output of the transition}
    \StateC{Let $i \gets index(v)+d$}\Comment{$i \gets next\_index(v)$}
    \If{$S[i]$ is defined} {
        \State{Let $s \gets \mathrm{output}(S[i])$}
        \State{Let $T \gets S[i]$}
        \State{Let $t \gets tier(T)+1$}
    }
    \Else{
        \State{Let $s \gets L[i]$ if $0 \le i < |L|$; otherwise $s \gets \epsilon$}
        \State{Let $T \gets \phi$}
        \State{Let $t \gets 0$}
    }
    \StateC{Let $v' \gets$ the node $u \in V_{i,t}^{q,s}$ such that $\IPrec(u)=T$} \label{alg_line:compute_next_node} \Comment{if $T=\phi$, v' is tier-0 node}
    \State{Append $(v, v')$ to the end of $W$}
    \StateC{Set $S[i] \gets V_{i,t}^{q,s}$}\Comment{Update surface at index $i$ to the current transition case}
    \State{\Return $(v', W)$}
}
\end{algorithm}

\begin{lemma}
A walk constructed by \cref{alg:simulation_turing_machine} is a computation walk representing the transition sequence, and the final state returned is the same as the final state of the simulated Turing machine.
\begin{proof}
We prove the lemma by induction on the number $n$ of calls to \ProceedNextTransition{}.

\textbf{Induction hypothesis:}  
Assume that after $(n - 1)$ calls to \ProceedNextTransition{}, the following conditions hold:
\begin{itemize}
    \item The current node $v$ represents the configuration of the Turing machine after $(n - 1)$ transitions;
    \item For every visited tape index $j$, $S[j]$ contains the latest transition case. For all tape indices other than the current one,  $\output(S[j])$ equals the current symbol at cell $j$;
    \item The computation walk $W$ records the correct sequence of $(n - 1)$ transitions using computation nodes with 6-tuple configurations.
\end{itemize}

\textbf{Base case ($n = 0$):}  
Before any transitions, the algorithm starts with node $v_0 \in V_{0,0}^{q_0, s_0}$ where $s_0$ is the initial tape symbol at index 0.  
The computation walk $W$ is empty. The surface $S$ is also empty except that $S[0]$ is set to $V_{0,0}^{q_0, s_0}$ corresponding to the initial tape head position.  
This correctly represents the initial configuration of the Turing machine.  
Hence, the claim holds trivially.

\textbf{Induction step:}  
Now consider the $n$-th call to \ProceedNextTransition{} with the current node $v$.

\begin{enumerate}
    \item The transition function computes $(q, l, d) = \delta(\state(v), \symbol(v))$,  
    where $l$ is the symbol to be written to the current tape cell, and $d \in \{-1, +1\}$ is the head movement direction.

    \item Let $j = \indexOf(v)$ be the current tape position, and update $S[j]$ by setting $\output(S[j]) \gets l$.  
    This models the write operation performed by the Turing machine at tape cell $j$ by the $n$-th transition.

    \item Let $i = j + d$ be the index of the next tape cell to visit.

    \item To determine the contents and transition history at index $i$, the algorithm
	distinguishes the following two cases depending on whether $S[i]$ is defined:
    \begin{itemize}
	    \item \textbf{If $S[i]$ is defined:}  
	    This means that tape cell $i$ has already been visited.  
	    $T := S[i]$ is the most recent transition case at index $i$.  
	    $s := \output(T)$ is the symbol currently stored at cell $i$, and
	    $t := \tier(T) + 1$.
	
	    \item \textbf{If $S[i]$ is undefined:}  
	    This means that tape cell $i$ has not yet been visited.  
	    $s := L[i]$ if $0 \le i < |L|$; otherwise, $s := \epsilon$.  
	    Since no prior transition case exists at index $i$,
	    $T := \phi$ and $t := 0$.
    \end{itemize}
This case distinction ensures that the simulation accurately reflects both the initial tape content and the symbols updated by prior transitions at index $i$.

    \item Let $v' \in V_{i,t}^{q,s}$ be the unique computation node such that $\IPrec(v') = T$.  
    This node $v'$ represents the 6-tuple configuration corresponding to the result of the $n$-th transition.

    \item Append the edge $(v, v')$ to the computation walk $W$.

    \item Set $S[i] \gets V_{i,t}^{q,s}$, recording the new transition case for cell $i$.
\end{enumerate}

By the induction hypothesis and the construction above:
\begin{itemize}
    \item The updating tape content reflects the write operation to cell $j$ by setting $\output(S[j]) \gets l$; thus, for all visited tape cells with index $k \ne i=j+d$, $\output(S[k])$ equals the current symbol at cell $k$.
    \item The next node $v'$ reflects the new state $q$ and symbol $l$ and correctly encodes the last transition history via $\IPrec(v') = T$ on the surface $S$;
    \item The computation walk $W$ now records $n$ transitions using computation nodes;
    \item The surface $S$ is correctly updated to the lastest transition cases for all the visited tape cell.
\end{itemize}
Hence, the induction hypothesis is preserved after the $n$-th call to \ProceedNextTransition{}.

The while loop in \SimulateTuringMachine{} terminates when $\state(v) \in F$, i.e., when the machine reaches a final state.  
Therefore, upon termination, the algorithm returns the final state of the simulated Turing machine and a computation walk $W$ that correctly simulates its behavior.

\end{proof}
\end{lemma}

Even though this algorithm directly uses the transition function $\delta$,
once the surface is constructed, the values $\nextState(T)$, $\output(T)$, and $\nextIndex(T)$ for any transition case $T$ on the surface can be computed without explicitly referring to $\delta$,
provided that the data structure is designed to include such properties.
This is possible because each transition case already encodes the current state and symbol, and the result of the transition is uniquely determined in a deterministic Turing machine.

\SetKwFunction{BruteForceSimulateVerifierForAllCertificates}{BruteForceSimulateVerifierForAllCertificates}
\subsection{Brute-Force Simulation of the Verifier (EXP version)}\label{subsec:simulate_verifier_exp}

This subsection provides a formal procedure to simulate the verifier Turing machine $M$ across the entire space of possible certificates $Y \in \Sigma^m$.
 To decide the language $\mathcal{L} \in \mathsf{NP}$, the simulation must exhaustively enumerate every candidate certificate, explore the resulting execution paths, and determine if $M$ reaches an accepting state for at least one such path. This procedure systematically maps the verifier’s computation onto the computation graph framework, capturing the evolution of configurations across the entire certificate ensemble.
 We first present a naive, brute-force simulation algorithm. Due to the exhaustive branching on every possible certificate symbol, this version inherently operates in exponential time.
  Nevertheless, this algorithm is \textbf{logically complete}—it serves as a rigorous baseline that establishes the correctness of the verification process and provides the conceptual foundation upon which our optimized, polynomial-time simulation is constructed in \cref{subsec:simulate_verifier_poly}.

Let $M = (Q, \Sigma, \Gamma, \delta, q_0, \qacc, \qrej)$ be a deterministic \textsf{NP} verifier machine (CSAT).
\begin{itemize}
    \item $Q$ is a finite set of states; $\Sigma$ is the input alphabet; $\Gamma \supseteq \Sigma \cup \{\epsilon\}$ is the tape alphabet, where $\epsilon$ is a fixed blank symbol.
    \item $\delta : Q \times \Gamma \to Q \times \Gamma \times \{-1,+1\}$ is the transition function.
    \item $q_0 \in Q$ is the initial state.
    \item $\qacc, \qrej \in F \subset Q$ are the accepting and rejecting final states, respectively.
\end{itemize}

We now describe the algorithm \BruteForceSimulateVerifierForAllCertificates{} (\cref{alg:compute_verifier_exp}), which simulates the verifier Turing machine $M$ for all possible certificates of size $m$, given a problem instance string $\Lfixed$. For each certificate, the machine runs in at most $p(n)$ steps, where $n = m + |\Lfixed|$ and $p$ is a polynomial function.

The verifier $M$ takes as input a fixed problem instance $\Lfixed \in \Sigma^*$ followed by a certificate string of length at most $m$, both over $\Sigma$ as illustrated in \cref{fig:verifier_tape_area}.

\begin{algorithm} \small
\SetKwFunction{GetNextNonFloorEdge}{GetNextNonFloorEdge} 
\SetKwFunction{GetNextFloorEdge}{GetNextFloorEdge}
\SetKwFunction{ProceedToNextNodes}{ProceedToNextNodes}
\caption{Brute-Force Simulate Verifier For All Certificates (EXP version)} \label{alg:compute_verifier_exp}
\Input{
$\Lfixed$: problem instance string ending with delimiter \texttt{\#}, $m$: certificate length,\\ 
$q_0$: initial state of verifier TM, $\Sigma$: input alphabet, $\delta$: transition function, \\
$Q$: Set of states of machine, $\qacc$: Accepting state, $\qrej$: Rejecting state
}
\Output{
  The decision result: \texttt{Yes} if any certificate makes verifier accept, otherwise \texttt{No}
}
\Function{\BruteForceSimulateVerifierForAllCertificates{$\Lfixed, m, q_0, \Sigma, \delta, Q, \qacc, \qrej$}} {
\StateC{Let $\Gamma \gets \Sigma \cup \{\epsilon\}$} \Comment{empty symbol is fixed as $\epsilon$}
\State{Let $G$ be a dynamic computation graph constructed during simulation}
\State{Let $S \gets$ an empty dynamic array of transition cases}
\State{Let $W \gets$ an empty list of edges}
\StateC{Let $s_0 \gets \Lfixed[0]$} \Comment{Problem Instance is not empty string}
\StateC{Let $v_0$ be the unique node in $V_{0,0}^{q_0,s_0}$}\Comment{$v_0$:Initial vertex for input symbol $s$}
\If{\ProceedToNextNodes{$G, v_0, S, W, \Gamma, \Lfixed, m, \qacc, \qrej$}} {
        \State{\Return \texttt{Yes}}
}
\State{\Return \texttt{No}}
}
\Function{\ProceedToNextNodes($G, v, S, W, \Gamma, \Lfixed, m, \qacc, \qrej$)}{
    \State{Let $i \gets \indexOf(v)$; Let $i' \gets \nextIndex(v)$}
    \StateC{Set $S[i] \gets$ the transition case containing $v$} \Comment{Update surface $S$ at $i$th element}
    \If{$\state(v) \in \{ \qacc, \qrej \} $} {
        \State{\Return $\state(v) = \qacc$} 
    }
    \ElseIf(\tcc*[f]{Case for the first visit to certificate area}){$S[i']$ is undefined \textbf{and} $|\Lfixed| \le i'<|\Lfixed|+m$}{
        \ForAll{$s' \in \Gamma$}{
            \State{Let $S' \gets$ copy of $S$ and let $W' \gets$ copy of $W$}
            \State{Let $(v,v') \gets$ \GetNextFloorEdge{$V(G), v, s'$}}
            \State{Append $e=(v,v')$ to $W'$ and Add $e$ to $G$}
            \If{\ProceedToNextNodes{$G, v', S', W', \Gamma, \Lfixed, m, \qacc, \qrej$}} {
                \State{\Return \True}
            }
        } 
        \State{\Return \False}
    }
    \ElseIf(\tcc*[f]{Fixed tape area}){$S[i']$ is undefined \textbf{and} \textbf{not} $(|\Lfixed| \le i'<|\Lfixed|+m)$}{
        \StateC{$s' \gets \Lfixed[i']$ if $0 \le i' < |\Lfixed|$ otherwise $s' \gets \epsilon$} \Comment{$|\Lfixed|$: length of $\Lfixed$}
        \State{Let $(v,v') \gets$ \GetNextFloorEdge{$V(G), v, s'$}}
    }
    \Else{ 
        \State{Let $(v,v') \gets$ \GetNextNonFloorEdge{$V(G), v, S$}}
    }
    \State{Append $e=(v,v')$ to $W$ and Add $e$ to $G$}
    \State{\Return \ProceedToNextNodes{$G, v', S, W, \Gamma, \Lfixed, m, \qacc, \qrej$}}
} 
\Function(\tcc*[f]{$V=V(G)$}){\GetNextFloorEdge($V, v, s'$)} {
\State{Let $(i', q') \gets (\nextIndex(v), \nextState(v))$}
\State{Let $v' \gets$ the unique node $u \in V_{i',0}^{q',s'}$}
\State{Return $(v,v')$}
}
\Function(\tcc*[f]{$V=V(G)$}){\GetNextNonFloorEdge($V, v, S$)}{
\State{Let $(i', q') \gets (\nextIndex(v), \nextState(v))$}
\State{Let $P' \gets$ the element(transition case) in surface $S$ with index $i'$}
\State{Let $T' \gets V_{i',t'}^{q', s'}$ with $s'=\output(P'), t'=\tier(P')+1$}
\State{Let $v'$ be the node in $T'$ such that $\lastSymbol(v')=\symbol(P')$, $\lastState(v')=\state(P')$}
\State{Return $(v,v')$}
}
\end{algorithm}
\noindent
The recursive simulation algorithm systematically traverses the layered graph constructed from Turing machine configurations. To distinguish between the inputs, let $L = X\#Y$ be the input string of the verifier Turing machine $M$. The simulation algorithm, however, does not take $L$ as a single input; instead, it receives the fixed problem instance $\Lfixed = X\#$ and the certificate length $m$ as separate parameters to explore all possible certificates $Y \in \Sigma^m$. Tier-0 nodes in the graph directly encode the initial tape contents derived from these inputs, and transitions propagate across tiers as the simulation unfolds.

\begin{lemma}[Correctness of Verifier Simulation(EXP version)]
\Cref{alg:compute_verifier_exp} correctly simulates the behavior of the verifier Turing Machine $M$ for all certificates of length $m$ on the given problem instance $\Lfixed$.
It returns \texttt{Yes} if and only if there exists a certificate that causes
$M$ to accept, and it returns \texttt{No} if and only if no such certificate exists.

Moreover, before termination, the computation graph $G$ contains all configurations (nodes and transitions) that appear in any computation walk of $M$ on all possible certificates of length $m$.

\begin{proof}
Let $G$ be the computation graph constructed by \cref{alg:compute_verifier_exp}.
Let $W = (v_0, v_1, \dots, v_k)$ be a walk generated by the algorithm during the
simulation of the verifier $M$ on input $\Lfixed=X\#$ with some certificate
$Y \in \Sigma^m$.
Then $W$ is a computation walk that represents the transition sequence of $M$
on input $X\#Y$.
Conversely, for any certificate string $Y$, if there exists a computation walk
representing the transition sequence of $M$ on input $X \# Y$, then such a
walk is generated by the algorithm.

\paragraph{Soundness)}  
We show that any computation walk $W = (v_0, \dots, v_k)$ constructed and recorded in $G$ corresponds to a valid execution of $M$ on some input $X \# Y$ for some certificate $Y \in \Sigma^m$.

We proceed by induction on the length $k$ of the computation walk $W$. Note that $\output(T)$ is the output symbol of transition for the $(\state(T), \symbol(T))$ by construction of the dynamic computation graph where $T$ is a transition case on the surface. 
\begin{itemize}
\item \textbf{Base Case ($k=0$):}  

The computation walk consists of a single node $W = (v_0)$.
The algorithm initializes $v_0$ as the unique computation node
\[
v_0 \in V_{0,0}^{q_0, s_0},
\]
where $q_0$ is the initial state of $M$, the tape head index is $0$,
and $s_0$ is the initial tape symbol at index $0$, determined by the
input string $\Lfixed$ (or $\epsilon$ if $\Lfixed$ is empty).

The surface $S$ is initialized such that $S[0]$ stores the transition
case corresponding to $v_0$, and $S[i]$ is undefined for all $i \neq 0$.
No transition has been applied, and thus the computation walk contains
no edges.

This configuration exactly represents the initial configuration of the
verifier $M$ on input $X \# Y$, before any transition is taken where $X$ is an NP instance and $Y$ is its certificate.
Hence, the computation walk of length $0$ is valid.

\item \textbf{Inductive Step:}  
Assume all computation walks of length $\leq k$ recorded in $G$ correspond to valid executions of $M$.

Let $W = (v_0, \dots, v_k, v_{k+1})$ be a walk extended in $G$ by some recursive call. There are four cases depending on whether the surface transition from $v_k$ to $v_{k+1}$ is deterministic or nondeterministic:

Given node $v$ with tape head at position $i = \indexOf(v)$, define $i' = \nextIndex(v)$ as the next tape cell to visit and $q'=\nextState(v)$ as the next state of the Turing Machine. 
Consider recursive call \ProceedToNextNodes{} with depth $k$. 
\begin{itemize}
  \item \textbf{Case 1 (Deterministic transition with first visiting the cell at fixed input region):} If $S[i']$ is \textbf{undefined} but $i'$ lies in the fixed input region or empty symbol region (i.e., $i' < |\Lfixed|$ or $i' \ge |\Lfixed| + m$), the algorithm deterministically sets $s' = \Lfixed[i']$ (or $\epsilon$ for empty cells), computes the unique next node $v'$ via \GetNextFloorEdge{}, and proceeds recursively without branching.
   Thus, node $v_{k+1}$ is the computation node $v'$ for tier $0$, symbol $\Lfixed[i']$(or $\epsilon$ for empty cells) and the current state.

  \item \textbf{Case 2 (Deterministic transition with re-visiting the cell):} If $S[i']$ is \textbf{defined} (the tape symbol at $i'$ is known from prior steps), then the next node $v'$ is deterministically computed via \GetNextNonFloorEdge{} using the transition case $S[i']$ on surface $S$,
 and the simulation proceeds recursively without branching.
   Thus, node $v_{k+1}$ is the computation node $v'$ for symbol $\output(S[i'])$ and the current state with its index-precedent $S[i']$.

  \item \textbf{Case 3 (Nondeterministic certificate branching with first visiting the cell at certificate input region):} If $S[i']$ is \textbf{undefined} and $i'$ lies in the certificate region (i.e., $|\Lfixed| \le i' < |\Lfixed| + m$), the algorithm iterates over all possible symbols $s' \in \Gamma$ for this cell.  
For each such $s'$, \GetNextFloorEdge{} computes the unique next node $v'$. The algorithm then recursively calls \ProceedToNextNodes{} using an updated surface $S'$, which is a copy of $S$ with the element $S'[i']$ set to the transition case of $v'$. Simultaneously, a copy of the computation walk $W'$ is created with the new edge $(v, v')$ appended. This process ensures that each nondeterministic branch is explored independently within its own distinct configuration context.

    Thus, node $v_{k+1}$ is the computation node for all possible certificate symbols and tier $0$ without index-precedent transition cases.

  \item \textbf{Case 4 (Final state):}  
  $v_k$ is a halting configuration with $q_k \in F = \{\qacc, \qrej\}$.  
  No further transitions are made, and $v_k$ becomes the endpoint of the computation walk $W$, returning \True for the accepting state or \False for the rejecting state.  
  If it returns \True, all recursive calls return as well, and when the initial \ProceedToNextNodes{} call returns, it outputs \texttt{Yes}.

\end{itemize}
In all cases, $S[i']$ is updated with the current transition case, ensuring that $\output(S[i'])$ is the correct tape symbol determined by $\state(S[i'])$ and $\symbol(S[i'])$.

Moreover, $v_{k+1}$ is a valid next node of $v_k$ according to $M$'s transition semantics.  
By the inductive hypothesis, the subwalk $(v_0, \dots, v_k)$ is valid, and the surface tracks the latest transition case, preserving both the preceding state and the symbol before overwritten as a current symbol.  
Furthermore, the edge $(v_k, v_{k+1})$ is appended to the computation walk $W$ or its copy $W'$, and also added to the graph $G$.  
Therefore, $W$ is a valid computation walk of length $k+1$.
\end{itemize}
Finally, the algorithm outputs \texttt{Yes} only if there exists a valid computation walk whose final node represents an accepting state.  
This guarantees that any \texttt{Yes} returned is always correct, completing the proof—that is, soundness holds.

\paragraph{Completeness)}  

Suppose, for contradiction, that there exists a valid computation walk
\[
W = (v_0, v_1, \dots, v_k)
\]
of the verifier $M$ on input $X \# Y$ for some certificate
$Y \in \Sigma^m$, such that the simulation algorithm fails to construct
$W$ in the computation graph $G$ generated by
\cref{alg:compute_verifier_exp}.

More precisely, assume that at least one of the following holds:
\begin{enumerate}
  \item The configuration corresponding to $v_k$ admits a valid next
        transition in $M$, but the algorithm does not generate the
        corresponding next computation node.
  \item The edge $(v_{k-1}, v_k)$ belongs to the valid transition
        sequence of $M$ but is not added to the computation graph $G$.
\end{enumerate}
We derive a contradiction in both cases.

First, consider the case where the configuration represented by the final
node $v_k$ is not halting and admits a valid next transition of $M$.

If the transition falls under Case~1 or Case~2 of the soundness proof
(deterministic transitions), then by determinism of the verifier,
there exists a unique next configuration. Since $v_k$ is not a halting
node, the algorithm deterministically computes the corresponding
next computation node and appends it to the current walk, which
contradicts the assumption that $W$ is maximal.

If the transition falls under Case~3 of the soundness proof
(nondeterministic certificate branching), then the algorithm enumerates
all possible symbols in $\Gamma$ for the newly visited certificate cell
and generates a next computation node for each such symbol.
Hence, all valid next nodes of $v_k$ are constructed and added to the
computation graph $G$, again contradicting the maximality of $W$.

Next, consider the case where the edge $(v_{k-1}, v_k)$ exists in the
valid transition sequence of $M$ but is not contained in $G$.

By construction of the algorithm, whenever an edge is appended to a
computation walk (or its copy), the same edge is simultaneously added
to the computation graph $G$. Therefore, it is impossible for an edge
to appear in a valid computation walk without being included in $G$.
This yields a contradiction.

In all cases, we reach a contradiction. Hence, every valid computation
walk of $M$ on input $X \# Y$ is fully generated and recorded in the
computation graph $G$ by the simulation algorithm.

Finally, since the algorithm exhaustively explores all valid computation
walks, if no computation walk ends in an accepting state, then no
certificate leads $M$ to accept. Therefore, when the algorithm returns
\texttt{No}, the answer is correct.

This completes the proof of completeness.
\end{proof}
\end{lemma}
\begin{remark}[Note on the Role of Tier]  
Although this simulation algorithm does not rely on the tier structure for its correctness or completeness, we maintain it as part of the computation graph definition.  
This is because the tier structure becomes crucial in corresponding polynomial algorithm in \cref{subsec:simulate_verifier_poly}, where we design a polynomial-time simulation algorithm that exploits the layered nature of computation to achieve efficiency.  
In particular, tiers will help organize configurations according to the number
of transitions executed at each tape cell, enabling bottom-up reasoning across the graph.
\end{remark}

\section{Detailed Proofs of Structural Lemmas}
\subsection{Computation Model Related Proofs}

\begin{lemma}\label{lem:sink_or_source_is_ex_pendant}
Let $e$ be an edge incident to a source or sink node in the footmark graph. Then $e$ is an ex-pendant edge. (A \emph{source node} has only outgoing edges, and a \emph{sink node} has only incoming edges.)
\end{lemma}

\begin{proof}
Let $e=(u,v)$. We consider each case separately:

\begin{itemize}
    \item \textbf{Case 1: $v$ is a source node.} \\
    If $v$ is a source node, it has only outgoing edges. By the properties of a Deterministic Turing Machine (DTM), all outgoing transitions from a given configuration (state and symbol) must be unique. This implies that all outgoing edges from $v$ share the same direction and, consequently, the same index. Therefore, there exists no edge incident to $v$ with index $i-1$ or $i+1$. Thus, $e$ is either left-pendant or right-pendant, and hence ex-pendant.

    \item \textbf{Case 2: $v$ is a sink node.} \\
    Suppose $v$ is a sink node, meaning it has only incoming edges. Suppose, for contradiction, that $v$ also has an incoming edge $f$ with the opposite direction to $e$. 
    Let $W_e$ and $W_f$ be computation walks containing $e$ and $f$, respectively. By \Cref{lem:property_of_index_predecessor}, the index-predecessors $\ipred_{W_e}(e)=({\vdown{v}}_1, w_1)$ and $\ipred_{W_f}(f)=({\vdown{v}}_2, w_2)$ must also have opposite directions to $e$ and $f$, respectively. 
    This implies that ${\vdown{v}}_1, {\vdown{v}}_2 \in \IPrec(v)$ are incident to edges with different directions. However, this contradicts the property that all nodes in $\IPrec(v)$ must lead to transitions with the same direction, as established in \Cref{def:index-precedent_nodes}. Hence, every edge incident to a sink node must be either left-pendant or right-pendant, and is therefore an ex-pendant edge.
\end{itemize}
\end{proof}

\subsection{Feasible Graph Related Proofs}

\subsubsection{Proof of \Cref{lem:floor_edge_condition}} \label{proof:lem:floor_edge_condition} 
\begin{proof}[Proof of Floor Edge Condition] 
We prove both directions of the equivalence.
\begin{itemize}
\item \textbf{(If direction)} Suppose, for contradiction, that $e=(u,v)$ is not a floor edge.
Then there exists $\vdown{e}$ which is an index-predecessor of $e$ in some computation walk $W$ by definition of floor edge. Let $s$ be the start node of $W$. 
Recall that any computation walk is a path, which means all computation nodes are distinct.
And there also exists node $\vdown{v}$ incident to $\vdown{e}$ with $\indexOf(\vdown{v})=\indexOf(v), \tier(\vdown{v})<\tier(v)$ other than $v$ in the subwalk from $s$ to $u$. 
Since $\tier(v)>\tier(\vdown{v})$, $\tier(v)>0$, even if $\tier(\vdown{v})=0$. This contradicts the assumption that $\tier(v) = 0$.

\item \textbf{(Only if direction)}  
Suppose, for contradiction, that $\tier(v) > 0$.  
We will show that $e = (u, v)$ cannot be a floor edge.

Since $tier(v) > 0$, there exists a node $\vbot{v}$ such that $\indexOf(\vbot{v}) = \indexOf(v)$ and $\tier(\vbot{v}) = 0$; that is, $\vbot{v}$ represents the first visit to the tape cell indexed by $\indexOf(v)$.

Let $s$ and $t$ denote the start and final nodes of the computation walk $W$, respectively.

Then the computation walk $W$ can be decomposed into three subwalks:  
$W_1$ — from $s$ to $\vbot{v}$,  
$W_2$ — from $\vbot{v}$ to $u$,  
$W_3$ — from $u$ to $t$.  

We consider two cases based on the direction of the edge $e$ in $W$:
\begin{itemize}
\item \emph{Case 1: $e$ is directed toward the origin.}  
In $W_2$, there must exist an edge $\vdown{e}$ with $\indexOf(\vdown{e}) = \indexOf(e)$ that precedes $e$ in $W$ due to $|\indexOf(u)|>|\indexOf(\vbot{v})|$, making $\vdown{e}$ an index-predecessor of $e$—contradicting the assumption that $e$ is a floor edge.

\item \emph{Case 2: $e$ is directed away from the origin.}  
In $W_2$, there must again exist an edge $\vdown{e}$ with $\indexOf(\vdown{e}) = \indexOf(e)$ due to $|\indexOf(u)|<|\indexOf(\vbot{v})|$ that appears before $e$ in $W$, serving as an index-predecessor—again contradicting the assumption.
\end{itemize}
In both cases, $e$ must have an index-predecessor with the same index, so it cannot be a floor edge.
\end{itemize}
\begin{figure}[htbp]
  \centering
  \begin{subfigure}[t]{0.45\textwidth}
    \centering
    \includegraphics[width=0.9\textwidth]{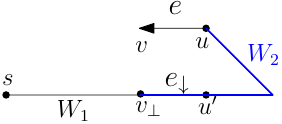}
    \caption{Floor edge case direction toward origin}
    \label{fig:floor_edge_toward_origin}
    \Description{This diagram depicts the case where $e=(u,v)$ is directed away from the origin ($|\indexOf(v)| > |\indexOf(u)|$, moving left-to-right). Similarly, after visiting the floor tier at $\vbot{v}$, the walk $W_2$ proceeds toward a higher index and tier to reach $u$. During this progression, the walk necessarily executes a transition $\vdown{e}$ corresponding to the index of $e$. The presence of this $\vdown{e}$ in $W_2$ proves that $e$ is not the first occurrence of its index in the walk, thus it cannot be a floor edge.}
  \end{subfigure}
  \hfill
  \begin{subfigure}[t]{0.45\textwidth}
    \centering
    \includegraphics[width=0.9\textwidth]{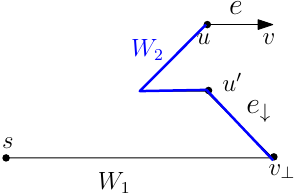}
    \caption{Floor edge case direction away from origin}
    \label{fig:floor_edge_away_from_origin}
    \Description{This diagram depicts the case where $e=(u,v)$ is directed away from the origin ($|\indexOf(v)| > |\indexOf(u)|$, moving left-to-right). Similarly, after visiting the floor tier at $\vbot{v}$, the walk $W_2$ proceeds toward a higher index and tier to reach $u$. During this progression, the walk necessarily executes a transition $\vdown{e}$ corresponding to the index of $e$. The presence of this $\vdown{e}$ in $W_2$ proves that $e$ is not the first occurrence of its index in the walk, thus it cannot be a floor edge.}
  \end{subfigure}
  \caption{Comparison of index-predecessor existence based on edge direction at $\tier > 0$}
  \label{fig:floor_edge_condition}
  \Description{This figure Illustrates why an edge $e=(u,v)$ with $\tier(v) > 0$ cannot be a floor edge by identifying an index-predecessor $\vdown{e}$ within the subwalk $W_2$ (from the floor node $\vbot{v}$ to the source node $u$). In both cases, the computation walk $W$ is partitioned into three segments: $W_1$ (from start $s$ to $\vbot{v}$), $W_2$ (from $\vbot{v}$ to $u$), and $W_3$ (from $u$ to final node $t$).}
\end{figure}
\end{proof}

\subsubsection{Proof of \Cref{lem:ceiling_edge_characterization}} \label{app:proof_ceiling_edge}
\begin{proof}[Characterization of Ceiling Edges]
Let $d = \dir(e_f)$ be the direction of the walk $W$.  

We prove the claim in two parts, depending on the sign of $d \cdot (\indexOf(e_f) - \indexOf(e))$. 
For each case, we first show that if $e$ is a ceiling edge, it must be ceiling-adjacent to another ceiling edge $e_c$ that is closer to $e_f$.
 Given the structural uniqueness—where each index in $W$ admits at most one ceiling edge—this adjacency relation fully characterizes all ceiling edges in the walk, thereby establishing the "only if" direction as well.
\begin{itemize}
\item \textbf{Case 1: $d \cdot (\indexOf(e_f) - \indexOf(e)) \ge 0$.}  

Let $(c_k, c_{k+d}, \dots, c_i, \dots, c_{m-d}, c_m)$ denote the subsequence of $W$, where $\indexOf(c_j) = j, c_k=e$, and $c_m = e_f$ is the final edge of $W$. Suppose that $c_j$ is ceiling-adjacent to $c_{j+d}$ for all $j$ in the following range:
\begin{itemize}
  \item If $d > 0$, then for all $j$ such that $k \le j < m$;
  \item If $d < 0$, then for all $j$ such that $k \ge j > m$.
\end{itemize}

We proceed by induction on $|m - i|$, the number of steps from $c_i$ to $c_m$.
\begin{itemize}
\item \textbf{Base case ($m = i$):}  
Then $c_i = c_m = e_f$. Since $\term(c_m)$ has no outgoing edge in $W$, we have $\ISucc_W(e_f) = \emptyset$, and thus $c_i$ is trivially a ceiling edge.

\item \textbf{Inductive Hypothesis:}  
Assume that the ceiling edge $c_j = c_{i+d}$ satisfies:
\begin{enumerate}
  \item $c_j$ is a ceiling edge;
  \item every edge $f$ in the subwalk $W_j$ from $c_j$ to $c_m$ satisfies $d \cdot (\indexOf(f) - \indexOf(c_j)) \ge 0$.
\end{enumerate}

\item \textbf{Inductive Step:}  
We now show that $c_i$ also satisfies the above conditions.

Let $c_{i+d}$ be the next ceiling edge after $c_i$ in $W$. Since $W$ is a computation walk, $c_i$ is ahead of $c_{i+d}$ in $W$.

Since $c_i$ is ceiling-adjacent to $c_{i+d}$, we consider two cases based on the type of adjacency between $c_i$ and $c_{i+d}$:

\begin{itemize}
  \item[(a)] \textit{ Direct adjacency:}  
  If $c_i$ is directly adjacent to $c_{i+d}$ and both involve non-folding nodes, then their indices differ by $\pm1$, and $c_i$ has no index-successors in $W$.  
  Therefore, $c_i$ is a ceiling edge, and $W_{i+d}$ contains no edge $f$ such that $d \cdot (\indexOf(f) - \indexOf(c_i)) < 0$.

  \item[(b)] \textit{ Indirect adjacency via folding edges:}  
  In this case, $c_i$ is not directly adjacent to $c_{i+d}$, but there exists a sequence of folding nodes from $c_i$ to an edge adjacent to $c_{i+d}$ (as per \cref{def:ceiling_adjacent}).  

  Suppose, for contradiction, that there exists an index-successor edge $(v', \vup{u})$ of $c_i = (u, v)$ in a subwalk from $c_i$ to $c_{i+d}$. Then:\\
  - $\indexOf(u) = \indexOf(\vup{u}) \text{ and } \tier(\vup{u}) = \tier(u) + 1 \quad (\text{due to } u \in \IPrec_H(\vup{u}))$. \\  
This follows because the index-successor of an edge $e$ in any computation walk must belong to the index-succedent of $e$ within the subgraph $H$ (by the definition of index-successor and index-succedent; see \cref{def:index-precedent_index-succedent_edges}).

   However, by the definition of a ceiling-adjacent edge, all nodes with $\indexOf(v)$ in the walk $W$ are folding nodes. Furthermore, since $W$ is a computation walk (viewed as an edge sequence), no other nodes with $\indexOf(v)$ exist in the subwalk of $W$ between $c_i$ and $c_{i+d}$ (and consequently in $H$). This implies there are no edges in $H$ directed toward any node with $\indexOf(u)$ within this specific interval (\cref{lem:folding_node_property}). This contradicts the assumption that an index-successor $(v', \vup{u})$ exists with $\indexOf(\vup{u}) = \indexOf(u)$ and $\tier(\vup{u}) = \tier(u) + 1$.

  Therefore, there is no index-successor in the subwalk $W_{i+d}$, since by the inductive hypothesis, all edges $f \in W_{i+d}$ satisfy $d \cdot (\indexOf(f) - \indexOf(c_{i+d})) \ge 0$.  

  Therefore, $c_i$ is a ceiling edge of $W$, and every edge $f \in W_i$ satisfies $d \cdot (\indexOf(f) - \indexOf(c_i)) \ge 0$.
  Furthermore, since $c_i$ and $c_{i+d}$ lie on the same computation walk $W$, and the inductive hypothesis guarantees that $d \cdot (\indexOf(f) - \indexOf(c_i)) \ge 0$
  for all intermediate edges, the subwalk from $c_i$ to $c_{i+d}$ itself constitutes a path satisfying the connectivity requirement in \cref{def:ceiling_adjacent}.

\end{itemize}

In both cases, $c_i$ is the unique ceiling edge with edge index $i$ satisfying the desired property, and all edges $f \in W_i$ satisfy $d \cdot (\indexOf(f) - \indexOf(c_i)) \ge 0$. 
Since each index admits at most one ceiling edge, any ceiling edge must be ceiling-adjacent to another ceiling edge in the walk, and thus the "only if" condition is satisfied.

Therefore, by induction, the stated properties hold for all indices $i$ in the given range.

\item \textbf{Case 2: $d \cdot (\indexOf(e_f) - \indexOf(e)) < 0$.}
Let $c_{m+d}$ be the edge ceiling-adjacent to $e_f = c_m$.
 Then, by a similar argument as in Case 1 (specifically Case (b) where $d \cdot (\indexOf(e_f) - \indexOf(e)) > 0$), $c_{m+d}$ is the unique ceiling edge with edge index $m + d$.
 Let $W'$ be the subwalk from the start of $W$ up to the edge $c_{m+d}$.
 Apply the same inductive structure as in Case 1, using the reverse direction ($-d$ as the direction).
The same reasoning applies symmetrically, ensuring that the ceiling-edge property is preserved under the reversed direction.
\end{itemize}
\item \textbf{Conclusion:}  
Through this inductive characterization, we have shown that every ceiling edge in $W$ is linked to the final edge $e_f$ through a unique sequence of ceiling-adjacent edges. 
This recursive structure ensures that the ceiling-edge property is uniquely determined for each index along the walk, 
providing a complete characterization of the "ceiling" of $W$ regardless of the walk's direction or index displacement.
\end{itemize}
\end{proof}

\subsubsection{Proof of \Cref{lem:ceiling_edge_characterization}} \label{app:proof_lem_cover_edge}
\begin{proof}[Ceiling Edges are Cover Edges] 
Let $H'$ be a supergraph of a graph $H$, meaning that $H'$ contains all vertices and edges of $H$.
First, consider any graph $H$ and its supergraph $H'$. Then the following properties hold:
\begin{itemize}
    \item If two edges are adjacent in $H$, then they are also adjacent in $H'$.
    \item If a vertex $u$ belongs to $\IPrec_H(v)$, then $u \in \IPrec_{H'}(v)$.
    \item If a vertex $v$ is a folding node in $H$, then $v$ is also a folding node in $H'$, since all outgoing edges from $v$ have the same direction due to the deterministic transition mechanism (DTM).
    \item Moreover, if there exists a path from an edge $f$ to an edge $e$ in $H$ such that no edge on the path shares the same index as $f$ except $f$ itself, then the same path exists in $H'$ with the same index property.
\end{itemize}

It follows that if $f$ is ceiling-adjacent to $e$ in $H$, then $f$ is also ceiling-adjacent to $e$ in any supergraph $H'$ of $H$.

Now let $W = (e_0, \dots, e_k)$ be any walk in $\mathcal{W}_f$, and let $c = e_i$ be a ceiling edge in $W$ for some $i < k$, with $e_k \in E_f$.  

By the definition of a ceiling edge, there exists a subsequence $(e_i, e_{i+1}, \dots, e_k)$ of ceiling edges in $W$ such that for each $j$ with $i \le j < k$, the edge $e_j$ is ceiling-adjacent to $e_{j+1}$.

Since all ceiling-adjacency relations are preserved in $G$, and since $e_k \in E_f \subseteq \widehat{C}$ (by initialization of the cover set), and cover edges are closed under backward ceiling-adjacency from $E_f$, it follows inductively that $c = e_i \in \widehat{C}$.

Therefore, every ceiling edge $c \in C_f$ is also in $\widehat{C}$.
\end{proof}

\subsubsection{Computing Cover Edge Correctness Lemma}
\begin{sublemma}[Correctness of Computing Cover Edges]
\label{sublem:correctness_cover_edges}
Given a computation graph $G$, and a set of designated final edges $E_f$, the set $C$ returned by \ComputeCoverEdges{} in \cref{alg:compute_cover_edges} contains an edge if and only if it is a cover edge in $G$ with respect to $E_f$.
\begin{proof}
Let the final edge $e_f$ be trivially included as the base case at line \ref{alg_line:compute_cover_edges:initial_cover_edge}. Define $c_0 = e_f$.
\begin{itemize}
\item \textbf{(Soundness)}:  
Assume as an induction hypothesis (IH) that the edges $c_0, c_1, \dots, c_k$ have been correctly appended to the set $C$ by the algorithm, and that each $c_{i+1}$ is ceiling-adjacent to $c_i$, and is a cover edge.  
Since all ceiling-adjacent edges are added to $Q$, $c_k$ is also added to $Q$.  
Consider the execution step when $c_k$ is retrieved from $Q$.  
If $c_{k+1}$ is already in $C$, then the proposition holds trivially.  
By the algorithm, every ceiling-adjacent edge $e$ that is either directly adjacent to $c_k$, or is weakly ceiling-adjacent to $c_k$ and admits a connecting path from $c_k$ to $e$ whose intermediate edges all have indices different from that of $c_k$, is added to $C$ at \cref{alg_line:compute_cover_edges:add_connected_cover_edge}, and $e$ becomes $c_{k+1}$ in the ceiling edge chain.
 
This implies that $e$ is a cover edge by the definition of cover edge.  
Thus, by induction, all edges appended by the algorithm form a ceiling-adjacent cover edge chain.  
Hence, the set $C$ constructed by the algorithm contains only valid cover edges and is therefore sound.

\item \textbf{(Completeness)}:  
Suppose there exists an edge $e$ that is a cover edge but is not appended by the algorithm.  
By the definition of cover edge, there exists a cover edge chain $(c_0, c_1, \dots, c_k)$ such that $c_{i+1}$ is ceiling-adjacent to $c_i$ for all $0 \le i < k$.  
Let $e = c_j$ be the first edge in this chain that is not in $C$.  
According to the algorithm, immediately after $c_{j-1}$ is retrieved from $Q$, every ceiling-adjacent edge $e$ to $c_{j-1}$ that is either directly adjacent to $c_{j-1}$, or admits a path from $c_{j-1}$ to $e$ whose intermediate edges all have indices different from $\indexOf(c_{j-1})$, is added to $C$ at \cref{alg_line:compute_cover_edges:add_connected_cover_edge}.

Therefore, $c_j$ must have been added to $C$, leading to a contradiction.  
This implies that no such missing cover edge exists.  
Hence, the algorithm includes all ceiling-adjacent cover edges reachable from $e_f$, and completeness is guaranteed.
\end{itemize}
Therefore, the algorithm correctly constructs the entire set of cover edges via a ceiling-adjacent path starting from $e_f$.
\end{proof}
\end{sublemma}

\subsubsection{Equivalence of Index-Adjacency and Non-Ex-Pendency}

\begin{sublemma}[Equivalence of Bidirectional Index-Adjacency and Horizontal Non-Pendency] \label{sublem:index_adj_equivalence}
An edge $e = (u,v)$ in $G$ with $\indexOf(e)=i$ is \textbf{index-adjacent to both edge slices $E_{i-1}$ and $E_{i+1}$} if and only if it is \textbf{not horizontally $E_f$-step-pendant}.
\end{sublemma}

\begin{proof}
We prove the equivalence by analyzing the structural constraints on $E_{\mathrm{init}}$, $E_f$, and folding nodes. Let $x$ and $y$ be the nodes with indices $i$ and $i+1$ respectively in $\{u, v\}$. Note that if a node is a folding node, it necessarily has incoming and outgoing edges.

\paragraph{Group 1: $e$ is horizontally $E_f$-step-pendant.}
In these cases, $e$ fails to be index-adjacent to at least one direction.
\begin{enumerate}
   \item \textbf{$e$ is left-pendant and $e \notin E_{\mathrm{init}} \cup E_f$}: It is not adjacent to any edge with index $i-1$, and $x$ is not a folding node, so it is not index-adjacent to the edge slice with index $i-1$.
   \item \textbf{$e$ is right-pendant and $e \notin E_{\mathrm{init}} \cup E_f$}: It is not adjacent to any edge with index $i+1$, and $y$ is not a folding node, so it is not index-adjacent to the edge slice with index $i+1$.   
   \item \textbf{$e \in E_{\mathrm{init}} \cap E_f$}: Since $e$ is an initial edge and a final edge, $e$ is index-adjacent to both directions. However, it is horizontally step-pendant by \cref{def:step_pendant_edge}(1).
   \item \textbf{$e$ is both-pendant and $e \in E_{\mathrm{init}} \setminus E_f$}: $e$ is not index-adjacent to the edge slice with index $i+\dir(e)$ because it lacks a right-neighbor and $y$ is not a folding node.
   \item \textbf{$e$ is both-pendant and $e \in E_f \setminus E_{\mathrm{init}}$}: $e$ is not index-adjacent to the edge slice with index $i-\dir(e)$ because it lacks a left-neighbor and $x$ is not a folding node.
\end{enumerate}

\paragraph{Group 2: $e$ is NOT horizontally $E_f$-step-pendant.}
In these cases, $e$ is index-adjacent to both edge slices.
\begin{enumerate}
   \item \textbf{$e$ is not ex-pendant and $e \notin E_{\mathrm{init}} \cup E_f$}: It is left-adjacent to some edge or $x$ is a folding node, and it is right-adjacent to some edge or $y$ is a folding node. Thus, it is index-adjacent to edge slices with index $i-1$ and $i+1$.
   \item \textbf{$e$ is only pendant to the direction $-\dir(e)$ and $e \in E_{\mathrm{init}} \setminus E_f$}: $e$ is index-adjacent to the $(i-\dir(e))$-th edge slice due to incidency to an initial vertex, and it is index-adjacent to the $(i+\dir(e))$-th edge slice due to its adjacency to another edge.
   \item \textbf{$e$ is only pendant to the direction $+\dir(e)$ and $e \in E_f \setminus E_{\mathrm{init}}$}: $e$ is index-adjacent to the $(i+\dir(e))$-th edge slice due to its inclusion in $E_f$, and it is index-adjacent to the $(i-\dir(e))$-th edge slice due to its adjacency to another edge.
\end{enumerate}

Therefore, bidirectional index-adjacency is the structural dual to horizontal non-pendency. An edge satisfies the bidirectional inclusion criteria if and only if it is not horizontally truncated.
\end{proof}

\subsubsection{Proof of \Cref{lem:no_step_pendant_edge_in_feasible_graph}} \label{app:proof_lem_no_step_pendant}
\begin{proof}[Absence of Non-Designated Step-Pendant Edges]
Suppose, for contradiction, that $H$ contains at least one $E_f$-step-pendant edge. 

Let $H^{(j)}$ denote the state of the graph after the $j$-th call to \SweepEdges{}, with the initial state defined as $H^{(0)} = G$. Let $k \geq 1$ be the smallest integer such that an $E_f$-step-pendant edge exists in $H^{(k-1)}$ but is purportedly retained in $H^{(k)}$. Let $e=(u,v)$ be such an edge with $\indexOf(e) = i$.
 For clarity, let $H' = H^{(k)}$. We analyze the cases under which $e$ is classified as a step-pendant edge in $H'$ and show its inclusion leads to a contradiction:

\begin{enumerate}
    \item \textbf{Case 1: $e$ is a horizontally step-pendant edge.} \\
    Suppose $e = (u, v)$ with index $i$ is horizontally step-pendant. By \cref{sublem:index_adj_equivalence}, this is equivalent to $e$ failing to be index-adjacent to at least one adjacent edge slice. We analyze this through the directional sweeps of \StepUpEdges{}:
    \begin{itemize}
        \item If $e$ is \textbf{left-pendant} and is $E_f$-step-pendant (notably, $e \notin E_{\mathrm{init}} \cup E_f$), it is not index-adjacent to the left edge slice $E_{i-1}$. During the forward sweep of \StepUpEdges{} (direction $+1$), the condition at \cref{alg_line:step_up_edges:add_index_adjacent_edges} fails because $e$ lacks the necessary left-hand connectivity or a folding node trigger. Thus, $e \notin I$.
        \item If $e$ is \textbf{right-pendant} and is $E_f$-step-pendant (notably, $e \notin E_{\mathrm{init}} \cup E_f$), it is not index-adjacent to the right edge slice $E_{i+1}$. During the backward sweep of \StepUpEdges{} (direction $-1$), $e$ fails the index-adjacency check for the same structural reason. Thus, $e \notin I$.
    \end{itemize}
    In either sub-case, since only edges already in $I$ can be added to $I'$ during the \StepDownEdges{} phase, the exclusion of $e$ from $I$ ensures that $e \notin H'$. This contradicts the assumption $e \in H$.

    \item \textbf{Case 2: $\IPrec_G(e) = \emptyset$ and $e$ is not a floor edge.} \\
    During the \StepUpEdges{} phase, edges are added to $I$ only if they possess an index-precedent or are designated as floor edges. Since $e$ has no index-precedent and is not a floor edge, it cannot be added to $I$, contradicting the assumption $e \in H'$.

    \item \textbf{Case 3: $\ISucc_G(e) = \emptyset$ and $e$ is not a cover edge (see \cref{def:cover_edges}).} \\
    In this case, $e$ lacks an index-succedent edge and is not a cover edge. During the \StepDownEdges{} phase, only edges with an index-succedent or cover edges are added to $I'$. Consequently, $e$ is not included in $I'$, and therefore $e \notin H'$, again contradicting the assumption.
\end{enumerate}

In all cases, the assumption that an $E_f$-step-pendant edge $e$ is included in $H'$ leads to a contradiction. Therefore, $H$ contains no such edges.
\end{proof}

\subsubsection{Proof of \Cref{lem:only_step_extended_component_removed}} \label{app:proof_lem_only_step_extended}
\begin{proof}[No Pruning beyond Step-Pendant Edges]
Suppose, for contradiction, that there exists an edge $e_0 \in G \setminus H$ that was removed by the algorithm but does not belong to any step-extended component with respect to $E_f$ and $V_0$.

Let $H^{(i)}$ denote the state of the graph $H$ after the $i$-th call to \SweepEdges{}, and let $C^{(i)}$ be the set of removed edges up to that point. There must exist some iteration $k$ such that $e_0 \in H^{(k-1)}$ but $e_0 \notin H^{(k)}$.

First, we observe the conditions under which edges are excluded from $H$.

During \StepUpEdges{}, an edge is added to the intermediate set $I$ if it:
\begin{itemize}
    \item is index-adjacent to some edge already in $I$ (i.e., not ex-pendant), or
    \item has an index-precedent in $H^{(k-1)}$ (unless it is a floor edge), since the algorithm recursively adds all index-succedents starting from floor edges.
\end{itemize}
Thus, only step-pendant edges fail to be added to $I$ in this phase.

We examine the following possible cases for $e_0$ to reach a contradiction:
\begin{description} 
    \item[Case 1:] $e_0$ is not $E_f$-step-pendant in $H^{(k-1)}$. \\
    By the structural duality established in \cref{sublem:index_adj_equivalence}, $e_0$ is not horizontally step-pendant if and only if it is \textbf{bidirectionally index-adjacent} to its neighboring edge slices $E_{i-1}$ and $E_{i+1}$ in $H^{(k-1)}$. Furthermore, since $e_0$ is not vertically step-pendant, it possesses at least one index-precedent and one index-succedent in $H^{(k-1)}$ (unless it is a floor or cover edge).
    In the \StepUpEdges{} phase, since $e_0$ is index-adjacent to the preceding edge slice and has a valid index-precedent in $I$, it satisfies the condition at \cref{alg_line:step_up_edges:add_index_adjacent_edges} and is necessarily added to $I$. Symmetrically, $e_0$ has index-sucedent edges or $e_0$ is cover edge and is added to $I'$ in \StepDownEdges{}.
     This implies $e_0 \in I \cap I' = H^{(k)}$, which directly contradicts the assumption $e_0 \notin H^{(k)}$.

    \item[Case 2:] $e_0$ is $E_f$-step-pendant in $H^{(k-1)}$, but not step-adjacent to any edge in $C^{(k-1)}$. \\
    If $e_0$ is $E_f$-step-pendant in $H^{(k-1)}$ without being step-adjacent to any previously removed edges in $C^{(k-1)}$, its step-pendant status must have been inherent in the original graph $G$. This implies $e_0 \in E_R$. By the definition of step-extended components, $e_0$ is therefore a base edge of a component, contradicting the assumption that $e_0$ belongs to no such component.

    \item[Case 3:] $e_0$ is $E_f$-step-pendant in $H^{(k-1)}$ and is step-adjacent to some edge $e' \in C^{(k-1)}$. \\
    By the inductive hypothesis, $e'$ belongs to a step-extended component. Since $e_0$ is step-adjacent to $e'$ and is $E_f$-step-pendant in $G \setminus C^{(k-1)}$, it satisfies the recursive definition of a step-extended component. This again contradicts the initial assumption.
\end{description}

In all cases, we reach a contradiction. Thus, every edge removed by the algorithm must belong to some step-extended component, completing the proof.
\end{proof}

\subsection{Existence of Ex-Pendant Edge with No Index-Succedent}

\begin{proposition}[Acyclicity and Monotonicity of the Footmark Graph]
\label{lem:monotone_time}
Any subgraph $G$ of a grid-aligned footmark graph with initial node $v_0$ is a Directed Acyclic Graph (DAG), and for any edge $(u, v)$ in $G$, $\timeOf(v) > \timeOf(u)$.
\end{proposition}
\begin{proof}
First, we show that for any arbitrary walk $W_G = (v_0, v_1, \dots)$ in $G$ (where $W$ is not restricted to a computation walk), the property $\timeOf(v_i) = i$ holds for the $i$-th vertex $v_i$.
Suppose, for the sake of contradiction, that this property does not hold. 
Let $v_k$ be the first vertex in a walk such that $\timeOf(v_k) \neq k$. However, by the definition of the grid-aligned property, for any computation walk, $\timeOf(\nextOf(v_{k-1})) = k$. 
Since $G$ is a subgraph of the grid-aligned footmark graph, all vertices $w \in \Next(v_{k-1})$ must satisfy $\timeOf(w) = k$. This contradicts the assumption that $\timeOf(v_k) \neq k$.
From $\timeOf(v_i) = i$, it follows immediately that for any edge $(v_i, v_{i+1})$, $\timeOf(v_{i+1}) = i+1 > i = \timeOf(v_i)$. This strict monotonicity ensures that any path in $G$ progresses to strictly higher time indices. 

Since the $\timeOf$ value is uniquely defined for each node in a grid-aligned footmark graph, no vertex can be revisited, which guarantees that $G$ is a Directed Acyclic Graph (DAG), completing the proof.
\end{proof}

\begin{sublemma}[Ex-Pendant Edge with No Index-Succedent]
\label{lem:ex-pendant_edges_with_no_index-succedent}
In any subgraph $G$ of the grid-aligned footmark graph of $\mathcal{W}$, any edge incident to the node with the maximum time value possesses neither next edges nor index-succedent edges.
\end{sublemma}

\begin{proof}
Let $e = (u, v)$ be an edge such that $\timeOf(v)$ is the maximum among all edges in $G$. 
For the sake of contradiction, suppose $e$ has a next edge $e' = (u', v')$ and an index-succedent edge $e'' = (u'', v'')$ in $G$. 
Since $G$ is a subgraph of the grid-aligned footmark graph $G'$, there exists a computation walk $W_1$ in $G'$ containing some edge $e_1 = (u_1, v_1)$ such that $e' = \nextOf_{W_1}(e_1)$ with $\timeOf(v_1) = \timeOf(v)$, and a computation walk $W_2$ in $G'$ containing some edge $e_2 = (u_2, v_2)$ such that $e'' = \isucc_{W_2}(e_2)$ with $\timeOf(v_2) = \timeOf(v)$.
By the definition of the time value, time values increase strictly monotonically along the walk; thus, $\timeOf(v') > \timeOf(v)$ and $\timeOf(v'') > \timeOf(v)$, which contradicts the assumption that $\timeOf(v)$ is the maximum time value among all edges in $G$. 
Therefore, any edge incident to the node with the maximum time value possesses no next edges or index-succedent edges.
\end{proof}

\subsection{Absence of Infeasible Merging Edge with the First Splitting Edge}
\begin{sublemma}\label{lem:exisitence_of_prior_splittng_edge}
Given a subgraph $G$ of a grid-aligned footmark graph with a unique initial node $v_0$, let $W$ be a computation walk and $e=(u,v)$ be an edge on it.
Let $e'=(u',v')$ be an edge such that $u \neq u'$ but $\indexOf(e)=\indexOf(e')$. If $G$ contains no step-pendant edges, then there exists a splitting edge before $e$ on $W$.
\begin{proof}
Since $G$ is a directed acyclic graph (DAG) rooted at a unique initial node $v_0$, every edge must be reachable from $v_0$ via a valid sequence of transitions. 
Let $e_{\text{min}}=(u_m, v_m)$ be an edge such that $\timeOf(u_m)$ is minimum among all edges backward-reachable from $e'$. 
If $e_{\text{min}}$ is not incident to $v_0$, then $e_{\text{min}}$ lacks its previous edges, rendering it an ex-pendant edge—a type of step-pendant edge—which contradicts the assumption that the graph is free of such edges.
Thus, $e_{\text{min}}$ is incident to $v_0$. The existence of the path $v_0 \rightarrow e'$ implies the existence of a splitting edge $(x,y)$ before $e'$ due to $u \neq u'$. Since $\timeOf(x) < \timeOf(u)$, this splitting edge must exist strictly before $e$ on $W$.
\end{proof}
\end{sublemma}

\begin{corollary}\label{cor:no_merging_edge_with_first_splitting_edge}
Given a subgraph $G$ of a grid-aligned footmark graph with a unique initial node $v_0$, let $W$ be a computation walk and $e=(u,v)$ be the first splitting edge on $W$.
If $G$ has no step-pendant edges, then there is no merging edge $e'$ with the first splitting edge $e$ on $W$.
\begin{proof}
Suppose there exists a merging edge $e'=(u', v)$ with $e=(u, v)$, for the sake of contradiction.
If $u=u'$, then $e=e'$, which is a violation of the merging edge condition. Thus, $u \neq u'$, and there must be a splitting edge before $e$ by \cref{lem:exisitence_of_prior_splittng_edge}, which contradicts the fact that $e$ is the first splitting edge.
\end{proof}
\end{corollary}

\subsection{Absence of Infeasible First Splitting Edge}
\begin{sublemma}[Absence of Infeasible First Splitting Edge] \label{lem:no_infeasible_first_splitting_edge}
Given a subgraph $G$ of a grid-aligned footmark graph with a unique initial node $v_0$, let $e$ be the first splitting edge encountered along any valid computation walk $W$.
If $G$ contains no step-pendant edges, then any splitting edge $e'$ co-originating with $e$ cannot exist in $G$ unless both $e$ and $e'$ are floor edges.
\end{sublemma}

\begin{proof}
If only one of $e=(u, v)$ or $e'=(u, v')$ is a floor edge, it constitutes a direct violation of the grid-aligned property, as $\tier(v) \neq \tier(v')$ while both belong to $\Outgoing(u)$.

Let $W$ be a valid computation walk containing the first splitting edge $e=(u, v)$, and let $e'=(u, v')$ be a distinct edge ($e \neq e'$) branching from $u$. 

We assume, for contradiction, that $e' \in E(G)$.
Since $G$ is free of step-pendant edges, every edge must satisfy historical continuity. 
For the non-floor edge $e'$, there must exist an index-precedent edge $f'=(x, y) \in \IPrec(e')$; otherwise, $e'$ becomes a step-pendant edge.
Note that if $\IPrec(e') = \IPrec(e)$, then $e'$ must be identical to $e$ ($e'=e$).
This follows from the deterministic transition property of the DTM: given the same configuration $(\text{state}, \text{symbol})$ at the end of the precedent edges, the transition function dictates a unique next edge. 
Thus, since we assumed $e \neq e'$, it must be that $\IPrec(e') \neq \IPrec(e)$.
Let $f=\ipred_W(e) \in \IPrec(e)$ and $f' \in \IPrec(e')$ such that $f' \neq f$. By the definition of index-predecessor, $f$ is located before $e$ on $W$.
By \cref{lem:exisitence_of_prior_splittng_edge}, there exists a splitting edge before $f$, which is located strictly before $e$. This contradicts the assumption that $e$ is the first splitting edge.
Therefore, no such non-floor splitting edge $e'$ can exist in $G$, completing the proof.
\end{proof}

\subsection{Index-Precedents as Cover Edges Following Edge Removal}
\begin{sublemma} \label{index-succedent_from_cover_edge}
Let $G' = G + f$ be a directed acyclic footmarks graph containing no step-pendant edges. If $f$ is an index-succedent of $f'$, then $f'$ is a cover edge of $G$.
\end{sublemma}

\begin{proof}
Let $v' = \term(f')$, $v = \init(f)$, and $E_b = \Prev_G(f)$. 

\begin{enumerate}
    \item \textbf{Case 1 ($v = v'$):} 
    If $v = v'$, then $v$ is a folding node by definition. This implies $f' \in E_b$, satisfying the lemma as $f'$ is in the predecessor set of $f$.

    \item \textbf{Case 2 ($v \neq v'$):}
    By the definition of index-succedent, there exists an \emph{ISucc-folding node chain} $(v_0, v_1, \dots, v_n)$ such that $v_0 = v$, $v_n = v'$, where each $v_i$ ($0 < i < n$) is a folding node. 
    Due to the DTM property, this implies an \emph{IPrec-folding node chain} $(v_0, v_1, \dots, v_n)$ where $v_i \in \IPrec(v_{i+1})$. This structural chain ensures that $f'$ is weakly ceiling-adjacent to some edge $f_b \in E_b$.

    To establish $f'$ as a cover edge, consider the reachability toward the final edge set $E_f$. Since $G'$ is a directed acyclic graph, any path starting from $f'$ must terminate. 
     If there exists a path from $f'$ to $f_b$, the ceiling-adjacency condition is satisfied. If such a path does not exist, the directed acyclic nature of $G'$ necessitates that the path from $f'$ must terminate at some sink node. 
    Any such terminal edge in $G'$ that is not in $E_f$ must, by definition, be an ex-pendant edge. However, the absence of step-pendant edges in $G'$ contradicts this. 
   Therefore, $f'$ must possess a path to a final edge in $E_f$, confirming that $f'$ is a cover edge of $G$.
\end{enumerate}
\end{proof}

\subsection{Verification Walk Related Proofs}

\begin{lemma}[Correctness of \FindFirstSplittingEdgeOrFinalEdge{}] \label{lem:correctness_of_find_first_splitting_edge}
Let $G$ be a computation graph containing at least two distinct computation walks, each of which is either computing-targeted or computing-futile to the verification target edge $e_t$.
Then, \FindFirstSplittingEdgeOrFinalEdge($W$) returns the first splitting edge $e$ on $W$ that is not shared with every other computation walk in $G$. If no such splitting edge exists, it returns the final edge of $W$.
\end{lemma}
\begin{proof}The correctness of \FindFirstSplittingEdgeOrFinalEdge{} follows directly from the structural properties of computation walks. 
By \cref{rem:equivalence_of_computation_walk_path}, every computation walk $W$ is strictly monotonic in its tier attribute and thus forms a simple path. 
This ensures that the while loop starting at line \ref{alg_line:find_first_splitting_edge:while_start} performs a finite, linear traversal of $W$ without re-visiting any edge.The algorithm systematically evaluates each edge $e \in W$ in sequence. 
If a splitting edge is encountered during this traversal, the algorithm immediately terminates and returns that edge. If the traversal completes without satisfying the splitting condition, the loop terminates only when $e$ reaches the final edge of $W$, which is then returned by default.
 Since $W$ is a simple path with a well-defined terminal edge, the algorithm is guaranteed to return either the first splitting edge or the final edge of the walk, completing the proof.
 \end{proof}

\begin{sublemma}\label{lem:floor_extended_edge_for_nested_futile_walk}
Given a $e_t$-augmented footmarks $G_U$ and feasible graph $G$, which is a subgraph of $G_U$, if an extendable computing-futile edge is not a floor edge,
then the corresponding computing-futile walk to $e_t$ cannot be a nested computing-futile walk.
Conversely, any computing-futile walk whose extendable computing-futile edge is a floor edge is a nested computing-futile walk in $G$.
\begin{proof}
\begin{figure}
	\centering
	\includegraphics[width=0.5\columnwidth]{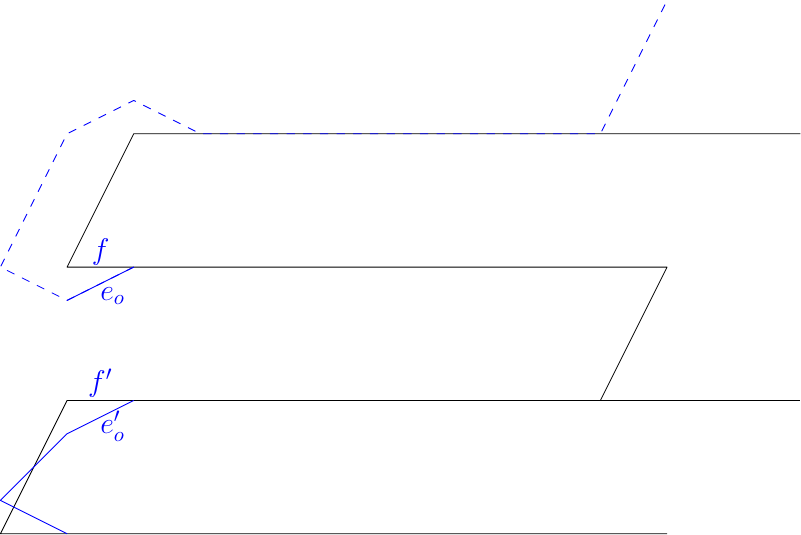}
	\caption{Non-floor extended computing-futile edge}
	\label{fig:non_floor_extended_futile_edge}
	\Description{A walk that is subsumed within a computing-futile walk, branching from a computing-targeted trajectory but terminating prematurely within the graph's observable boundary. A dotted line extends from its terminal vertex, indicating a potential continuation that remains outside the feasible structure and leads to a structural dead end.}
\end{figure}

Let $e_o$ be a non-floor extendable computing-futile edge.

First, suppose that $e_o$ is an extendable computing-futile edge of a nested computing-futile walk
$W$. Then there exists a computing-futile walk $W'$ such that $W$ is a subwalk of $W'$.
Since $e_o$ is not a floor edge, it has an index-predecessor
$\vdown{e_o} = \ipred_W(e_o)$, and there exists a splitting edge $f$ of $e_o$ such that $e_o$
appears in $W'$ but not in $W$.
Consequently, there also exists an index-predecessor
$\vdown{f} = \ipred_{W'}(f)$.

If $\vdown{f} = \vdown{e_o}$, it follows that $f = e_o$ (due to the determinism of the transition), 
which contradicts the assumption that $e_o$ is an extendable computing-futile edge (as $e_o \notin E(G)$ 
but $f$ would imply a valid transition already in $G$). 
If $\vdown{e_o} \neq \vdown{f}$, this contradicts the premise that $W$ is a subwalk of $W'$ sharing the same computation trace up to that point. 
Hence, a computing-futile walk whose extendable computing-futile edge is not a floor edge cannot be a nested computing-futile walk.

If $\vdown{e_o} \neq \vdown{f}$, this contradicts the premise that $W$ is a subwalk of $W'$ sharing a common prefix up to that point.
Hence, a computing-futile walk whose extendable computing-futile edge is not a floor edge cannot
be a nested computing-futile walk.

Second, consider a computing-futile walk $W$ whose extendable computing-futile edge $e_o$ is a floor edge.
Then there must exist another floor edge $f$, splitting from $e_o$, that belongs
to a computing-targeted walk or computing-futile walk $W'$ in the feasible graph $G$; otherwise, the
edge immediately preceding $e_o$ would be a final edge of a computing walk,
preventing $e_o$ from being extendable.
It follows that $W$ is a subwalk of $W'$, and therefore $W$ is a nested computing-futile walk by definition.
\end{proof}
\end{sublemma}

\begin{sublemma}[Correctness of \ExtendFutileWalks{}] \label{lem:add_final_edges_correctness}
Let $G_U$ be $e_t$-augmented footmarks of computation walks $\mathcal{W}$.
Let $G_{\mathrm{in}}$ be the input feasible graph, subgraph of $G_U$, with respect to a set of designated final edges $e_t$, and let $G_{\mathrm{out}}$ be the output graph produced by \ExtendFutileWalks{} in \cref{alg:pruning_an_edge_of_walk} given inputs $G_{\mathrm{in}}$ and $e_t$.  
Let $E_o$ denote the set of edges added by \ExtendFutileWalks{} (see \cref{def:extended_futile_edge}).  
Then, except for nested computing-futile walks, every computing-futile walk has at least one extended computing-futile edge in $G_{\mathrm{out}}$, and each such edge belongs to $E_o$ for verification target edge $e_t$.
\end{sublemma}
\begin{proof}
Suppose, for contradiction, that there exists a computing-futile walk $W \subseteq W' \in \mathcal{W}$ that is not nested in $G_{\mathrm{in}}$,
 yet no edge of $W'$ is added to $E_o$ by \ExtendFutileWalks{}.
 Let $e = (u,v)$ be the first edge of $W' \setminus E(G_{\mathrm{in}})$ such that $u \in V(G_{\mathrm{in}})$.
  By \cref{lem:floor_extended_edge_for_nested_futile_walk}, since $W$ is not a nested computing-futile walk,
  its extendable edge $e$ cannot be a floor edge. Thus, we must have $\text{tier}(v) > 0$ by \cref{lem:floor_edge_condition}.
  Since $W'$ is a valid computation walk in $G_U$, any non-floor edge $e=(u,v)$ must possess an index-precedent, thus $\IPrec_{G_U}(v) \neq \emptyset$. 
  
  Consequently, the edge $e$ satisfies all conditions of the \texttt{if} statement in \ExtendFutileWalks{}:
  \begin{enumerate}
  	\item $e \notin E(G_{\mathrm{in}})$ and $u$ is not a final node;
  	\item $u$ is incident to $V(G_{\mathrm{in}})$;
  	\item $\text{tier}(v) > 0$ and $\IPrec_{G_U}(v) \neq \emptyset$.
  \end{enumerate}
  Thus, $e$ must have been added to $G_{\mathrm{out}}$ and included in $E_o$, which contradicts the assumption that $W$ was not extended. Therefore, every non-nested computing-futile walk has its extendable edge included in $E_o$.
\end{proof}
 
\section{Detailed Algorithm and Time Complexity Analysis}
\subsection{Feasible Graph Related Anaylysis}
\begin{lemma}[Time Complexity of \SweepEdges{} ]
\label{lemma:sweepedges_complexity_final}

Let $G$ be a computation graph of width $w$ (i.e., the number of distinct indices) and height $h$ (i.e., the maximum number of vertices per index layer), and assume that the number of edges per index is $\bigO(h^2)$. Then $G$ contains at most $m = \bigO(wh^2)$ edges.

Then, the worst-case time complexity of the algorithm \SweepEdges{} is:
\[
\bigO\big(wh^2(h\log h + \log w \big)
\]

\begin{proof}

\textbf{StepUpEdges}:  
First, the floor edges of the edge slice are extracted to initialize the queue $Q$. As shown in \cref{sublem:floor_edge_extraction}, this requires $\bigO(h)$ time, as the number of nodes per index in a deterministic transition is bounded by $h$.
The total number of edges processed inside the \texttt{while} loop with queue $Q$ is $\bigO(h^2)$, as visited edges are not revisited. Within each iteration of the loop for a dequeued edge $e$, we perform the following operations:

\begin{itemize}

    \item \textbf{Index-Succedent Retrieval}: We call $\ISucc_G(e)$ to obtain a set $S$ of index-successor edges. According to \cref{sublem:get_edges_complexity}, this operation takes $\bigO(h \log h)$ time and returns at most $\bigO(h)$ edges.
    \item \textbf{Set Minus Operations ($S \setminus E_v$)}: For each of the $\bigO(h)$ retrieved edges, we perform a membership check against $E_v$ and filter only those not in $E_v$. Since $|E_v| = \bigO(h^2)$, each such set operation (check) takes $\bigO(\log h^2) = \bigO(\log h)$ time in an ordered set. Thus, the total cost for processing these $\bigO(h)$ successors is $\bigO(h \log h)$.
    \item \textbf{Index-Adjacency Check}: The condition for index-adjacency to the edge slice $Hs$ (considering boundary sets $V_0$ and $E_f$) is verified in $\bigO(h + \log w)$ time by \cref{sublemma:time_is-index-adjacent}.
    \item \textbf{Addition to Result Set and Visited Set}: Adding the confirmed edge $e$ to the result set $I$ or updating the visited set $E_v$ incurs an overhead of $\bigO(\log h^2) = \bigO(\log h)$ per insertion.
\end{itemize}

Therefore, the total cost is $\bigO(h^2(h \log h + \log w)) = \bigO(h^3 \log h + h^2 \log w)$, as the dominant operations are the set minus operations and the index-adjacency check.

\textbf{StepDownEdges}:  
First, initial set intersection ($C \cap I$) is performed to initialize $Q$, using the global ceiling set $C$ and the index-adjacent set $I$. Since $|I| = \bigO(h^2)$ and $|C| = \bigO(wh)$, performing this intersection via membership checks against $C$ (stored as an ordered set) takes $\bigO(h^2 \log(wh))$ time. The resulting set $C \cap I$ contains $\bigO(h^2)$ edges. Similar to \StepUpEdges{}, the total number of edges processed inside the \texttt{while} loop is $\bigO(h^2)$, as the visited set $E_v$ ensures each edge in the slice is dequeued at most once. Within each iteration for a dequeued edge $e$, the following operations are performed:
 
\begin{itemize}
    \item \textbf{Index-Precedent Retrieval}: We call $\IPrec_G(e)$ to obtain a set $P$ of index-precedent edges. According to \cref{sublem:get_edges_complexity}, this operation takes $\bigO(h \log h)$ time and returns at most $\bigO(h)$ edges. 
    \item \textbf{Set Minus Operations ($P \setminus E_v$)}: For each of the $\bigO(h)$ retrieved edges, we perform a membership check against $E_v$ and filter only those not in $E_v$. Since $|E_v| = \bigO(h^2)$, each such set operation (check) takes $\bigO(\log h^2) = \bigO(\log h)$ time in an ordered set. Thus, the total cost for processing these $\bigO(h)$ predecessors is $\bigO(h \log h)$.
    \item \textbf{Set Updates (Visited and Result Sets)}: Adding the edge $e$ to the result set $I'$ or updating the visited set $E_v$ incurs an overhead of $\bigO(\log h^2) = \bigO(\log h)$ per insertion.
    \item \textbf{Adjacency List Insertion}: Each edge in $I'$ is added to the adjacency structure of the corresponding edge slice $H$. This insertion or union operation takes at most $\bigO(h)$ time per edge to maintain the graph structure.
\end{itemize}
Therefore, the total cost is $\bigO(h^2 \log(wh) + h^2(h \log h + h)) = \bigO(h^3 \log h +h^2 \log w)$, as the dominant operations are the index-precedent retrieval and the associated set operations.

\textbf{SweepEdges (Overall)}:
The algorithm iterates through all $w$ index positions of the computation graph. The total time complexity is derived by aggregating the costs of the core operations across all slices:
\begin{itemize}
    \item \textbf{Slice-wise Traversal}: For each of the $w$ index positions, \StepUpEdges{} and \StepDownEdges{} are called exactly once. As established in the previous analyses, each call incurs a worst-case complexity of $\bigO(h^3 \log h + h^2 \log w)$.
    \item \textbf{Dynamic Graph Expansion}: The insertion of edges into the adjacency structures within \StepDownEdges{} occurs $\bigO(h^2)$ times per slice, totaling $\bigO(wh^2)$ across the entire graph. This is strictly dominated by the traversal and set-operation costs.
    \item \textbf{Total Complexity}: Summing the costs over $w$ iterations, the total time complexity is:
    \[
    \bigO\left( w \cdot \left( h^3 \log h + h^2 \log w \right) \right) = \bigO(wh^3 \log h + wh^2 \log w) = \bigO (wh^2(h\log h + \log w).
    \]
\end{itemize}
This result demonstrates that the algorithm is polynomial with respect to the width and the height of a computation graph. 
\end{proof}
\end{lemma}

\subsection{Verification Walk Related Analysis}
\begin{sublemma}[Time Complexity of \ExtendFutileWalks{}]\label{lem:add_final_edges_time_complexity}
Let the input feasible graph have width $w$, height $h$, and a total of $wh^2$ edges. Assume that all relevant edge and node sets, including $E_o$, $E(G)$, and the index-precedent set $\IPrec_G(v)$, are stored as ordered sets. Then the algorithm \ExtendFutileWalks{} runs in time $\bigO(wh^3)$.

\begin{proof}
The outer loop iterates over at most $wh^2$ edges from $E(G) \setminus E_f$, each of which may potentially be re-added.  
For each edge $(u, v)$:
\begin{itemize}
\item[-] Collecting $(u,v)$ incident to any node in $V(G)$ takes at most $\bigO(wh^2)$ time and there are at most $\bigO(wh^2)$ edges to inspect.
\item[-] Verifying whether $\IPrec_G(v) \ne \emptyset$ or $tier(v) = 0$ takes at most $\bigO(h)$ time.
\item[-] Adding an edge to the graph takes $\bigO(h)$ time, as the adjacency list of each vertex has size at most $\bigO(h)$.
\item[-] Other set operations require $\bigO(\log h)$ time due to ordered set representation.
\end{itemize}
Therefore, the total time complexity is $\bigO(wh^2 \cdot h) = \bigO(wh^3)$.
\end{proof}
\end{sublemma}

\subsection{Proof of P=NP Related Analysis}
\begin{sublemma}[Time Complexity of Collecting Boundary Edges]\label{lem:time_complexity_collecting_boundary_edges}
The \CollectBoundaryEdges{} funtion runs in $T_c=\bigO(wh^3)$ time per invocation,  
and the number of edges it collects is at most $\bigO(wh^2)$ for the computation graph $G$ of width $w$ and height $h$.
\begin{proof}
The total number of edges in $G$ is bounded by $\bigO(wh^2)$.  
The \CollectBoundaryEdges{} function iterates over all nodes and their incident edges. 
Since there are at most $wh$ nodes and each node has at most $h$ incident edges,  
the total number of edges considered is $\bigO(wh^2)$.  
Moreover, for each edge, checking whether it has index-precedent edges requires at most $\bigO(h)$ time.  
Hence, the overall time complexity of \CollectBoundaryEdges{} is $T_c = \bigO(wh^3)$ per invocation.
\end{proof}
\end{sublemma}
 
\begin{corollary}[Time Complexity of Extending Edges by Verification]\label{lem:time_complexity_extending_edges_by_verification}
Let $G$ be a dynamic complete computation graph with width $w$ and height $h$,  
and let $T_v$ denote the time complexity of the \VerifyExistenceOfWalk{} procedure.  

Then, the total time complexity of \ExtendByVerifiableEdges{} per call in \cref{alg:compute_footmarks} is bounded by
\[
\bigO\big(|E(G)| \cdot T_v\big) = \bigO(wh^2 \cdot T_v).
\]
\begin{proof}
The procedure \ExtendByVerifiableEdges{} invokes \VerifyExistenceOfWalk{} on edges in the boundary set $Q$.  
Each edge is processed at most once: once verified, it is added to $E_v$ and removed from $Q$ permanently.  
By \cref{lem:time_complexity_collecting_boundary_edges}, the total number of edges in $G$ is at most $\bigO(wh^2)$.  

Therefore, the total number of calls to \VerifyExistenceOfWalk{} is $\bigO(|E(G)|) = \bigO(wh^2)$,  
yielding the stated time complexity.
\end{proof}
\end{corollary}

\section{Computation Graph Implementation} \label{sec:appendix_computation_graph}

\begin{algorithm}[!ht]
\caption{Computation Node } \label{alg:computain_node}
\AlgDescription{This class represents a computation node}
\SetKwFunction{ComputationNode}{ComputationNode}
\Class{\ComputationNode}  {
\State{\textbf{static} \Field $delta$: the class field for transition function common to all nodes}
\StateC{\Field $index$ : the index of this vertex}\Comment{$\indexOf(v)$}
\StateC{\Field $tier$ : the tier of this  vertex}\Comment{$\tier(v)$} 
\StateC{\Field $state$ : the symbol of this  vertex}\Comment{$\state(v)$} 
\StateC{\Field $symbol$ : the symbols of this vertex}\Comment{$\symbol(v)$}
\StateC{\Field $last\_state$ : the last state of this vertex}\Comment{$\lastState(v)$}
\StateC{\Field $last\_symbol$ : the last symbol of this vertex}\Comment{$\lastSymbol(v)$}
\StateC{\Field $next\_index$ : the next index of this vertex}\Comment{$\nextIndex(v)$}
\StateC{\Field $next\_state$ : the next state of this vertex}\Comment{$\nextState(v)$}
\StateC{\Field $output$ : the output of this vertex}\Comment{$\output(v)$}
\StateC{\Field $dir$: the output of this vertex}\Comment{$\dir(v)$}
\SetKwFunction{ClassInitialize}{ClassInitialize}
\Function{\ClassInitialize{$\delta$}} {
	\State{Set $\delta \gets \delta$}
}
\SetKwFunction{Initialize}{Initialize}
\Function{\Initialize{$i, t, q, s$}} {
	\State{Set $(index,tier,state,symbol) \gets (i,t,q,s)$}
	\State{Set $(q',s',d) \gets \delta(q,s)$}
	\State{Set $(dir, next\_index) \gets (d, d+index)$}
	\State{Set $(next\_state,output) \gets (q',s')$}
}
}

\end{algorithm}

\begin{algorithm}[H]
\caption{Transition Case} \label{alg:transition_case}
\AlgDescription{This class represents a transition case containing computation nodes}
\SetKwFunction{TransitionCase}{TransitionCase}
\Class{\TransitionCase \textbf{extends} \textsf{Set}}{
\State {\textbf{static} \Field $delta$: the class filed for transition function common to all nodes }
\StateC {\Field $index$ : the index of this vertex}\Comment{$index(v)$}
\StateC {\Field $tier$ : the tier of this  vertex}\Comment{$tier(v)$} 
\StateC {\Field $state$ : the symbol of this  vertex}\Comment{$state(v)$} 
\StateC {\Field $symbol$ : the symbols of this vertex}\Comment{$symbol(v)$}
\StateC {\Field $next\_index$ : the next index of this vertex}\Comment{$next\_index(v)$}
\StateC {\Field $next\_state$ : the next index of this vertex}\Comment{$next\_state(v)$}
\StateC {\Field $output$ : the output of this vertex}\Comment{$output(v)$}
\StateC {\Field $dir$: the otput of this vertex}\Comment{$dir(v)$}
\tcp{This contains computation nodes $v$ where $\lastState(v) \in Q$ and $\lastSymbol(v) \in \Gamma$, if tier is $0$, then it contains the unique node}
\Function{\ClassInitialize{$\delta$}} {
	\State{Set $\delta \gets \delta$}
}
\Function{\Initialize{$i, t, q, s$}} {
	\State{Set $(index,tier,state,symbol) \gets (i,t,q,s)$}
	\State{Set $(d, q',s') \gets \delta(q,s)$}
	\State{Set $(dir, next\_index) \gets (d, d+index)$}
	\State{Set $(next\_state,output) \gets (q',s')$}
}
}
\end{algorithm}

\begin{algorithm}[H]
\caption{Adjacent Node Structure for Representing Incident Edges} \label{alg:computation_edge}
\AlgDescription{This class represents the adjacency relations of a node in the Dynamic Computation Graph, partitioned by index direction and incidence.}
\SetKwFunction{AdjacencyList}{AdjacencyList}
\Class{\AdjacencyList} {
    \State{\Field \texttt{v}: The computation node to which edges are incident.}
    \State{\Field \texttt{left\_incoming}: List of nodes $u$ such that there exists an incoming edge $e=(u,v)$ with $index(e)=index(v)-1$}
    \State{\Field \texttt{left\_outgoing}: List of nodes $w$ such that there exists an outgoing edge $e=(v,w)$ with $index(e)=index(v)-1$}
    \State{\Field \texttt{right\_incoming}: List of nodes $u$ such that there exists an incoming edge $e=(u,v)$ with $index(e)=index(v)$}
    \State{\Field \texttt{right\_outgoing}: List of nodes $w$ such that there exists an outgoing edge $e=(v,w)$ with $index(e)=index(v)$}
}
\end{algorithm}

\begin{algorithm}[H]
\caption{Edge Slice Structure and Floor Edge Extraction} \label{alg:edge_slice_and_floor_edges}
\AlgDescription{This class manages a collection of edges with the same index, stored as an adjacency map for efficient access.}
\SetKwFunction{EdgeSlice}{EdgeSlice}
\Class{\EdgeSlice} {
    \State{\Field \texttt{index}: The common edge index of all contained edges.}
    \State{\Field \texttt{edges}: A hash map where keys are computation nodes $v$, and values are \texttt{AdjacencyList} structures containing nodes incident to $v$.}
    
    \SetKwFunction{GetFloorEdges}{GetFloorEdges}
    \Function{\GetFloorEdges{}} {
        \State{$E_b \gets$ an empty list of edges}
        \ForAll(\tcc*[f]{$v$ is a node and $L$ is its \texttt{AdjacencyList}}){$(v, L) \in edges$} {
            \If{$tier(v) = 0$} { 
                \State{Add $(u, v)$ to $E_b$ for each $u \in L.right\_incoming$}
                \State{Add $(u, v)$ to $E_b$ for each $u \in L.left\_incoming$}
            }
        }
        \State{\Return $E_b$}
    }
}
\end{algorithm}

\begin{sublemma}\label{sublem:floor_edge_extraction}
In an \textsc{EdgeSlice} structure, the function \GetFloorEdges{} returns all floor edges in $\bigO(h)$ worst-case time, where $h$ is the maximum tier of the incident nodes.
\end{sublemma} 
\begin{proof}
An \textsc{EdgeSlice} consists of edges $(u, v)$ sharing a common edge index. By the definition of the computation graph $G$, the number of nodes $v$ contained within a single index-specific slice is bounded by the maximum degree of the nodes, which is $\bigO(h)$. 
The function iterates over the hash map of nodes in the slice to identify floor nodes. Since the total number of nodes in a slice is $\bigO(h)$, investigating each node requires $\bigO(h)$ time. Although the total number of edges in a slice could theoretically reach $\bigO(h^2)$, the number of nodes $v$ such that $\text{tier}(v) = 0$ is bounded by a constant, as there are only a fixed number of possible symbol and state configurations at that tier for a given index. Consequently, the subset of edges incident to these constant number of \text{tier}-$0$ nodes is strictly limited by their respective degrees. Given the deterministic transitions of the verifier, the number of such "floor" entries for a specific index remains $\bigO(h)$. Thus, the traversal and collection process is completed in $\bigO(h)$ time.
\end{proof}

\begin{algorithm}[H]
\caption{Dynamic Computation Graph Structure} \label{alg:dynamic_computation_graph}
\AlgDescription{A graph structure that manages computation nodes and edge slices for local consistency verification.}
\SetKwFunction{IsFoldingNode}{IsFoldingNode}
\SetKwFunction{IsAdjacent}{IsAdjacent}
\SetKwFunction{IsMergingEdge}{IsMergingEdge}
\SetKwFunction{DynamicComputationGraph}{DynamicComputationGraph}
\Class{\DynamicComputationGraph} {
    \State{\Field $V$: Multi-dimensional dynamic array of computation nodes $v_{i,t}^{q,s}$.}
    \State{\Field \texttt{edgeSlices}: Dynamic array of \texttt{EdgeSlice} objects, indexed by edge index $i$.}

    \Function{\IsAdjacent{$u,v$}} {
        \State{$i \gets \min(\indexOf(u), \indexOf(v))$}
        \StateC{$L \gets \texttt{edgeSlices}[i].\text{edges}[u]$ \Comment{$L$ is the \texttt{AdjacencyList} for node $u$}}
        \State{\Return $v$ is in ($L.left\_incoming$ \textbf{or} $L.left\_outgoing$ \textbf{or} $L.right\_incoming$ \textbf{or} $L.right\_outgoing$)}
    }

    \Function{\IsFoldingNode{$u$}} {
        \State{$i \gets \min(\indexOf(u), \nextIndex(u))$}
        \State{$L \gets \texttt{edgeSlices}[i].\text{edges}[u]$}
        \If{($L.left\_incoming$ is not empty \textbf{and} $L.left\_outgoing$ is not empty)} {
            \State{\Return \True}
        }
        \ElseIf{($L.right\_incoming$ is not empty \textbf{and} $L.right\_outgoing$ is not empty)} {
            \State{\Return \True}
        }
        \State{\Return \False} 
    }

    \Function{\IsMergingEdge{$e=(u,v)$}} {
        \State{$i \gets \indexOf(e)$}
        \If{$e \notin \text{this graph}$} { \Return \False }
        \StateC{$L_v \gets \texttt{edgeSlices}[i].\text{edges}[v]$ \Comment{\texttt{AdjacencyList} of the terminal node $v$}}
        
        \If{$\text{index}(u) < \text{index}(v)$} {
            \State{\Return $|L_v.left\_incoming| > 1$}
        } \Else {
            \State{\Return $|L_v.right\_incoming| > 1$}
        }
    }
}
\end{algorithm}
\begin{lemma} \label{lem:edge_property_time}
Let $N$ be the number of nodes within a single edge slice in the dynamic computation graph $G$. In a map-based implementation of adjacency lists, the operations \IsAdjacent, \IsFoldingNode, and \IsMergingEdge as defined in \cref{alg:dynamic_computation_graph} each execute in $O(h)$ time in the worst case, where $h$ denotes the maximum tier of the incident nodes.
\end{lemma}

\begin{proof}
Each operation involves querying the adjacency list for a specific node. Since these lists are implemented via hash maps, retrieving or iterating through incident edges requires $O(h)$ time in the worst-case scenario of hash collisions or map traversals.
\end{proof}

\begin{remark}
By \cref{lem:folding_node_property}, edges incident to a folding node share the same index; thus, it suffices to investigate only two of the four available adjacency lists to identify a merging edge. While the current map-based implementation may incur $O(h)$ time in the worst case, the \textbf{amortized time complexity} remains \textbf{constant}. Furthermore, these operations can be strictly implemented in $O(1)$ time by employing a fixed-index dynamic array.
\end{remark}

\section{Primitive Functions}
\begin{algorithm}[H]
\caption{Primitive Graph Operations} \label{alg:primitive_functions}
\AlgDescription{Basic functions used in the construction and analysis of the feasible graph.}
\SetKwFunction{IPrecedent}{IPrecedent}
\SetKwFunction{ISuccedent}{ISuccedent}
\SetKwFunction{IncomingEdges}{IncomingEdge}
\SetKwFunction{OutgoingEdges}{OutgoingEdge}
\SetKwFunction{GetPrevEdges}{GetPrevEdges}
\SetKwFunction{GetNextEdges}{GetNextEdges}
\SetKwFunction{GetIPrecedents}{GetIPrecedents}
\SetKwFunction{GetISuccedents}{GetISuccedents}
 
\Function{\IncomingEdges{$G, v$}} {
    \StateC{$i \gets \indexOf(v)$} \Comment{Combine nodes from the current and the left tape cells}
    \State{$L \gets G.E[i][v].\texttt{right\_incoming} + G.E[i-1][v].\texttt{left\_incoming}$}
    \StateC{\Return $\{ (w, v) \mid w \in L \}$} \Comment{Return as a set of edges}
}

\Function{\OutgoingEdges{$G, v$}} {
    \StateC{$i \gets \indexOf(v)$} \Comment{Combine nodes going to the current and the left tape cells}
    \State{$L \gets G.E[i][v].\texttt{right\_outgoing} + G.E[i-1][v].\texttt{left\_outgoing}$}
    \StateC{\Return $\{ (v, w) \mid w \in L \}$} \Comment{Return as a set of edges}
}

\Function{\GetPrevEdges{$G, e$}} {
    \State{$(u, v) \gets e$}
    \StateC{\Return \IncomingEdges($G, u$)}\Comment{Edges entering the source of $e$}
}

\Function{\GetNextEdges{$G, e$}} {
    \State{$(u, v) \gets e$}
    \StateC{\Return \OutgoingEdges}($G, v$) \Comment{Edges leaving the target of $e$}
}

\Function{\IPrecedent{G,v}} {
	\State{Let $(i, t, q, s) \gets (\indexOf(v), \tier(v)-1, \lastState(v), \lastSymbol(v)) $}
	\State{\Return V[i][t][q][s]}
}
\Function{\ISuccedent{G,v}} {
\State{Let $(i, t) \gets (\indexOf(v), \tier(v)+1)$}
\State{Let $s \gets \output(v)$}
\State{Let $Q \gets$ Set of possible states of computation graph $G$}
\State{Let $S \gets$ an empty set of computation nodes}
\ForAll(\tcc*[f]{Constant size}){$q \in Q$} {
	\ForAll{$w$ in $V[i][t][q][s]$} {
		\If{$(\lastState(w)), \lastSymbol(w))=(\state(v),\symbol(v))$} {
			\State{Add $w$ to $S$}
                }
        }
}
\State{\Return $S$}
}
\end{algorithm}

\begin{sublemma}\label{sublem:primitive_complexity}
The auxiliary functions for graph traversal, specifically \IncomingEdges{} and \OutgoingEdges{}, operate in $\bigO(h \log h)$ time and return $\bigO(h)$ elements. The functions \IPrecedent{} and \ISuccedent{} operate in $\bigO(1)$ time and return $\bigO(1)$ elements, where $h$ is the maximum tier of incident nodes.
\end{sublemma}

\begin{proof}
The time complexity for each function is derived as follows: 
\begin{itemize}
    \item \textbf{\IncomingEdges{} and \OutgoingEdges{}}: These functions aggregate nodes from adjacency lists representing transitions across tape indices $i$ and $i-1$. Since the number of incident edges per node is bounded by $h$, these functions run in $\bigO(h \log h)$ due to the ordered set conversion.
    \item \textbf{\GetPrevEdges{} and \GetNextEdges{}:} As simple wrappers for the above functions, they inherit the $\bigO(h \log h)$ complexity and return $\bigO(h)$ elements.
    \item \textbf{\IPrecedent{}:} This function performs a direct lookup in a multidimensional vertex table using the node attributes. Due to the deterministic nature of the Turing machine, the returned vertex set is of constant size, leading to $\bigO(1)$ complexity.
    \item \textbf{\ISuccedent{}:} This function iterates over the finite set of states $Q$. Because $|Q|$ is a constant determined by the DTM transition function and each lookup in $V[i][t][q][s]$ yields a constant number of nodes, the nested loops perform $\bigO(1)$ operations in total. 
\end{itemize}
Consequently, all functions are computationally efficient within the structural constraints of the DTM.
\end{proof}
 
\begin{algorithm}[!ht]
\caption{Retrieving Index-Precedent and Index-Succedent Edges}
\Function{\GetIPrecedents{$G, e$}} {
    \State{$P \gets$ An empty Ordered Set }
    \State{$(u, v) \gets e$}
    \IfC{Floor edges have no precedents}{$\tier(v) = 0$}{ \Return $\emptyset$}  

    \StateC{Let $V_v \gets \emptyset$} \Comment{Visited vertex set}
    \StateC{Let $V_{ipc} \gets \emptyset$} \Comment{IPrec folding node chain}
    \State{Let $Q$ be a queue initialized with $\IPrec_G(u)$}
    \While{$Q$ is not empty} {
        \State{Dequeue $w$ from $Q$}
        \If{$w \in V_v$} {
            \State{\Continue}
        }
        \If{$w$ is a folding node}	{
			\State{Enqueue all nodes of $\IPrec_G(w)$ }
		}
		\Else {
			\State{Add $w$ to $V_{ipc}$}
		}
	}

    \ForAllC{Constant size of Transition Case}{$v' \in \IPrec_G(v)$} {
        \ForAll{edge $e' = (v', u')$ in $\Outgoing(v')$} {
            \If{$u = u'$ \textbf{or} $u' \in V_{ipc}$} {
                \State{Add $e'$ to $P$} 
            }
        }
    }
    \State{\Return $P$}
}
 
\Function{\GetISuccedents{$G, e$}} {
    \State{$S \gets$ An empty Ordered Set}
    \State{$(u, v) \gets e$}

    \StateC{Let $V_v \gets \emptyset$} \Comment{Visited vertex set}
    \StateC{Let $V_{isc} \gets \emptyset$} \Comment{ISucc folding node chain}
    \State{Let $Q$ be a queue initialized with $ISucc_G(v)$}
    \While{$Q$ is not empty} {
        \State{Dequeue $w$ from $Q$}
        \If{$w \in V_v$} {
            \State{\Continue}
        }
        \If{$w$ is a folding node}	{
			\State{Enqueue all nodes of $ISucc_G(w)$ }
		}
		\Else {
			\State{Add $w$ to $V_{isc}$}
		}
	}
    \ForAllC{Constant size of Transition Case}{$u' \in \ISucc_G(u)$} {
        \ForAll{edge $e' = (v', u')$ in $\Incoming(u')$} {
            \If{ $v = v'$ \textbf{or} $v' \in V_{isc}$} {
                \State{Add $e'$ to $S$} 
            }
        }
    }
    \State{\Return $S$}
}
\end{algorithm}

\begin{sublemma}\label{sublem:get_edges_complexity}
The functions \GetIPrecedents{} and \GetISuccedents{} both operate in $\bigO(h \log h)$ worst-case time and return a set containing $\bigO(h)$ edges, where $h$ is the maximum number of tiers in the computation graph.
\end{sublemma}

\begin{proof} 
We analyze the complexity of \GetIPrecedents{} based on the primitive operations defined in \cref{sublem:primitive_complexity}.

The algorithm first traverses the \emph{IPrec-folding node chain} via BFS. Since each node $w$ in the chain is visited at most once and the \IPrecedent{} operation for each node takes $\bigO(1)$ time, the traversal complexity is dominated by the set operations (e.g., maintaining the \emph{visited} set), which require $\bigO(\log h)$ per node (using ordered set). Thus, the total traversal cost is $\bigO(h \log h)$.

Following this, the algorithm iterates over the set $\IPrec_G(v)$, which contains a constant number of nodes ($\bigO(1)$). For each node $v' \in \IPrec_G(v)$, the algorithm examines up to $h$ outgoing edges, as established in \cref{sublem:primitive_complexity}, identifying $\bigO(h)$ candidate edges in total.

To prevent duplicated edges, the confirmed edges are stored in an ordered set. Employing a comparison-based data structure ensures each insertion requires $\bigO(\log h)$ time. With $\bigO(h)$ such insertions, the edge construction phase requires $\bigO(h \log h)$ time.

Similarly, for \GetISuccedents{}, both the BFS traversal and the subsequent construction of the ordered set of edges are bounded by $\bigO(h \log h)$, given that \ISuccedent{} operations and the inspection of incoming edges also adhere to the same complexity bounds.

Therefore, both functions identify $\bigO(h)$ edges within $\bigO(h \log h)$ total time, confirming their computational efficiency within the established primitive bounds.
\end{proof}

\begin{algorithm}[H]
\caption{GetWeakCeilingAdjacentEdges} \label{alg:weakly_ceiling_adjacent_edges}
\AlgDescription{Returns the set of edges ceiling adjacent to $e_0$ toward $E_f$ in computation graph $G$.}
\SetKwFunction{GetWeakCeilingAdjacentEdges}{GetWeakCeilingAdjacentEdges}
\Function{\GetWeakCeilingAdjacentEdges{$G, e_0, E_f$}} {
    \StateC{Let $C \gets \emptyset$}\Comment{Collected wealky-ceiling-adjacent edges}
    \State{Let $(u_0,v_0) \gets e_0$}
    \StateC{Let $V_v \gets \emptyset$} \Comment{Visited vertex set}
    \State{Let $Q$ be a queue initialized with $u_0$}
    \If{$e_0 \in E_f$} {
        \State{Enqueue $v_0$ to $Q$}
    }
    \While{$Q$ is not empty} {
        \State{Dequeue $u$ from $Q$}
        \If{$u \in V_v$} {
            \State{\Continue}
        }
        \State{Add $u$ to $V_v$}
        \If{$u$ is a folding node \textbf{or} $u=v_0$} {
            \State Let $P \gets \IPrec_G(u)$
            \ForAll{$v \in P$ such that $v \notin V_v$} {
                \State{ Enqueue $v$ to $Q$ }
            }
	}
        \Else {
            \State{Add all incoming edges of $u$ to $C$}
        }
    }
    \State{\Return $C$}
}
\end{algorithm}

\begin{sublemma}[Time Complexity of Weak Ceiling-Adjacent Edge Computation]
\label{sublem:time_complexity_weakly_ceiling_adjacent_edges}
Let $G$ be a computation graph of height $h$ and width $w$, and let $e_0 \in E(G)$ be
a given edge.
Then \GetWeakCeilingAdjacentEdges{} (\Cref{alg:weakly_ceiling_adjacent_edges})
terminates in time $\bigO((h^2+\log w) \log h)$.
\end{sublemma}

\begin{proof}
The algorithm performs a backward traversal over nodes $u$ of the computation graph,
starting from the initial endpoint of $e_0$ and proceeding along backward folding paths.

First, note that the membership test $e_0 \in E_f$ at the initialization step
can be implemented in time $\bigO(\log |E_f|)$ using an ordered set.
Since $|E_f| \le |E(G)| = \bigO(w h^2)$, this check costs
$\bigO(\log (w h^2))$ time.

Since the computation graph has height $h$, it contains at most $\bigO(h)$ nodes.
Each node is enqueued and dequeued at most once due to the visited node set $V_v$.
Therefore, the while-loop iterates at most $\bigO(h)$ times.

For each dequeued node $u$, the algorithm performs one of the following operations:
\begin{itemize}
    \item If $u$ is a folding node or $u = \term(e_0)$, it enumerates
    $\IPrec_G(u)$.
    The number of index-precedent nodes is bounded by a constant determined by the
    tape alphabet and the deterministic transition function, hence
    $|\IPrec_G(u)| = \bigO(1)$.
    Each membership test or insertion into $V_v$ costs $\bigO(\log h)$ time.
    
    \item Otherwise, all incoming edges of $u$ are added to the set $C$.
    Each node has at most $\bigO(h)$ incoming edges.
    For each such edge, membership testing and insertion into $C$ costs
    $\bigO(\log |C|) = \bigO(\log (h^2))$ time since each edge slice have at most $\bigO(h^2)$ edges.
\end{itemize}

Thus, each iteration of the while-loop costs
\[
\bigO(h \cdot \log h).
\]
Since there are at most $\bigO(h)$ iterations, the total running time is
\[
\bigO(h^2 \log h).
\]

Because $\bigO(\log (w h^2)) = \bigO(\log h + \log w)$, the algorithm terminates in time
\[
\bigO(\log h + \log w) + \bigO(h^2 \log h)
= \bigO\bigl((h^2 + \log w)\log h\bigr).
\]

\end{proof}

\begin{algorithm}[H]
\caption{Check Index Adjacency} \label{alg:check_index_adjacency}
\AlgDescription{Returns true if edge $f$ is index-adjacent to $Es$, final edge set $E_f$, and initial vertex set $V_0$.}
\SetKwFunction{IsIndexAdjacent}{IsIndexAdjacent}
\Function{\IsIndexAdjacent{$G, f, Es, E_f, V_0$}} {
    \State{Let $(u,v) \gets f$}
    \StateC{Let $idx \gets$ the index of edge slice $E_i$} \Comment{All edges in $E_i$ share the same edge index}
    
    \If{edge list of $u$ in $E_i$ is not empty \textbf{or} edge list of $v$ in $E_i$ is not empty} {
        \StateC{\Return \True} \Comment{Condition 1: adjacency to edge in $E_i$}
    }

    \If{$idx = \indexOf(f) - \dir(f)$ \textbf{and} ($G$.IsFoldingNode($u$)  \textbf{or} $u \in V_0$)} {
        \StateC{\Return \True} \Comment{Condition 2: outgoing folding edge or initial node}
    }

    \If{$idx = \indexOf(f) + \dir(f)$ \textbf{and} ($G$.IsFoldingNode($v$) \textbf{or} $f \in E_f$} { 
        \StateC{\Return \True} \Comment{Condition 3: incoming folding edge or final edge }
    }
    \State{\Return \False}
}
\end{algorithm}

\begin{sublemma}
\label{sublemma:time_is-index-adjacent}
The function \IsIndexAdjacent() runs in $\bigO(h+\log w)$ time, where $h$
denotes the height of the graph and $w$ its width.

\begin{proof}
The function evaluates the following four conditions:

\begin{itemize}
    \item The first condition iterates over the edge lists of $u$ and $v$ in $E_i$
    and compares each edge with $f$ whose index is $idx$.
    Each edge slice contains at most $4h$ edge lists, and each such list
    contains at most $h$ edges.
    Since the number of edge lists per slice is bounded by a constant factor
    of $h$, the total cost of this step is $\bigO(h)$.

    \item The second and third conditions verify whether the source or target of $f$ satisfies the folding-node property
     alongside the required index relation. These involve linear-time traversals of the incident adjacency lists, 
     thus taking $\bigO(h)$ time each (see \cref{lem:edge_property_time}).

    \item One condition verifies whether $u \in V_0$ and whether
    $idx = \indexOf(f) - \dir(f)$. Assuming constant-time access and membership
    testing, this step also runs in $\bigO(1)$ time.

    \item The final condition checks whether $(u,v) \in E_f$ and whether
    $idx = \indexOf(f) + \dir(f)$. This requires at most $\bigO(\log(hw))$ time.
\end{itemize}

Overall, the total running time is $\bigO(h + \log w)$, dominated by the
$\bigO(h)$ cost of the first condition.
\end{proof}
\end{sublemma}

\end{document}